\pdfoutput=1

\documentclass{article}

\usepackage{PRIMEarxiv}

\usepackage[utf8]{inputenc} 
\usepackage[T1]{fontenc}    
\usepackage{hyperref}       
\usepackage{url}            
\usepackage{booktabs}       
\usepackage{amsfonts}       
\usepackage{nicefrac}       
\usepackage{microtype}      
\usepackage{lipsum}
\usepackage{fancyhdr}       
\usepackage{graphicx}       
\graphicspath{{media/}}     

\usepackage{times}
\usepackage{framed} 
\usepackage{amsthm}
\usepackage{amsmath}
\usepackage{amssymb}
\usepackage{algorithm}
\usepackage{algorithmic}
\usepackage{enumerate}
\usepackage{geometry}
\usepackage{subfigure}
\usepackage{float}
\usepackage{fancyhdr}
\usepackage{booktabs}
\usepackage{color,xcolor}
\usepackage{marvosym}

\usepackage{caption}
\usepackage{multirow}
\usepackage[normalem]{ulem}
\useunder{\uline}{\ul}{}

\makeatletter
\DeclareRobustCommand\TGD{\mathop{\operator@font TGD}\nolimits}
\DeclareRobustCommand\GD{\mathop{\operator@font GD}\nolimits}
\renewcommand{\d}[1]{\ensuremath{\operatorname{d}\!{#1}}}
\DeclareMathOperator\erf{erf}
\makeatother
\newtheorem{theorem}{Theorem}

\newtheorem{corollary}{Corollary}
\newtheorem{definition}{Definition}

\title{A Theory of General Difference in \\Continuous and Discrete Domain}

\author{
  Linmi Tao$^{\ast}$\textsuperscript{\Letter}, Ruiyang Liu$^{\ast}$, Donglai Tao, Wu Xia, Feilong Ma, Yu Cheng, Jingmao Cui \\ \\
  Department of Computer Science and Technology, Tsinghua University \\
  Beijing {\rm 100084}, China\\ \\
  \normalsize{$^\ast$Co-first author.}\\
  \normalsize{\textsuperscript{\Letter}Corresponding author. E-mail: linmi@tsinghua.edu.cn} \\
}

\begin{document}
\maketitle

\begin{abstract}

Though a core element of the digital age, numerical difference algorithms struggle with noise susceptibility. This stems from a key disconnect between the infinitesimal quantities in continuous differentiation and the finite intervals in its discrete counterpart. This disconnect violates the fundamental definition of differentiation (Leibniz and Cauchy). To bridge this gap, we build a novel general difference (Tao General Difference, TGD). Departing from derivative-by-integration, TGD generalizes differentiation to finite intervals in continuous domains through three key constraints. This allows us to calculate the general difference of a sequence in discrete domain via the continuous step function constructed from the sequence. Two construction methods, the rotational construction and the orthogonal construction, are proposed to construct the operators of TGD. The construction TGD operators take same convolution mode in calculation for continuous functions, discrete sequences, and arrays across any dimension. Our analysis with example operations showcases TGD's capability in both continuous and discrete domains, paving the way for accurate and noise-resistant differentiation in the digital era.
\end{abstract}


\newpage
\tableofcontents
\newpage

\section{Introduction}
In the Calculus, the derivative is defined as the limit:
\begin{equation}
  \begin{aligned}
    f^{\prime}(x_0)=\lim _{\Delta x \rightarrow 0} \frac{f(x_0+\Delta x)-f(x_0)}{\Delta x}.
  \end{aligned}
  \label{eq:conventional_derivative}
\end{equation}

Formula~\eqref{eq:conventional_derivative} depicts properties of smooth functions, but fails on non-smooth functions that are commonly observed in real systems. Consequently, various types of generalized derivatives have been proposed to describe the properties of these non-smooth functions~\cite{dini1878fondamenti, borwein2000convex, lanczos1956applied,lanczos1988applied,groetsch1998lanczos}. These generalized derivatives are equivalent to the traditional derivative $f^{\prime}(x_0)$ at each differentiable point $x_0$ of the function $f$. And, for other points that the derivative cannot deal with, the generalized derivatives provide a set of scalars or vectors that depict the change in $f(x)$ at non-differentiable points. Typically, Lanczos proposed “differentiation-by-integration” and introduced Lanczos Derivative (LD) in the following form~\cite{lanczos1956applied,lanczos1988applied,groetsch1998lanczos}: 

\begin{equation}
  \begin{aligned}
    f_{\text{LD}}^{\prime}(x_0) & =\lim _{W \rightarrow 0} \frac{3}{2 W^3} \int_{-W}^W t f(x_0+t) \d{t}\\
                                & = \lim_{W \rightarrow 0} \frac{3}{2 W^3} \left(\int_{0}^W t f(x_0+t) \d{t} - \int_{0}^W t f(x_0-t) \d{t}\right)
  \end{aligned}
  \label{eq:LGD}
\end{equation}

Teruel further extended LD by substituting $t$ with $t^k$ in Formula~\eqref{eq:LGD}, and giving a new class of LD~\cite{teruel2018new}:
\begin{equation}
  \begin{aligned}
    f_{\text{LD}}^{\prime}(x_0; k) & =\lim _{W \rightarrow 0} \frac{k+2}{2 W^{k+2}} \int_{-W}^W t^k f(x_0+t) \d{t} \quad\quad\left\{k=2 a+1 \mid a \in \mathbf{Z}^{+}\right\} \\
    & = \lim _{W \rightarrow 0} \frac{k+2}{2 W^{k+2}} \left(\int_{0}^W t^k f(x_0+t) \d{t} - \int_{0}^W t^k f(x_0-t) \d{t}\right)  \quad\left\{k=2 a+1 \mid a \in \mathbf{Z}^{+}\right\}
  \end{aligned}
  \label{eq:LGD-k}
\end{equation}

Hicks proposed that the approximation of LD can be seen as averaging the derivative over a small marginal interval $[x_0 - W, x_0 + W]$ around the point of differentiation $x_0$~\cite{hicks2000lanczos}. This average might be understood as a weighted average with a weight function $w$:
\begin{equation}
  \begin{aligned}
   f^{\prime}\left(x_0\right) \approx \lim _{W \rightarrow 0} \int_{-W}^{W} w(t) f^{\prime}(x_0+t) \d{t} \\
    \text { where } \int_{-W}^{W} w(t) \d{t}=1 \text { and } w(t) > 0.
    \end{aligned}
  \label{eq:LGD-Hicks}
\end{equation}
Following Hicks, Liptaj presented some interesting non-analytic weight functions for LD~\cite{liptaj2019maximal}.

From Leibniz to Lanczos, then Liptaj, the evolution of the Calculus is evident in its progressive development and perfection over the centuries. Nonetheless, this perfect theory is unsuitable for modern real systems, which are discrete. In discrete systems, numerical computation occurs within a finite interval, which renders Calculus's foundations of infinitesimal and limit unsuitable. To address this issue, this paper transcends the bound of infinitesimal and limit to the theory of general difference in a finite interval for both continuous and discrete domain. We believe the theory propels the Calculus toward modern numerical computation for sequences and multidimensional arrays in the real world.

This paper is organized in four sections: Firstly, we deduce the general difference at a finite interval in continuous domain. Secondly, the general difference in discrete domain is built for the difference calculation of sequences and arrays. Thirdly, examples of general difference operators and analysis are provided for further understanding. Finally, a short conclusion and outlook are given. The practical applications of the general difference are presented in the paper Applications.

\clearpage
\newpage

\section{General Difference in Continuous Domain: From Infinitesimal to Finite Interval}
\subsection{Tao Derivative}

Both LD and Teruel's extended LD are based on the function $t^k$. We further extend LD by a family of function $w(t)$, and define a new generalized derivative, \textbf{T}ao \textbf{D}erivative (TD), of any Lebesgue integrable function $f$ as:
\begin{equation}
  \begin{aligned}
    f_{\text{TD}}^{\prime}(x_0; w) \triangleq \lim _{W \rightarrow 0} \frac{\int_{0}^{W} f(x_0+t) w(t) \d{t}-\int_{0}^{W} f(x_0-t) w(t) \d{t}}{2\int_{0}^{W}t w(t)\d{t}}.
  \end{aligned}
  \label{eq:TD}
\end{equation}

where $W > 0$ and $w(t)$ satisfies the following \textbf{Normalization Constraint (C1)}:
\begin{equation}
  \begin{aligned}
    \left\{\begin{array}{ll}
        \int_{0}^{W}w(t) \d{t}=1 \\
        w(t)>0 & \quad t \in (0,W]\\
        w(t)=0 & \quad t \in (-\infty,0] \bigcup (W,+\infty)
        \end{array}\right.
  \end{aligned}
  \label{eq:constraintw}
\end{equation}

Noticing that Formula~\eqref{eq:LGD}~and~\eqref{eq:LGD-k} can be obtained by specifying $w(t)$ with $2t/W^2$ and $(k+1)t^k/W^{k+1}$, respectively. Thus, TD defined by Formula~\eqref{eq:TD} is the general form of LD.
\begin{equation}
  \begin{aligned}
    f_{\text{TD}}^{\prime}(x_0; \frac{2t}{W^2}) &= f_{\text{LD}}^{\prime}(x_0) \\
    f_{\text{TD}}^{\prime}(x_0; \frac{(k+1)t^k}{W^{k+1}}) &= f_{\text{LD}}^{\prime}(x_0;k)\quad\left\{k=2 a+1 \mid a \in \mathbf{Z}^{+}\right\}
  \end{aligned}
  \label{eq:LD_TD}
\end{equation}

It is proved that, at each differentiable point $x_0$ of $f$, the TD result of Formula~\eqref{eq:TD} is equivalent to the traditional derivative specified in Formula~\eqref{eq:conventional_derivative}.

\begin{framed}
\begin{theorem}\label{theo:t1}
\textbf{Suppose that $f$ is a $C^1$ continuous function defined on a neighborhood of point $x_0$, and $w(t)$ is a function satisfying Normalization Constraint (Formula~\eqref{eq:constraintw}) with a parameter $W > 0$. $f'_{\text{TD}}(x_0; w)$ given by Formula~\eqref{eq:TD} equals to $f'(x_0)$ given by Formula~\eqref{eq:conventional_derivative}.}
\end{theorem}

\begin{proof}
\begin{equation*}
  \begin{aligned}
    f_{\text{TD}}^{\prime}(x_0; w)=&\lim _{W \rightarrow 0} \frac{\int_{0}^{W} f(x_0+t) w(t) \d{t}-\int_{0}^{W} f(x_0-t) w(t) \d{t}}{2\int_{0}^{W}t w(t)\d{t}} \\
    =& \lim _{W \rightarrow 0} \frac{\int_{0}^{W} (f(x_0+t)-f(x_0-t)) w(t) \d{t}}{2\int_{0}^{W}t w(t)\d{t}} \quad\quad(\text{Taylor's Formula}) \\
    =& \lim _{W \rightarrow 0} \frac{\int_{0}^{W} ((f(x_0)+t f^{\prime}(x_0)+O(t^2))-(f(x_0)-t f^{\prime}(x_0)+O(t^2))) w(t) \d{t}}{2\int_{0}^{W}t w(t)\d{t}} \\
    =& \lim _{W \rightarrow 0} \frac{\int_{0}^{W} (2t f^{\prime}(x_0)) w(t) \d{t}}{2\int_{0}^{W}t w(t)\d{t}} \\
    =& f^{\prime}(x_0) \lim _{W \rightarrow 0} \frac{2\int_{0}^{W} t w(t) \d{t}}{2\int_{0}^{W}t w(t)\d{t}} \\
    =& f^{\prime}(x_0).
  \end{aligned}
\end{equation*}
\end{proof}
\end{framed}

The Leibniz Rule of differentiation is applicable to TD as well. Specifically, the following theorem holds:

\begin{framed}
\begin{theorem}\label{theo:t2}
\textbf{Suppose that $f$ and $g$ are both $C^1$ continuous functions defined on a neighborhood of point $x_0$, and $w(t)$ is a function satisfying Normalization Constraint (Formula~\eqref{eq:constraintw}) with a parameter $W > 0$. We have $(f \cdot g)'_{\text{TD}}(x_0; w) = f'(x_0) g(x_0) + g'(x_0) f(x_0)$, or in short, $(f \cdot g)'_{\text{TD}} = f'g + g'f$.}
\end{theorem}
\begin{proof}
According to Taylor's Formula, for $f$ and $g$ on a neighborhood of $x_0$, $t \in [0, W)$ and $W \rightarrow 0$, we have:
\begin{equation}
\begin{aligned}
   f(x_0 + t) &= f(x_0) + t f'(x_0) + O(t^2)\\
   g(x_0 + t) &= g(x_0) + t g'(x_0) + O(t^2).
\end{aligned}
\end{equation}

Then:
\begin{equation}
    f(x_0+t)g(x_0+t) - f(x_0-t)g(x_0-t) = 2t (f'g + g'f) + O(t^2).
    \label{eq:fggf}
\end{equation}

Finally, we can apply Formula~\eqref{eq:fggf} to the definition of $(f \cdot g)'_{\text{TD}}$:
\begin{equation*}
  \begin{aligned}
    (f \cdot g)_{\text{TD}}^{\prime}=&\lim _{W \rightarrow 0} \frac{\int_{0}^{W} f(x_0+t) g(x_0+t) w(t) \d{t}-\int_{0}^{W} f(x_0-t) g(x_0-t) w(t) \d{t}}{2\int_{0}^{W}t w(t)\d{t}} \\
    =& \lim_{W \rightarrow 0} \frac{\int_{0}^W \left(f(x_0+t) g(x_0+t) - f(x_0-t) g(x_0-t)\right) w(t) \d{t}}{2\int_{0}^{W}t w(t)\d{t}} \\
  \end{aligned}
\end{equation*}
\begin{equation*}
  \begin{aligned}
    =& \lim_{W \rightarrow 0} \frac{\left(2\int_{0}^{W} tw(t) \d{t}\right)(f'g+g'f) + O\left(\int_{0}^{W}t^3 w(t)\d{t}\right)}{2\int_{0}^{W}t w(t)\d{t}} \\ 
    =& f^{\prime}g + g^{\prime}f.
  \end{aligned}
\end{equation*}
\end{proof}
\end{framed}

Combining Theorem~\ref{theo:t1} and~\ref{theo:t2}, we can immediately get the following corollary:

\begin{framed}
\begin{corollary}
\textbf{Suppose that $f$ and $g$ are both $C^1$ continuous functions defined on a neighborhood of point $x_0$, and $w(t)$ is a function satisfying Normalization Constraint (Formula~\eqref{eq:constraintw}) with a parameter $W > 0$. We have $(f \cdot g)'_{\text{TD}}(x_0; w) = f'_{\text{TD}}(x_0; w) g(x_0) + g'_{\text{TD}}(x_0; w) f(x_0)$, or in short, $(f \cdot g)'_{\text{TD}} = f'_{\text{TD}} g + g'_{\text{TD}} f$.}
\end{corollary}
\end{framed}

Similarly, we express our second-order Tao Derivative of a Lebesgue integrable function $f$ as follows:
\begin{equation}
    \begin{aligned}
      f^{\prime\prime}_{\text{TD}}(x_0; w) \triangleq &\lim_{W \rightarrow 0} \frac{\int_0^W f(x_0+t) w(t) \d{t} + \int_0^W f(x_0-t) w(t)\d{t} - 2f(x_0)\int_{0}^{W}w(t) \d{t}}{\int_0^W t^2 w(t)\d{t}}
    \end{aligned}
    \label{eq:2ndtd}
\end{equation}
where $w(t)$ is a continuous function with a parameter $W > 0$, satisfying the Normalization Constraint given by Formula~\eqref{eq:constraintw}.

It can be further proved that, at each differentiable point $x_0$ of $f$, the second-order TD is equivalent to the traditional second-order derivative of $f$.

\begin{framed}
\begin{theorem}
\textbf{Suppose $f$ is a $C^2$ continuous function defined on a neighborhood of point $x_0$, and $w(t)$ is a function satisfying Normalization Constraint (Formula~\eqref{eq:constraintw})  with a parameter $W > 0$. $f^{\prime\prime}_{\text{TD}}(x_0;w,W)$ given by Formula~\eqref{eq:2ndtd} equals to $f^{\prime\prime}(x_0)$, namely
\begin{equation}
    f^{\prime\prime}(x_0) = \lim_{\Delta x \rightarrow 0} \frac{f^{\prime}(x_0 + \Delta x) - f'(x_0)}{\Delta x}.
\end{equation}}
\end{theorem}
\begin{proof} 
$\\$
\begin{equation*}
    \begin{aligned}
      f_{\text{TD}}^{\prime\prime}(x_0; w) =& \lim_{W \rightarrow 0} \frac{\int_0^W f(x_0+t) w(t) \d{t} + \int_0^W f(x_0-t) w(t)\d{t} - 2f(x_0)\int_{0}^{W}w(t) \d{t}}{\int_0^W t^2 w(t)\d{t}}\\
    =& \lim_{W \rightarrow 0} \frac{\int_{0}^{W} (f(x_0+t)+f(x_0-t)-2f(x_0)) w(t) \d{t}}{\int_{0}^{W}t^2 w(t)\d{t}} \quad(\text{Taylor's Formula}) \\
    =& \lim_{W \rightarrow 0} \frac{\int_{0}^{W} (t^2 f^{\prime\prime}(x_0) + O(t^3)) w(t) \d{t}}{\int_{0}^{W}t^2 w(t)\d{t}} \\
    =& f^{\prime\prime}(x_0) \lim _{W \rightarrow 0} \frac{\int_{0}^{W} t^2 w(t) \d{t}}{\int_{0}^{W}t^2 w(t)\d{t}} \\
    =& f^{\prime\prime}(x_0).
    \end{aligned}
\end{equation*}
\end{proof}
\end{framed}

\subsection{General Difference: Definition and Property}
\label{subsec:1Dwindowderivativedefinitions}

Generalized derivative is intended for the calculation of non-smooth functions with $W \rightarrow 0$, and is developed to finite difference~\cite{milne1933calculus,boole2022calculus} in discrete processing. However, in modern discrete systems, discrete signals are produced over a certain time, space, or spatiotemporal interval~\cite{li2000generalized}. In other words, $W$ does not tend toward $0$. This creates a gap between the theory of generalized derivative and the practice of discrete processing with regard to the derivative interval. We addressed this problem by studying the theory of derivative over a finite interval ($W > 0$) as the basis for establishing difference within a finite interval. 

We expand the integration interval of \textbf{TD} defined in Formula~\eqref{eq:TD} and~\eqref{eq:2ndtd} to finite interval (Formula ~\eqref{eq:TGD} and ~\eqref{eq:2NDTGD}). We name this TD computation in a finite interval as \textbf{G}eneral \textbf{D}ifference in continuous domain, and call them as \textbf{T}ao \textbf{G}eneral \textbf{D}ifference (TGD).
\begin{equation}
    f_{\TGD}^{\prime}(x_0; w, W) \triangleq \frac{\int_{0}^{W} f(x_0+t) w(t) \d{t}-\int_{0}^{W} f(x_0-t) w(t) \d{t}}{2\int_{0}^{W}t w(t)\d{t}}
    \label{eq:TGD}
\end{equation}
\begin{equation}
     f^{\prime\prime}_{\TGD}(x_0; w, W) \triangleq \frac{\int_0^W f(x_0+t) w(t) \d{t} + \int_0^W f(x_0-t) w(t)\d{t} - 2f(x_0)\int_{0}^{W}w(t) \d{t}}{\int_0^W t^2 w(t)\d{t}}
     \label{eq:2NDTGD}
\end{equation}
where $w$ satisfies Normalization Constraint (Formula~\eqref{eq:constraintw}). Definitely, these \textbf{TGDs} implement the first- and second-order difference computation in a finite interval $[x_0-W, x_0+W]$ of continuous domain. Relatively, $w(t)$ is the \emph{kernel function}, and $W$ is the \emph{kernel size} of TGD, respectively.

Through the following transformation of Formula~\eqref{eq:TGD}, it is evident that TGD can be abbreviated to a convolution form:
\begin{equation}
    \begin{aligned}
      f_{\TGD}^{\prime}(x_0; w, W) &= C_1 \left( \int_{-W}^0 f(x_0+t)(-w(-t))\d{t} + \int_0^W f(x_0+t) w(t) \d{t}\right) \\
      &= C_1 (T \ast f)(x_0)
    \end{aligned}
    \label{eq:TGDConv}
\end{equation}
where $C_1$ is the normalization constant, and $T(t)$ is defined as the \emph{First-order Tao General Difference Operator} (1st-order TGD operator):
\begin{equation}
\begin{aligned}
    C_1 &= \frac{1}{2\int_0^W tw(t) \d{t}}
\end{aligned}
\end{equation}
\begin{equation}
\begin{aligned}
    T(t) &= \left\{
    \begin{array}{ll}
      w(-t) & \quad -W \leq t < 0,\\
      -w(t) & \quad 0 < t \leq W,\\
      0 & \quad otherwise.
    \end{array}
    \right.
\end{aligned}
    \label{eq:TGD_T}
\end{equation}

Similarly, Formula~\eqref{eq:2NDTGD} can also be abbreviated to a convolution form:
\begin{equation}
    \begin{aligned}
      f_{\TGD}^{\prime\prime}(x_0; w, W) &= C_2 \left( \int_{-W}^W f(x_0+t) w(|t|)\d{t} - 2f(x_0) \int_0^W  w(t) \d{t}\right) \\
      &= C_2 (R \ast f)(x_0)
    \end{aligned}
    \label{eq:TGDConv2}
\end{equation}
where $C_2$ is the normalization constant, and $R(t)$ is defined as the \emph{Second-order Tao General Difference Operator} (2nd-order TGD operator):
\begin{equation}
\begin{aligned}
    C_2 &= \frac{1}{\int_0^W t^2 w(t) \d{t}}
\end{aligned}
\end{equation}
\begin{equation}
  \begin{aligned}
    R(t) &= \left\{
    \begin{array}{ll}
      w\left(|t|\right) - 2\delta\left(t\right) \int_0^W w\left(t\right) \mathrm{d}t & \quad -W \leq t \leq W,\\
      0 &  \quad otherwise.
    \end{array}
    \right.
  \end{aligned}
  \label{eq:TGD_R}
\end{equation}
where $\delta$ is Dirac delta function~\cite{hassani2009dirac}.

And it is clear that:
\begin{equation}
    \begin{aligned}
        \int_{-\infty}^{\infty} T(t) \d{t} = \int_{-W}^{W} T(t) \d{t} = 0 \\
        \int_{-\infty}^{\infty} R(t) \d{t} = \int_{-W}^{W} R(t) \d{t} = 0
    \end{aligned}
    \label{eq:sum_of_TGD}
\end{equation}
which means that for a constant function, its TGD value is constantly $0$.

Formula~\eqref{eq:TGDConv} and~\eqref{eq:TGDConv2} illustrate the calculation of TGD by convolution. The linearity of the convolution verifies the linearity of TGD, that is, the TGD of any linear combination of functions equals the same linear combination of the TGD of these functions:
\begin{equation}
    \begin{aligned}
        (af+bg)^{\prime}_{\TGD} = af^{\prime}_{\TGD} + bg^{\prime}_{\TGD} \\
        (af+bg)^{\prime\prime}_{\TGD} = af^{\prime\prime}_{\TGD} + bg^{\prime\prime}_{\TGD}
    \end{aligned}
\end{equation}
where $a$ and $b$ are real numbers.

In principle, many functions satisfying Normalization Constraint (Formula~\eqref{eq:constraintw}) can be applied to TGD. However, the TGD calculation should locate the centre for both slow change and fast changes in a function $f$. We introduce the unbiasedness constraint for the kernels of TGD, which mathematically describes the requirement that the calculated difference robustly fits the change at the function $f$.

~

\begin{definition}\label{def:unbiasedness}
(Unbiasedness of TGD) $f'_{\TGD}$ is unbiased within $W_T$, if $\exists W_T>0$, s.t. $\forall W_1, W_2, 0 < W_2 < W_1 \le W_T$, for arbitrary $w_1(t) \in \mathbb{E}_{W_1}$ and $w_2(t) \in \mathbb{E}_{W_2}$ holds that if $f'_{\TGD}(x_1;w_2,W_2) > f'_{\TGD}(x_2;w_2,W_2)$ then $f'_{\TGD}(x_1;w_1,W_1) > f'_{\TGD}(x_2;w_1,W_1)$. $\mathbb{E}_W$ denotes the family of suitable kernel functions $w$ with kernel size $W$. 
\end{definition}

~

As a necessary condition for preserving the definition of unbiasedness in the differential calculation, we introduce the following constraint for the kernel function $w(t)$ in TGD.

\textbf{Monotonic Constraint (C2): the kernel function $w(t)$ is monotonic decreasing within (0, W].}

\begin{framed}
\begin{theorem}
\textbf{Monotonic Constraint (C2) is the necessary condition for the Unbiasedness of TGD.}
\label{theo:unbiasedness}
\end{theorem}
\begin{proof}

Let $f(x) = \delta(x)$, the Dirac delta function~\cite{hassani2009dirac}. According to Formula~\eqref{eq:TGDConv}, for any $W_1, 0 < W_1 \le W_T$ and $w_1 \in \mathbb{E}_{W_1}$ we have
\begin{equation}
\begin{aligned}
  f'_{\TGD}(x;w_1,W_1) & = C_{1}(T \ast \delta)(x) = C_{1} T(x)\\
  & = \left\{
    \begin{array}{ll}
      C_{1} w_1(-x) & -W_1 \leq x < 0,\\
      -C_{1} w_1(x) & 0 < x \leq W_1,\\
      0 & otherwise.
    \end{array}
    \right.
\end{aligned}
\end{equation}
where $C_{1}$ is a constant positive number.

Without loss of generosity, we only discuss the case of $0 < x_2 < x_1 \le W_1$. Let $W_2 = (x_1+x_2) / 2$ and arbitrary $w_2 \in \mathbb{E}_{W_2}$, we have $0 < x_2 < W_2 < x_1 \le W_1$ and
\begin{equation}
\begin{aligned}
    f'_{\TGD}(x_2;w_2,W_2) &= -C_{2} w_2(x_2) < 0 \\
    f'_{\TGD}(x_1;w_2,W_2) &= 0 
\end{aligned} 
\end{equation}
where $C_{2}$ is also a constant positive number. So we have:
\begin{equation}
    f'_{\TGD}(x_1;w_2,W_2) > f'_{\TGD}(x_2;w_2,W_2).
\end{equation}

According to Definition~\ref{def:unbiasedness}, we have:
\begin{equation}
    -C_{1} w_1(x_1) = f'_{\TGD}(x_1;w_1,W_1) > f'_{\TGD}(x_2;w_1,W_1) = -C_1 w_1(x_2).
\end{equation}
Thus $\forall x_1, x_2, 0 < x_2 < x_1 \le W_1$ and $w_1 \in \mathbb{E}_{W_1}$, $0<w_1(x_1) < w_1(x_2)$, a.k.a. $w(t)$ is a monotonic decreasing function within $(0, W]$.
\end{proof}
\end{framed}

It is mentioned that the kernel function of LD and extended LD are $2t/W^2$ and $(k+1)t^k/W^{k+1}$ respectively. Though Normalization Constraint holds, violating Monotonic Constraint shows that they are not suitable for TGD. Contrary, Gaussian function $G(t; \sigma)$ can fit both Normalization Constraint and Monotonic Constraint when $t > 0$, and be a kernel function of TGD.

We further define the 1st-order smooth operator $S_T$ and 2nd-order smooth operator $S_R$ by integrating 1st-order TGD operator $T$ once and 2nd-order TGD operator $R$ twice, as follows:

\begin{equation}
  \begin{aligned}
    & S_T(x) \triangleq \left\{\begin{array}{ll}
        \int_{-\infty}^{x}w(-t)\d{t} & x < 0 \\
        \int_{-\infty}^{0}w(-t)\d{t} + \int_{0}^{x}-w(t)\d{t} & x \geq 0 \\
        \end{array}\right. \\
    & S_R(x) \triangleq \left\{\begin{array}{ll}
        \int_{-\infty}^{x} \left( \int_{-\infty}^{u}w(-t)\d{t} \right) \d{u} & x < 0 \\
        \int_{-\infty}^{x} \left( \int_{-\infty}^{0}w(-t)\d{t} + \int_{0}^{u}w(t)\d{t} \right) \d{u} & x \geq 0 \\
        \end{array}\right.
  \end{aligned}
  \label{function:TGD-smoothing1D}
\end{equation}
According to the \emph{first fundamental theorem of the calculus}\footnote{Let $f$ be a continuous real-valued function defined on a closed interval $[a, b]$. Let $F$ be the function defined, for all $x$ in $[a, b]$, by 
\begin{equation*}
  \begin{aligned}
    F(x)=\int_{a}^{x} f(t) d t.
  \end{aligned}
\end{equation*}
Then $F$ is uniformly continuous on $[a, b]$ and differentiable on the open interval $(a, b)$, and $F^{\prime}(x)=f(x)$ for all $x$ in $(a, b)$.
}~\cite{apostol1967calculus,spivak1994calculus}, we get both $S^{\prime}_T(x)$ and $S^{\prime\prime}_R(x)$ when $x$ is not equal to 0. Noting that $S^{\prime}_T(0)$ and $S^{\prime\prime}_R(0)$ do not exist. If we add $S^{\prime}_T(0) = 0$ and $S^{\prime\prime}_R(0) = -2\delta(0)$, then $S^{\prime}_T = T$ and $S^{\prime\prime}_R = R$ always hold. 
\begin{equation}
  \begin{aligned}
    & T(x) = \left\{\begin{array}{ll}
        S_T^{\prime}(x) & x \neq 0 \\
        0 & x = 0 \\
        \end{array}\right. \\
    & R(x) = \left\{\begin{array}{ll}
        S_R^{\prime\prime}(x) & x \neq 0 \\
        -2\delta(0) & x = 0 \\
        \end{array}\right.
  \end{aligned}
\end{equation}
Thus, based on the differential property of convolution, the Formula\eqref{eq:TGDConv} and \eqref{eq:TGDConv2} can be written as:

\begin{equation}
    \begin{aligned}
      f_{\TGD}^{\prime}(x_0; w, W) = C_1 (T \ast f)(x_0) = C_1 (S_T \ast f)^{\prime}(x_0) 
    \end{aligned}
    \label{eq:TGDConv3}
\end{equation}

\begin{equation}
    \begin{aligned}
      f_{\TGD}^{\prime\prime}(x_0; w, W) = C_2 (R \ast f)(x_0) = C_2 (S_R \ast f)^{\prime\prime}(x_0)
    \end{aligned}
    \label{eq:TGDConv4}
\end{equation}



\begin{figure}[htb]    	
    \centering    	
    \subfigure[Smooth operator]{   	 		 
        \includegraphics[width=0.31\linewidth]{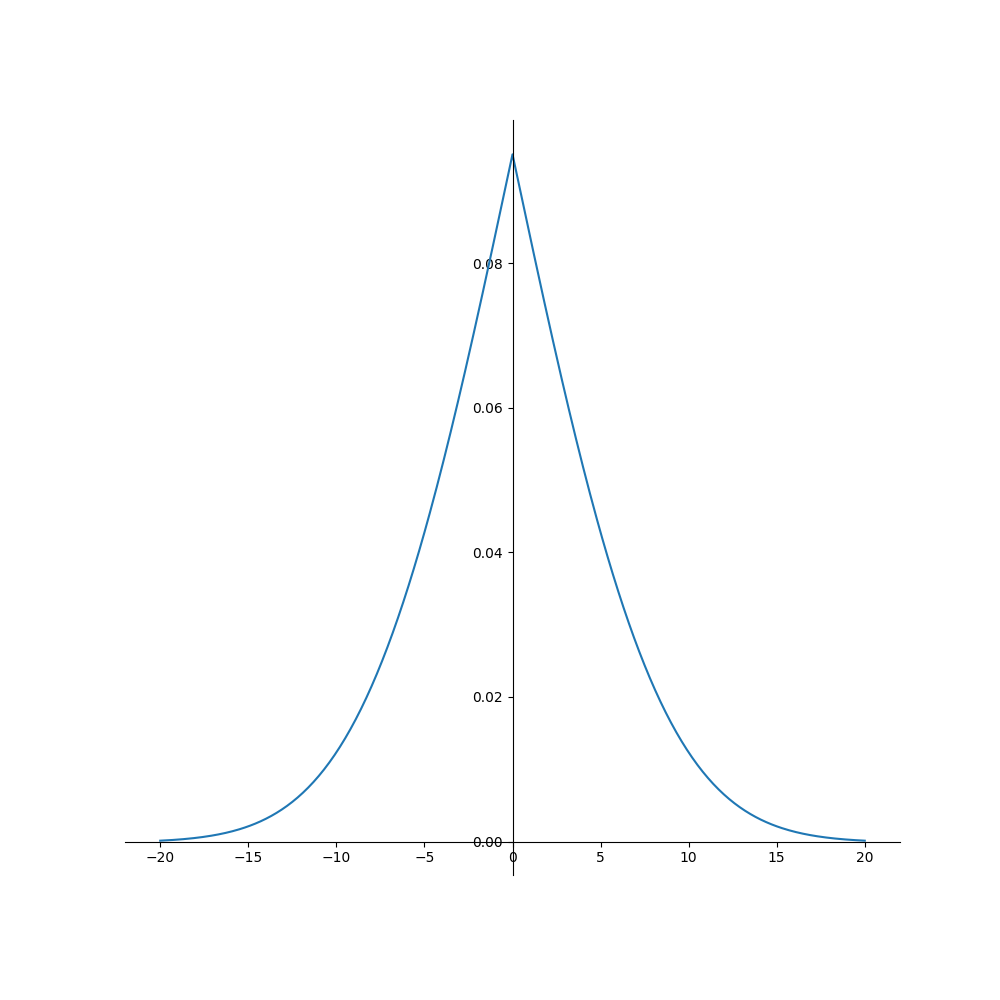}
    }
    \subfigure[First-order TGD operator]{   	 		 
        \includegraphics[width=0.31\linewidth]{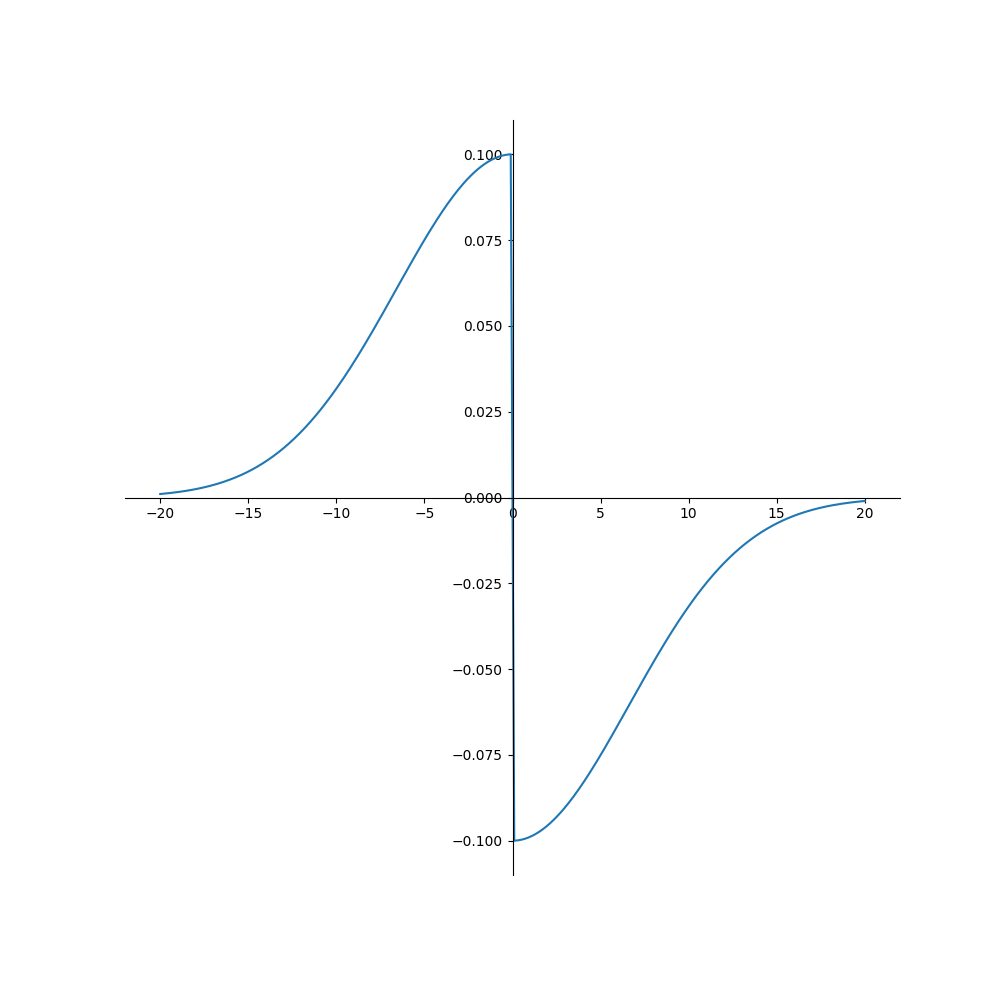}
    }   
    \subfigure[Second-order TGD operator]{   	 		 
        \includegraphics[width=0.31\linewidth]{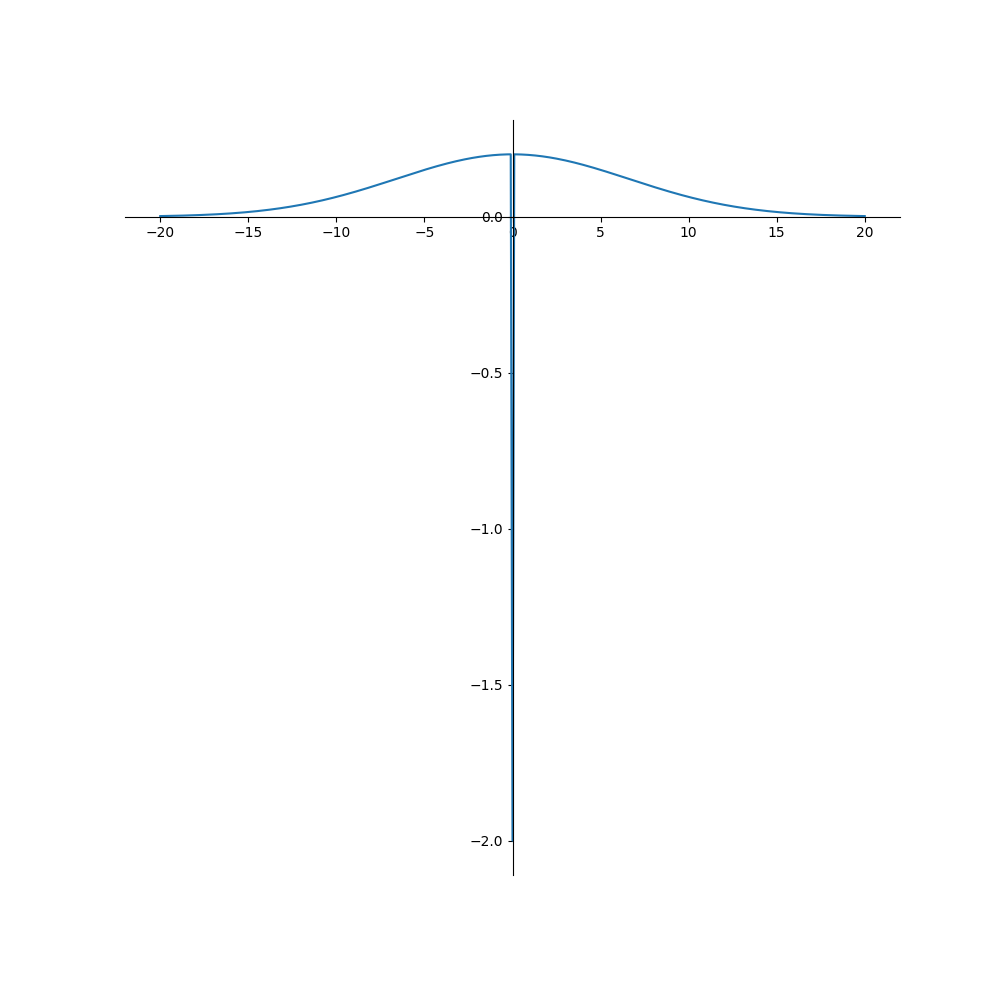}
    }
    \caption{
        The smooth operator, first- and second-order TGD operators constructed based on the Gaussian kernel function.
    }
    \label{fig:GaussConstructor}
\end{figure}

The differentiation feature\footnote{$(f*g)^{\prime} = f^{\prime}*g = g^{\prime}*f$} inherent in convolution operations highlights the significance of both $S_T$ and $S_R$ in determining the properties of TGD. TGD specifically aims to quantify the rate of change within a defined interval, demanding sensitivity to both fast and slow variations in the function $f$. To achieve this, we introduce the third constraint to both $S_T$ and $S_R$ which effectively injects sensitivity to rapid fluctuations, allowing TGD to capture both subtle and abrupt changes in the function's behavior within the interval.

\textbf{Monotonic Convexity Constraint (C3): The smooth operator $S$ is an even function with monotonic downward convexity in the interval $(0, W]$, which is $S^{\prime}(x)<0$, $S^{\prime\prime}(x)>0$ in the interval $(0, W]$, and $S(W)\rightarrow 0$, $S^{\prime}(W)\rightarrow 0$, $S^{\prime\prime}(W)\rightarrow 0$.}



It is worth pointing out that if $S_T$ satisfies the Monotonic Convexity Constraint, then the 2nd-order smooth operator $S_R$, whose kernel function $w$ corresponds to the $S_T$, also satisfies the Monotonic Convexity Constraint. Therefore, we concentrate on determining whether $S_T$ adheres to the Monotonic Convexity Constraint in later sections, and use $S$ to simplify the representation of $S_T$, namely the smooth operator. These constraints serve as beneficial supplements to Normalization Constraint.

There are various functions meeting the three constraints we introduce for TGD, which can serve as kernel function $w$ for constructing TGD operators. Here, we present a representative example of the kernel function $w$, the Gaussian function. 

Traditionally Gaussian function is widely used as a filter with convolution and differentiation calculations in wide areas. However, the Gaussian function itself does not satisfy the Monotonic Convexity Constraint, which cannot be used as a smooth operator in TGD. Its first- and second-order derivatives cannot be used as the kernel function of TGD. Unexpectedly, the Gaussian function with mean value 0 can be used as a kernel function (Formula~\eqref{eq:Gaussian}), which is controlled by the parameters $k>0$, $\delta>0$, and in general $W = 3\delta$. It is obvious from Figure~\ref{fig:GaussConstructor} that, $S_{gaussian}^{\prime}(x)<0$, $S_{gaussian}^{\prime\prime}(x)>0$ in the interval $(0, W]$, and $S_{gaussian}(W)\rightarrow 0$, $S_{gaussian}^{\prime}(W)\rightarrow 0$, $S_{gaussian}^{\prime\prime}(W)\rightarrow 0$.
\begin{equation}
  \begin{aligned}
    \text{Gaussian}(x) = ke^{-x^2/2\delta^2}, \quad x\in(0, W].
  \end{aligned}
  \label{eq:Gaussian}
\end{equation}

Coefficient $k$ controls normalization, we have:
\begin{equation*}
    \int_{0}^{W}\left(ke^{-x^2/2\delta^2}\right)\d{x} = 1. 
\end{equation*}

The first- and second-order TGD operators constructed with Gaussian function are:
\begin{equation}
  \begin{aligned}
    T_{\text{gaussian}}(x) = \left\{\begin{array}{cc}
        ke^{-x^2/2\delta^2}& \quad\quad x \in [-W,0)\\
        0& \quad\quad x=0 \\
        -ke^{-x^2/2\delta^2}& \quad\quad x\in(0, W]
    \end{array}\right.
  \end{aligned}
\end{equation}

\begin{equation}
  \begin{aligned}
    R_{\text{gaussian}}(x) = \left\{\begin{array}{cc}
        ke^{-x^2/2\delta^2}& \quad\quad x \in [-W,0) \cup (0,W]\\
        -2\delta(0)\int_{0}^{W}(ke^{-x^2/2\delta^2})\d{x}& \quad\quad x=0 
    \end{array}\right.
  \end{aligned}
\end{equation}

The integral of the Gaussian function is the error function of Gaussian, denoted by $erf(x)$. This function satisfies the Monotonic Convexity Constraint, that works as the smooth operator of TGD:
\begin{equation}
  \begin{aligned}
    S_{\text{gaussian}}(x) = \left\{\begin{array}{cc}
        k(1+\erf(x / \sqrt{2} \delta))& \quad\quad x \in [-W,0)\\
        k& \quad\quad x=0 \\
        k(1-\erf(x / \sqrt{2} \delta))& \quad\quad x \in (0,W]
    \end{array}\right.
  \end{aligned}
\end{equation}

Figure~\ref{fig:GaussConstructor} shows, from left to right, the smooth operator, the first- and second-order TGD operators constructed on the Gaussian kernel function.

\subsection{Directional General Difference of Multivariate Function}
\label{Sec:TGDofMultivariateFunction}


In multivariate calculus, the directional derivative represents the instantaneous rate of change of a multivariate function with respect to a particular direction vector $\mathbf{v}$ at a given differentiable point $\mathbf{x}$. The rate of change is interpreted as the velocity of the function moving through $\mathbf{x}$ in the direction of $\mathbf{v}$. This concept generalizes the partial derivative, where the rate of change is calculated along one of the curvilinear coordinate curves~\cite{courant2012introduction}.

The directional derivative of a $n$-variate function $f\left(\mathbf{x}\right) = f\left(x_1,x_2,\ldots,x_n\right)$ along with a vector $\mathbf{v} = \left(v_1,v_2,\ldots,v_n\right)$ is defined by the limit~\cite{wrede2010schaum}:
\begin{equation}
  \begin{aligned}
    \frac{\partial f}{\partial \mathbf{v}}\left(\mathbf{x}\right)=\lim _{h \rightarrow 0} \frac{f\left(\mathbf{x}+h\mathbf{v}\right)-f\left(\mathbf{x}\right)}{h}.
  \end{aligned}
  \label{function:originDirectionDerivative}
\end{equation}
In the follow-up description, $\mathbf{v}$ is with respect to an arbitrary nonzero vector after normalization, i.e., $\Vert \mathbf{v} \Vert = 1$. Thus the calculation of the partial derivative of a variable $x_i$ is equivalent to finding the directional derivative of $\mathbf{v}=\left(0, \ldots,0,1,0,\ldots,0\right)$. 
Relatively, the 1D TGD, at the point $\mathbf{x}$ along the direction of the unit vector $\mathbf{v}$ within an interval of length $2W$ centered at $\mathbf{x}$, can be written in the following form by replacing scalar with vector in Formula~\eqref{eq:TGD}:
\begin{equation}
  \begin{aligned}
    f_{\TGD}^{\prime}\left(\mathbf{x}; w, W, \mathbf{v}\right) &\triangleq \frac{\int^{W}_{0} f\left(\mathbf{x} + t\mathbf{v}\right)w(t) \d{t} - \int^{W}_{0} f\left(\mathbf{x} - t\mathbf{v}\right)w(t) \d{t}}{{2\int_{0}^{W}t w(t) \d{t}}},
  \end{aligned}
  \label{function:naiveDirectionDerivative}
\end{equation}

The traditional directional derivative calculation (Formula~\eqref{function:originDirectionDerivative}) yields stable derivative values since $h$ tends to $0$. The 1D TGD expands the calculation interval in one dimension, however, the function values are inherently coupled in all dimensions. TGD results become sensitive not only to the function values at the derivative direction but also to the values around the direction $\mathbf{v}$. For instance, the function value could change significantly when $\mathbf{v}$ deflects by a tiny angle. This scenario is akin to a circle with a small central angle, but the larger the radius of the circle, the farther the two points on the corresponding circumference are from each other. Large variations in the calculated values can cause instability in calculating TGD. Consequently, the expanded derivative interval (Formula~\eqref{function:naiveDirectionDerivative}) in the direction $\mathbf{v}$ is insufficient in representing the directional TGD of multivariate functions.

\subsubsection{Rotational Construction}

To ensure computational stability, the calculation interval of TGD needs to be further extended in all dimensions. In more detail, we weighted average the values of the 1D TGD of all directions whose acute angle is to $\mathbf{v}$. Initially, we define $\mathbb{S}_{\mathbf{v}}$ as the set of all unit vectors whose angle with $\mathbf{v}$ is acute:
\begin{equation}
  \begin{aligned}
       \mathbb{S}_{\mathbf{v}} \triangleq \{ \mathbf{y} \mid \Vert \mathbf{y} \Vert =1, 0 < \mathbf{v}\cdot\mathbf{y} \leq 1 \},
  \end{aligned}
  \label{function:directionvectorset}
\end{equation}
where $\Vert \cdot \Vert$ denotes the modulus length. As $\Vert \mathbf{v} \Vert = \Vert \mathbf{y} \Vert = 1$, $\mathbf{v}\cdot\mathbf{y} = \frac{\mathbf{v} \cdot \mathbf{y}}{\Vert \mathbf{v} \Vert \cdot \Vert \mathbf{y} \Vert}$ indicates the cosine similarity of $\mathbf{v}$ and $\mathbf{y}$.

Based on this, we define the \textbf{First-order Directional TGD of Multivariate Function} as the weighted average of the first-order TGD in directions where the angle with the target direction is acute:
\begin{equation}
  \begin{aligned}
       \frac{\partial_{\TGD} f}{\partial \mathbf{v}}\left(\mathbf{x}\right) &\triangleq
       \frac{ \int_{\mathbb{S}_{\mathbf{v}}} 
       f_{\TGD}^{\prime}\left(\mathbf{x}; w, W, \mathbf{y}\right) \left(\mathbf{x}\right) \widetilde{w}(\arccos{(\mathbf{v}\cdot\mathbf{y})}) \d{S}}
       {\int_{\mathbb{S}_{\mathbf{v}}}\widetilde{w}(\arccos{(\mathbf{v}\cdot\mathbf{y})}) \d{S}} \\
       &= K_1 (\widetilde{T}_{\mathbf{v}} * f)(\mathbf{x}),
  \end{aligned}
  \label{function:realDirectionDerivative}
\end{equation}
where $\d{S}$ denotes the differential element, i.e., the arc length element in the 2D case and the area element in the 3D case. And $\arccos{(\mathbf{v}\cdot\mathbf{y})}$ represents the angle between $\mathbf{v}$ and $\mathbf{y}$. $K_1$ is the normalization factor:
\begin{equation}
    \begin{aligned}
       K_1 = \frac{1}{2\int_{0}^{W}t w(t) \d{t}}\cdot \frac{1}{\int_{\mathbb{S}_{\mathbf{v}}}\widetilde{w}(\arccos{(\mathbf{v}\cdot\mathbf{y})}) \d{S}}.
    \end{aligned}
\end{equation}
$\widetilde{w}(\theta)$ is the \textbf{rotation weight function} defined in $(-\frac{\pi}{2}, \frac{\pi}{2})$. And it is an non-negative even function monotonically non-increasing when $\theta \geq 0$ and satisfies the normalization condition $\int_{-\frac{\pi}{2}}^{\frac{\pi}{2}} \widetilde{w}(\theta) \d{\theta} = 1$. For example, we can take cosine function $\widetilde{w}(\theta) = c\cos(\theta)$ where $c$ is the normalization factor. At this point, the above computational model (Formula~\eqref{function:realDirectionDerivative}) retains the  \textit{Difference-by-Convolution} mode, that is, the convolution of the multivariate function $f$ with the multivariate \textbf{first-order directional TGD operator} $\widetilde{T}_{\mathbf{v}}$.
\begin{equation}
  \begin{aligned}
       \widetilde{T}_{\mathbf{v}}(\mathbf{t}) = \left\{\begin{array}{cc}
        T\left(\Vert \mathbf{t} \Vert\right)\widetilde{w}\left(\arccos{(\cos(\mathbf{v}, \mathbf{t}))}\right) & \quad  0 < \cos(\mathbf{v}, \mathbf{t}) \leq 1, \Vert\mathbf{t}\Vert \leq W \\
        T\left(-\Vert\mathbf{t}\Vert\right)\widetilde{w}\left(\arccos{(\cos(\mathbf{v}, \mathbf{t}))} - \pi\right) & \quad -1 \leq \cos(\mathbf{v}, \mathbf{t}) < 0, \Vert\mathbf{t}\Vert \leq W \\
        0 & \quad \cos(\mathbf{v}, \mathbf{t}) = 0, \Vert\mathbf{t}\Vert \leq W \text{or} \  \Vert\mathbf{t}\Vert > W
    \end{array}\right.
  \end{aligned}
  \label{eq:TGD_highD_1st}
\end{equation}
where $T$ is the 1D first-order TGD operator, 
$\Vert \cdot \Vert$ denotes the modulus length, $\cos(\mathbf{v}, \mathbf{t}) = \frac{\mathbf{v} \cdot \mathbf{t}}{\Vert \mathbf{v} \Vert \cdot \Vert \mathbf{t} \Vert}$ indicates the cosine similarity of $\mathbf{v}$ and $\mathbf{t}$, and $\arccos{(\cos(\mathbf{v}, \mathbf{t}))}$ denotes the angle between $\mathbf{v}$ and $\mathbf{t}$.

Similarly, we define the \textbf{Second-order Directional TGD of Multivariate Function} as the weighted average of the second-order TGD in directions where the angle with the target direction is acute.
\begin{equation}
  \begin{aligned}
    f_{\TGD}^{\prime\prime}\left(\mathbf{x}; w, W, \mathbf{v}\right) \! \triangleq \! \frac{\int^{W}_{0} \! f(\mathbf{x} \!+\! t\mathbf{v})w(t)\!\d{t} \!+\! \int^{W}_{0} \! f(\mathbf{x} \!-\! t\mathbf{v})w(t)\!\d{t} \!-\! 2f(\mathbf{x})\!\int_{0}^{W}\!w(t)\!\d{t}}{\int_{0}^{W}t^2 w(t)\d{t}}
  \end{aligned}
\end{equation}
\begin{equation}
  \begin{aligned}
       \frac{\partial_{\TGD}^2 f}{\partial \mathbf{v}^2} \left(\mathbf{x}\right) &\triangleq
       \frac{\int_{\mathbb{S}_{\mathbf{v}}}
       f_{\TGD}^{\prime\prime}\left(\mathbf{x}; w, W, \mathbf{y}\right) \widetilde{w}(\arccos{(\mathbf{v}\cdot\mathbf{y})}) \d{S}}{\int_{\mathbb{S}_{\mathbf{v}}}\widetilde{w}(\arccos{(\mathbf{v}\cdot\mathbf{y})}) \d{S}} \\
       &= K_2 (\widetilde{R}_{\mathbf{v}} * f)(\mathbf{x}),
  \end{aligned}
  \label{function:realDirectionDerivative2}
\end{equation}
where $\d{S}$ denotes the differential element. And $\arccos{(\mathbf{v}\cdot\mathbf{y})}$ represents the angle between $\mathbf{v}$ and $\mathbf{y}$. $K_2$ is the normalization factor:
\begin{equation}
  \begin{aligned}
      K_2 = \frac{1}{\int_{0}^{W}t^2 w(t) \d{t}}\cdot \frac{1}{\int_{\mathbb{S}_{\mathbf{v}}}\widetilde{w}(\arccos{(\mathbf{v}\cdot\mathbf{y})}) \d{S}}
  \end{aligned}
\end{equation}
All the definitions of mathematical symbols are consistent with the first-order case. And, the above computational model (Formula~\eqref{function:realDirectionDerivative2}) is equivalent to the convolution of the multivariate function $f$ with the multivariate \textbf{second-order directional TGD operator} $\widetilde{R}_{\mathbf{v}}$.
\begin{equation} 
  \begin{aligned}
       \widetilde{R}_{\mathbf{v}}(\mathbf{t}) \!=\! \left\{\begin{array}{cc}
       \!R\left(\Vert\mathbf{t}\Vert\right)\widetilde{w}\left(\arccos{(\cos(\mathbf{v}, \mathbf{t}))}\right)\! & \quad 0 < \cos(\mathbf{v}, \mathbf{t}) \leq 1, \Vert\mathbf{t}\Vert \leq W \\
        \!R\left(-\Vert\mathbf{t}\Vert\right)\widetilde{w}\left(\arccos{(\cos(\mathbf{v}, \mathbf{t}))} \!-\! \pi\right)\! & \quad -1 \leq \cos(\mathbf{v}, \mathbf{t}) < 0, \Vert\mathbf{t}\Vert \leq W \\
        \!0\! & \quad \cos(\mathbf{v}, \mathbf{t}) \!=\! 0, \Vert\mathbf{t}\Vert \!\leq\! W \text{or} \  \Vert\mathbf{t}\Vert \!>\! W
    \end{array}\right.
  \end{aligned}
  \label{eq:TGD_highD_2nd}
\end{equation}
where $R$ is the 1D second-order TGD operator, 
$\Vert \cdot \Vert$ denotes the modulus length, $\cos(\mathbf{v}, \mathbf{t}) = \frac{\mathbf{v} \cdot \mathbf{t}}{\Vert\mathbf{v}\Vert \cdot \Vert\mathbf{t}\Vert}$ indicates the cosine similarity of $\mathbf{v}$ and $\mathbf{t}$, and $\arccos{(\cos(\mathbf{v}, \mathbf{t}))}$ represents the angle between $\mathbf{v}$ and $\mathbf{t}$.

Since the construction of $\widetilde{T}_{\mathbf{v}}$ and $\widetilde{R}_{\mathbf{v}}$ involves combining one-dimensional TGD operators with rotation weights $\widetilde{w}$, this construction method is called the \textbf{Rotational Construction Method}. Moreover, both $\widetilde{T}_{\mathbf{v}}$ and $\widetilde{R}_{\mathbf{v}}$ demonstrate symmetry, and their integrals amount to $0$. Consequently, the TGD value remains constant at $0$ for constant functions:
\begin{equation} 
  \begin{aligned}
    \int\widetilde{T}_{\mathbf{v}}(\mathbf{t})\mathrm{d}\mathbf{t} = \int\widetilde{R}_{\mathbf{v}}(\mathbf{t})\mathrm{d}\mathbf{t} = 0
  \end{aligned}
  \label{eq:TGD_highD_0}
\end{equation}

Despite the complexity of the mathematical formula, the directional TGDs defined by Formula~\eqref{function:realDirectionDerivative} and~\eqref{function:realDirectionDerivative2} are still \textit{Difference-by-Convolution}, and the actual calculation process is quite simple and efficient. Due to the unbiasedness of TGD, the normalization constants ($K_1$ and $K_2$) are usually negligible in practical applications, allowing us to calculate the directional TGD through a convolution after obtaining the operators ($\widetilde{T}_{\mathbf{v}}$ and $\widetilde{R}_{\mathbf{v}}$). 



We instantiate 2D directional TGD operators by employing various weight functions in the rotational constructed method. Without loss of generality, we consider the directional TGD at the $x$-direction, where the angle variable $\theta$ is measured from the positive direction of the $x$-axis. The rotation weight $\widetilde{w}(\theta)$ is aimed at decreasing the weight as the rotation angle increases, which can be achieved by a variety of functions, such as the linear, cosine (Formula~\eqref{function:cos-weight-function}), and exponential functions. The constructed TGD operators exhibit anisotropy and decrease as $\theta$ increases. This property allows for proper characterization of the function in a particular direction. Take the cosine function as an example of rotational weight:




\begin{equation}
  \begin{aligned}
    \widetilde{w}_{\cos}(\theta) = \frac{1}{2} \cos{\theta} \quad\quad \theta\in (-\pi/2, \pi/2)
  \end{aligned}
  \label{function:cos-weight-function}
\end{equation}
and the constraint $\int_{-\pi/2}^{\pi/2} \widetilde{w}_{\cos}(\theta) d{\theta} = 1$ is satisfied. 
Recall that when using the Gaussian kernel function, we have: 
\begin{equation*}
  \begin{aligned}
    T_{\text{gaussian}}(x) = \left\{\begin{array}{cc}
        ke^{-x^2/2\delta^2}& \quad\quad x \in [-W,0)\\
        0& \quad\quad x=0 \\
        -ke^{-x^2/2\delta^2}& \quad\quad x\in(0, W]
    \end{array}\right.
  \end{aligned}
\end{equation*}
\begin{equation*}
  \begin{aligned}
    R_{\text{gaussian}}(x) = \left\{\begin{array}{cc}
        ke^{-x^2/2\delta^2}& \quad\quad x \in [-W,0) \cup (0,W]\\
        -2\delta(0)& \quad\quad x=0 
    \end{array}\right.
  \end{aligned}
\end{equation*}

Combining the rotation weights $\widetilde{w}(\theta)$ and the 1D TGD operators $T_{\text{gaussian}}(x)$ and $R_{\text{gaussian}}(x)$, the 2D TGD operators $\widetilde{T}_{x}(x,y)$ and $\widetilde{R}_{x}(x,y)$ can be constructed as follows:
\begin{equation}
  \begin{aligned}
       \widetilde{T}_{x}(x,y) \!=\! \left\{\begin{array}{cc}
        \!T_{\text{gaussian}}\left(\sqrt{x^2+y^2}\right)\cdot\frac{x}{2\sqrt{x^2+y^2}} & \quad x > 0, \sqrt{x^2+y^2} \leq W \\
        \!0 & \quad x = 0, \sqrt{x^2+y^2} \leq W \\
        \!T_{\text{gaussian}}\left(-\sqrt{x^2+y^2}\right)\cdot\frac{-x}{2\sqrt{x^2+y^2}} & \quad x < 0, \sqrt{x^2+y^2} \leq W
    \end{array}\right.
  \end{aligned}
\end{equation}
\begin{equation}
  \begin{aligned}
       \widetilde{R}_{x}(x,y) \!=\! \left\{\begin{array}{cc}
        \!R_{\text{gaussian}}\left(\sqrt{x^2+y^2}\right)\cdot\frac{x}{2\sqrt{x^2+y^2}} & \quad x > 0, \sqrt{x^2+y^2} \leq W \\
        \!-2\delta(0) & \quad x = y = 0 \\
        \!0 & \quad  x = 0, y \in [-W,0) \cup (0,W] \\
        \!R_{\text{gaussian}}\left(-\sqrt{x^2+y^2}\right)\cdot\frac{-x}{2\sqrt{x^2+y^2}} & \quad x < 0, \sqrt{x^2+y^2} \leq W
    \end{array}\right.
  \end{aligned}
\end{equation}

\begin{figure}[!htb]
  \begin{minipage}[b]{1.0\linewidth}
    \centering
    \subfigure[2D First- and Second-order directional TGD operator]{   	 		 
        \includegraphics[width=0.8\linewidth]{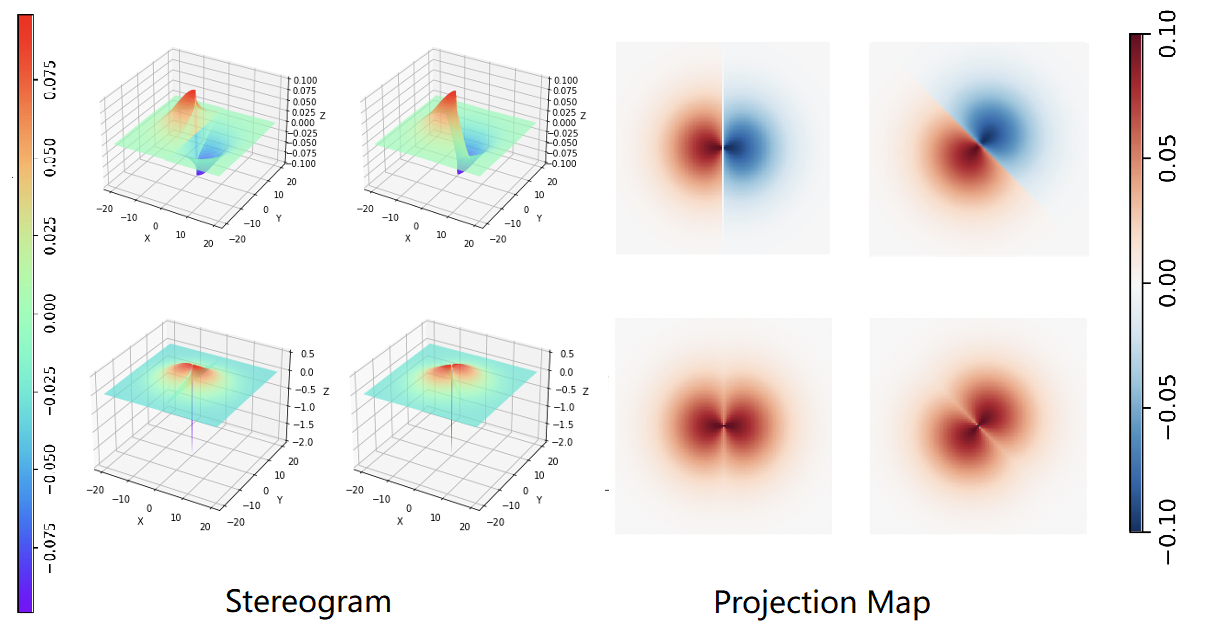}
    }  
    \subfigure[The cross-section projection of the First and Second order directional TGD operator]{   	 		 
        \includegraphics[width=0.8\linewidth]{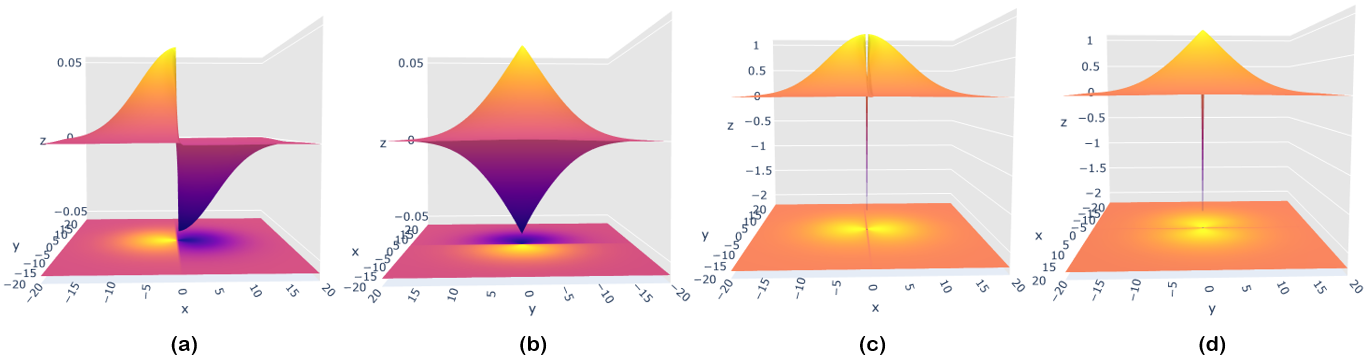}
    }  
  \end{minipage}
  \caption{
    The 2D First- and Second-order directional TGD operator (a) $\widetilde{T}_{\mathbf{v}}$ and its projection map, cross-section projection (b). Gaussian kernel function is used, and the derivation directions include 0' and 45', $\widetilde{w}(\theta)$ is a cosine function.
  }
  \label{fig:2DTGDkernel1}
\end{figure}

Based on the TGD operators $\widetilde{T}_{x}$ and $\widetilde{R}_{x}$, the 2D directional TGD operators $\widetilde{T}_{\theta}$ and $\widetilde{R}_{\theta}$ with the angle $\theta$ between the difference direction and the $x$-axis, can be obtained by the coordinate rotation operation.
\begin{equation}
  \begin{aligned}
    \left[\begin{array}{l}
        x' \\
        y'
        \end{array}\right]&=\left[\begin{array}{cc}
        \cos \theta & \sin \theta \\
        -\sin \theta & \cos \theta
        \end{array}\right]\left[\begin{array}{l}
        x \\
        y
    \end{array}\right] \\
    &\widetilde{T}_{\theta}(x,y) = \widetilde{T}_{x}(x',y') \\
    &\widetilde{R}_{\theta}(x,y) = \widetilde{R}_{x}(x',y')
  \end{aligned}
  \label{eq:rotation_methods}
\end{equation}

Figure~\ref{fig:2DTGDkernel1} displays the 2D first-order TGD operators $\widetilde{T}_{\theta}$ and the 2D second-order TGD operators $\widetilde{R}_{\theta}$ based on the Gaussian kernel function and their projection diagram, where $\widetilde{w}(\theta)$ chooses cosine function. The figures feature examples of the 0' and 45' derivation directions. The projection diagram demonstrates that the operator energy is primarily concentrated in the differential direction. 


Of particular interest is the demonstration that this operator performs TGD calculations in the differential direction, while filtering via approximately integral Gaussian function (Fig. ~\ref{fig:GaussConstructor}.a) in the orthogonal direction (Fig. ~\ref{fig:2DTGDkernel1}.b), thereby suppressing noise during TGD calculation. This indicates that the 2D TGD is filtered in its orthogonal direction to suppress the noise while calculating the TGD. The two together make the TGD calculation stable to function values as well as noise, highlighting the practicality of this method in real-world applications. 


In mathematics, the Laplace operator is a differential operator defined as the divergence of the gradient of a function in Euclidean space. The Laplacian is invariant under all Euclidean transformations: rotations and translations. The Laplacian appears in differential equations that describe various physical phenomena, hence its extensive application in scientific modeling. Additionally, in computer vision, the Laplace operator has proven useful for multiple tasks, including blob and edge detection~\cite{marr1980theory,kong2013generalized,wang2007laplacian}. 

When we take a constant function in the rotation weight $\widetilde{w}_{\text{constant}}$ (Formula~\eqref{function:constant-weight-function}) and $\int_{-\pi/2}^{\pi/2} \widetilde{w}_{\text{constant}}(\theta) d{\theta} = 1$,
then the 2D second-order TGD operator becomes an isotropic operator with rotational invariance. We notice that these operators conform to the definition of the Laplace operator and take various forms depending on the employed kernel functions. We collectively refer to them as \textbf{L}aplace \textbf{o}f \textbf{T}GD (LoT) operators~\footnote{Please refer to the Figure 1 in Applications of Tao General Difference in Discrete Domain.}.
\begin{equation}
  \begin{aligned}
    \widetilde{w}_{\text{constant}}(\theta) = 1/\pi \quad \theta\in [-\pi/2, \pi/2],
  \end{aligned}
  \label{function:constant-weight-function}
\end{equation}
\begin{equation}
  \begin{aligned}
       \text{LoT}(x,y) = \left\{\begin{array}{cc}
        R\left(\sqrt{x^2+y^2}\right) \cdot \frac{1}{\pi} & \quad\quad 0< \sqrt{x^2+y^2} \leq W \\
        -2\delta(0) & \quad\quad x = y = 0
    \end{array}\right.
  \end{aligned}
  \label{eq:TGD_highD_laplace}
\end{equation}

It is worth noting that the LoT operator still satisfies the integral of $0$, which indicates that for a constant function, its Laplace TGD is constant and equal to $0$.


The rotational construction method can also be utilized to construct multivariate directional TGD operators. The constructed operator is expressed in dense high-dimensional convolution kernel, which is the high dimensional directional General Difference by Convolution.

\begin{figure}[htb]
  \begin{minipage}[b]{1.0\linewidth}
    \centering
    \centerline{\includegraphics[width=\linewidth]{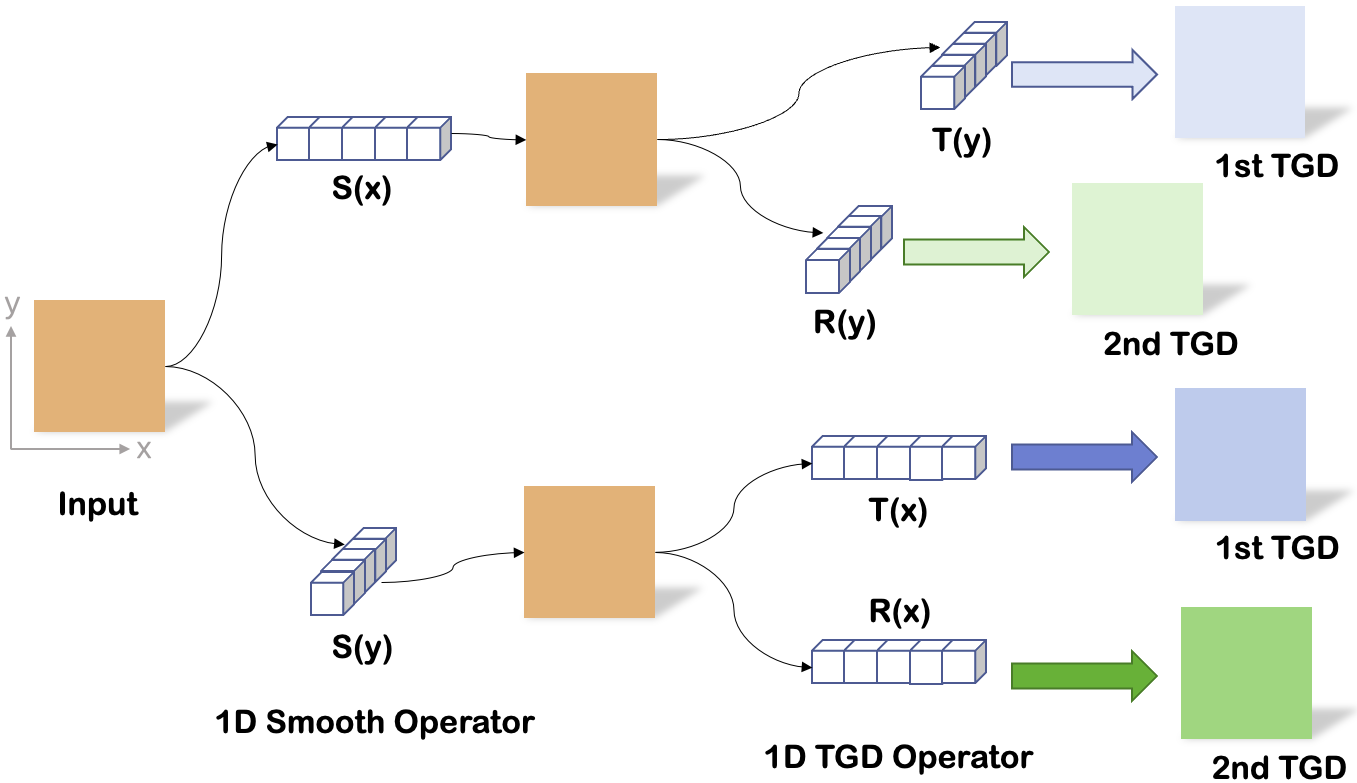}}
  \end{minipage}
  \caption{
    The decomposition diagram of the convolution calculation process of an image and 2D TGD operators obtained by the orthogonal construction method.
  }
  \label{fig:2DkernelConv}
\end{figure}

\subsubsection{Orthogonal Construction}
\label{subsec:Orthogonal construction method}

\begin{figure}[htb]
  \begin{minipage}[b]{1.0\linewidth}
    \centering
    \centerline{\includegraphics[width=\linewidth]{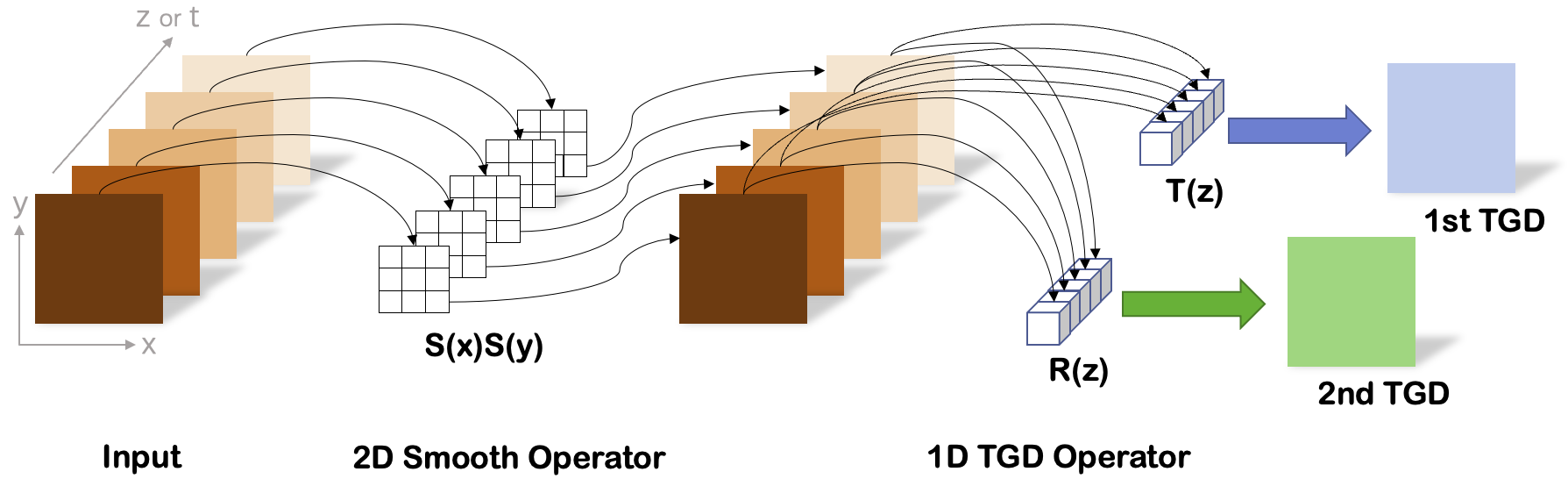}}
  \end{minipage}
  \caption{
    The decomposition diagram of the convolution calculation process of a series of images and 3D TGD operators obtained by the orthogonal construction method.
  }
  \label{fig:3DkernelConv}
\end{figure}

For calculating the TGD of high dimensional functions efficiently, we introduce a simple and intuitive way of constructing operators, called the \textbf{Orthogonal Construction Method}. Figure~\ref{fig:2DTGDkernel1}.b demonstrate that the 2D TGD operator performs general differential calculation in the differential direction and low-pass filtering based on a convex function in the orthogonal direction. This suggests that the 2D TGD computes the 1D TGD in its differential direction while filtering in its orthogonal direction to suppress noise. Therefore, we define the \textbf{2D orthogonal TGD operators} as follows (Figure~\ref{fig:2DkernelConv}):
\begin{equation}
  \begin{aligned}
        \widetilde{T}_{x\perp}(x,y) \triangleq T(x) \cdot S(y) \\
        \widetilde{R}_{x\perp}(x,y) \triangleq R(x) \cdot S(y) 
  \end{aligned}
\end{equation}
where $T(x)$ is the 1D first-order TGD operator, $R(x)$ is the 1D second-order TGD operator, and $S(x)$ is the 1D smooth operator. 

Similarly, the multivariate \textbf{orthogonal TGD operators} are defined as follows (Figure~\ref{fig:3DkernelConv}):
\begin{equation}
  \begin{aligned}
        \widetilde{T}_{x_i \perp}(x_1,\ldots,x_{i-1},x_i,x_{i+1},\ldots,x_n) \triangleq S(x_1)\ldots S(x_{i-1}) T(x_i) S(x_{i+1})\ldots S(x_n) \\
        \widetilde{R}_{x_i \perp}(x_1,\ldots,x_{i-1},x_i,x_{i+1},\ldots,x_n) \triangleq S(x_1)\ldots S(x_{i-1}) R(x_i) S(x_{i+1})\ldots S(x_n) 
  \end{aligned}
  \label{eq:TGD_highD_orth}
\end{equation}

Based on Formula~\eqref{eq:sum_of_TGD}, it is easy to deduce that the orthogonal TGD operator satisfies an integral of $0$:
\begin{equation} 
  \begin{aligned}
    \int\ldots\int\widetilde{T}_{x_i \perp}\left(x_1,\ldots,x_{i-1},x_i,x_{i+1},\ldots,x_n\right)\mathrm{d}x_1\ldots\mathrm{d}x_n = 0 \\
    \int\ldots\int\widetilde{R}_{x_i \perp}\left(x_1,\ldots,x_{i-1},x_i,x_{i+1},\ldots,x_n\right)\mathrm{d}x_1\ldots\mathrm{d}x_n = 0 \\
  \end{aligned}
  \label{eq:TGD_highD_0_2}
\end{equation}
which means that the TGD value is constant to $0$ for constant functions.

The multivariate TGD operators constructed via the orthogonal construction method perform TGD calculations along the coordinate axes, which is referred to as \textbf{partial TGD}, corresponding to the conventional partial derivative. It is, however, practical to construct TGD operators in any direction by employing coordinate rotation (Formula~\eqref{eq:rotation_methods}).

Figure~\ref{fig:orthogonalKernel2D} illustrates the visualization results of the 2D first- and second-order TGD operators created using the orthogonal construction method. Compared to the TGD operators produced with the rotational construction method (Figure~\ref{fig:2DTGDkernel1}), 
the first-order TGD operator remains practically unaltered, while the second-order TGD operator changes from having only one negative value at the origin to comprising a line, which is perpendicular to the differential direction and negative at the center. Furthermore, Figure~\ref{fig:3Dkernel} displays the cross-sections and contours of the 3D first- and second-order TGD operators in the $z$-direction, generated by the orthogonal construction method. 

\begin{figure}[htb]
    \centering
    \begin{minipage}[b]{0.8\linewidth}
        \centerline{\includegraphics[width=\linewidth]{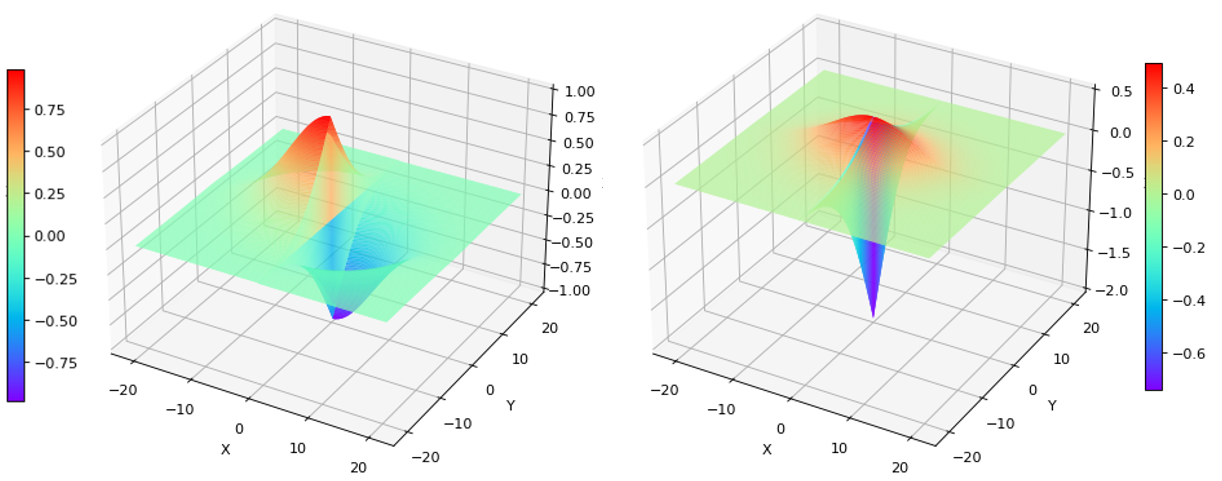}}
    \end{minipage}
    \caption{
    2D first- and second-order TGD operators based on the orthogonal construction method. Gaussian kernel function is used for both TGD and filtering.
    }
    \label{fig:orthogonalKernel2D}
\end{figure}

\begin{figure}[!htb]    	
    \centering    	
    \subfigure[]{  			 
        \includegraphics[width=0.8\linewidth]{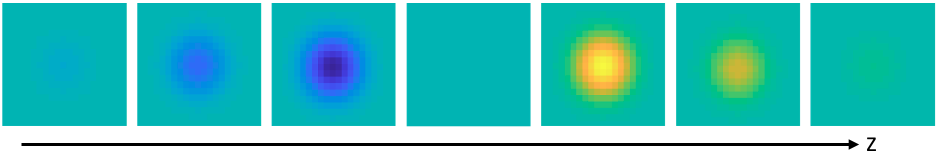}
    }    	 
    \subfigure[]{   	 		 
        \includegraphics[width=0.8\linewidth]{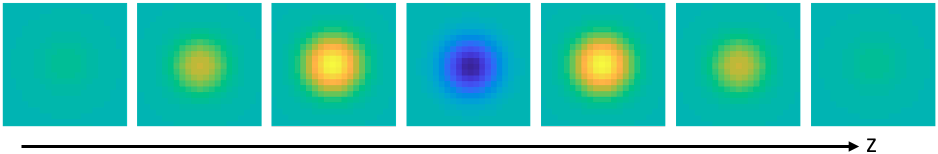}
    }  
    \subfigure[]{  			 
        \includegraphics[width=0.8\linewidth]{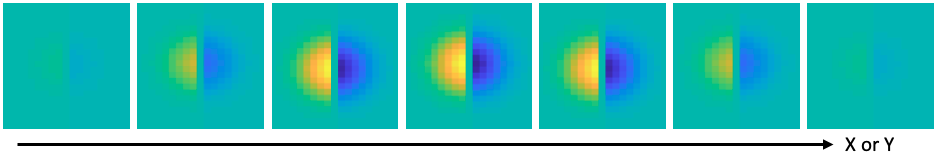}
    }    	 
    \subfigure[]{   	 		 
        \includegraphics[width=0.8\linewidth]{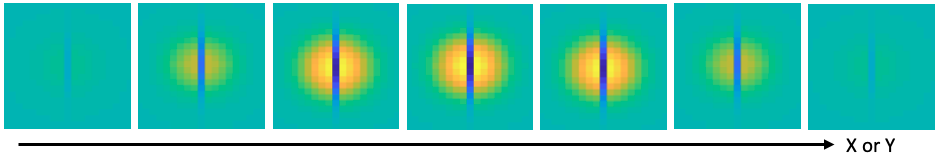}
    }  
    \caption{Schematic diagrams of 3D first- and second-order z-directional TGD operators obtained by the orthogonal construction method: (a) a vertical z-directional cross-section of the first-order operator; (b) a vertical z-directional cross-section of the second-order operator; (c) a vertical x-(or y-)directional cross-section of the first-order operator; (d) a second-order operator vertical x-(or y-)direction cross-sectional diagram. Gaussian kernel function is used for both TGD and filtering.}  
    \label{fig:3Dkernel} 
\end{figure}

It is worth pointing out that the convex smooth operator can be approximated using a traditional Gaussian filter. Consider the operator obtained by the rotational construction method, and suppose the TGD kernel function uses Gaussian and the rotation weight $\widetilde{w}(\theta)$ is taken as a constant. In that case, the filter function in the orthogonal direction can be approximated using a Gaussian function. In practice, the smooth operator in the orthogonal construction method can be implemented using a variety of functions based upon specific needs.

The use of orthogonal operators brings about decoupling between dimensions, allowing each dimension to operate independently. This means that the 1D TGD is computed exclusively in the differential direction while filtering occurs in all orthogonal directions. The resulting decoupling enhances the speed of convolution computation while enhancing efficiency. Figure~\ref{fig:2DkernelConv} and Figure~\ref{fig:3DkernelConv} show decomposition diagrams detailing the convolution calculation process of the 2D and 3D TGD operators, respectively, constructed using the orthogonal construction method. In the 2D TGD computation, the 2D convolution can be split into two 1D convolutions. And in the 3D TGD calculation process, the dense 3D convolution can be divided into two parts: a 2D filtering and a 1D convolution.

\clearpage
\newpage

\section{General Difference in Discrete Domain}


In modern discrete systems, a sequence is sampled or digitized over a certain time, space, or spatiotemporal interval from a continues function $f(x)$. A sequence is denoted as $X(n)$, and the one-dimensional sequence is written as:
\begin{equation}
    X(n) = f(n\Delta),
    \label{eq:discrete_signal}
\end{equation}
where $n$ is integer, and $\Delta > 0$ is a discrete step known as the \emph{interval}. Without loss of generality, we set the interval $\Delta = 1$ in the following text to facilitate the expression.

Differential calculus necessitates continuous functions. However, in practice, we cannot reverse the formula~\eqref{eq:discrete_signal} to obtain the original continuous function $f(x)$ from the sequence $X(n)$. Thus, the difference is used to perform a similar calculation in a sequence as the differentiation in a continuous function. The 1st- and 2nd-order center difference of a sequence is defined among the neighbors in a sequence:
\begin{equation}
    \begin{aligned}
        X^{\prime}\left(n\right) & = \frac{X(n+1) - X(n-1)}{2} \\
        X^{\prime\prime}\left(n\right) & = X(n+1) - 2X(n) + X(n-1)
    \end{aligned}
    \label{eq:difference_definition}
\end{equation}
Notably, $\Delta > 0$ violates the fundamental definition of differentiation $\Delta \rightarrow 0$ in Formula~\eqref{eq:difference_definition}. The accuracy of this approximation degrades sharply with increasing $\Delta$. Therefore, the finite difference in discrete domain should be built on a solid theoretical foundation of finite interval-based difference in continue domain, rather than conventional differentiation.

\subsection{General Difference of One Dimensional Sequence}
TGD calculations in the continuous domain (Formula~\eqref{eq:TGD} and~\eqref{eq:2NDTGD}) cannot be directly applied to discrete sequences. Here, we construct a continuous function from a sequence simply by treating it as a continuous step function for the Difference-by-Convolution in a sequence. This step function aligns closely with the key properties of the sequence:
\begin{equation}
    \begin{aligned}
        \hat{f}(x) \triangleq X(n) \quad\quad x \in (n-\frac{1}{2}, n+\frac{1}{2}]
    \end{aligned}
    \label{eq:discrete_xn}
\end{equation}


Thus, the first-order difference of a sequence  can be calculated as the first-order Tao General Difference (Formula~\eqref{eq:TGD}) of the corresponding step function (Formula~\eqref{eq:discrete_xn}):
\begin{equation}
    \begin{aligned}
        X^{\prime}_{\TGD}\left(n; w, N\right) &\triangleq \hat{f}^{\prime}_{\TGD}\left(n; w, N+\frac{1}{2}\right) \\ 
        &= \frac{\int_{0}^{N+\frac{1}{2}} \hat{f}(n+t) w(t) \d{t}-\int_{0}^{N+\frac{1}{2}} \hat{f}(n-t) w(t) \d{t}}{2\int_{0}^{N+\frac{1}{2}}t w(t)\d{t}} \\
        &= C_1 \int_{-N-\frac{1}{2}}^{N+\frac{1}{2}} \hat{f}(n+t) T(-t) \d{t} \\
        &= C_1 \sum_{i=-N}^{N} \int_{i-\frac{1}{2}}^{i+\frac{1}{2}} \hat{f}(n+t) T(-t) \d{t} \\
        &= C_1 \sum_{i=-N}^{N} X(n + i) \int_{i-\frac{1}{2}}^{i+\frac{1}{2}} T(-t) \d{t} \\
        &= C_1 \left(\widehat{T} * X\right)(n)
    \end{aligned}
    \label{eq:discrete_TGD1}
\end{equation}
where $C_1$ is the normalization constant, and $\widehat{T}$ is the interval integral of the 1st-order TGD operator $T$, which is a discrete sequence with $2N+1$ elements. Here, we call $\widehat{T}$ as the \emph{First-order Discrete Tao General Difference Operator} (1st-order Discrete TGD Operator), and $N$ as the \emph{kernel size}.
\begin{equation}
    \begin{aligned}
        C_1 = \frac{1}{2\int_0^{N+\frac{1}{2}} tw(t) \d{t}}
    \end{aligned}
\end{equation}
\begin{equation}
    \begin{aligned}
        \widehat{T}(i) = \int_{i-\frac{1}{2}}^{i+\frac{1}{2}} T(t) \d{t}, \quad i = -N, -N+1, \dots, N-1, N
    \end{aligned}
    \label{eq:DTGD_operator_1st_1D}
\end{equation}

Similarly, we express the second-order Tao General Difference (TGD) of a sequence as the second-order Tao General Difference (Formula~\eqref{eq:2NDTGD}) of the corresponding step function:
\begin{equation}
    \begin{aligned}
        X^{\prime\prime}_{\TGD}\left(n; w, N\right) &\triangleq \hat{f}^{\prime\prime}_{\TGD}\left(n; w, N+\frac{1}{2}\right) \\ 
        &= C_2 \int_{-N-\frac{1}{2}}^{N+\frac{1}{2}} \hat{f}(n+t) R(-t) \d{t} \\
        &= C_2 \sum_{i=-N}^{N} X(n + i) \int_{i-\frac{1}{2}}^{i+\frac{1}{2}} R(-t) \d{t} \\
        &= C_2 \left(\widehat{R} * X\right)(n)
    \end{aligned}
    \label{eq:discrete_TGD2}
\end{equation}
where $C_2$ is the normalization constant, and $\widehat{R}$ is the interval integral of the 2nd-order TGD operator $R$, which is a discrete sequence with $2N+1$ elements. Here, we call $\widehat{R}$ as the \emph{Second-order Discrete Tao General Difference Operator} (2nd-order Discrete TGD Operator).
\begin{equation}
    \begin{aligned}
        C_2 &= \frac{1}{\int_0^{N+\frac{1}{2}} t^2 w(t) \d{t}}
    \end{aligned}
\end{equation}
\begin{equation}
    \begin{aligned}
        \widehat{R}(i) = \int_{i-\frac{1}{2}}^{i+\frac{1}{2}} R(t) \d{t}, \quad i = -N, -N+1, \dots, N-1, N
    \end{aligned}
    \label{eq:DTGD_operator_2nd_1D}
\end{equation}

Formula~\eqref{eq:discrete_TGD1} and~\eqref{eq:discrete_TGD2} manifest the Difference-by-Convolution for a sequence. Recalling the expressions for TGD operator $T$ (Formula~\eqref{eq:TGD_T}) and $R$ (Formula~\eqref{eq:TGD_R}), we can further expressed Formula~\eqref{eq:DTGD_operator_1st_1D} and~\eqref{eq:DTGD_operator_2nd_1D} in the following form:
\begin{equation}
    \begin{aligned}
        \widehat{T}(i) = \left\{\begin{array}{cc}
        \int_{i-\frac{1}{2}}^{i+\frac{1}{2}} w(-t) \d{t} > 0 \quad &  i = -N, -N+1,\dots, -1 \\
        \int_{-\frac{1}{2}}^{0} w(-t) \d{t} + \int_{0}^{\frac{1}{2}} -w(t) \d{t} = 0 \quad & i = 0 \\
        \int_{i-\frac{1}{2}}^{i+\frac{1}{2}} -w(t) \d{t} < 0 \quad & i = 1,2,\dots, N
    \end{array}\right.
    \end{aligned}
    \label{eq:DTGD_operator_1st_1D_2}
\end{equation}
\begin{equation}
    \begin{aligned}
        \widehat{R}(i) = \left\{\begin{array}{cc}
        \int_{i-\frac{1}{2}}^{i+\frac{1}{2}} w(-t) \d{t} > 0 \quad & i = -N, -N+1,\dots, -1 \\
        \int_{-\frac{1}{2}}^{\frac{1}{2}} w(|t|) \d{t} - 2\int_{-\frac{1}{2}}^{\frac{1}{2}}\delta(t) \d{t} = 2\int_{0}^{\frac{1}{2}} w(t) \d{t} - 2 < 0 \quad & i = 0 \\
        \int_{i-\frac{1}{2}}^{i+\frac{1}{2}} w(t) \d{t} > 0 \quad & i = 1,2,\dots,N
    \end{array}\right.
    \end{aligned}
    \label{eq:DTGD_operator_2nd_1D_2}
\end{equation}
And we have $\widehat{T}(i) = -\widehat{T}(-i)$, $\widehat{R}(i) = \widehat{R}(-i)$. Besides, from Formula~\eqref{eq:sum_of_TGD}, it is clear that:
\begin{equation}
    \begin{aligned}
        &\sum_{i=-N}^{N} \widehat{T}(i) = \int_{-N-\frac{1}{2}}^{N + \frac{1}{2}} T(t) \d{t} =  0 \\
        &\sum_{i=-N}^{N} \widehat{R}(i) = \int_{-N-\frac{1}{2}}^{N + \frac{1}{2}} R(t) \d{t} = 0
    \end{aligned}
    \label{eq:DTGD_sum_to_0}
\end{equation}
which means that for a constant sequence, the TGD result is constant at $0$. After normalizing the discrete TGD operator here by adjusting the normalization constant ($C_1$ and $C_2$), it follows:
\begin{equation}
    \begin{aligned}
    \left|\sum_{i=-N}^{-1}\widehat{T}(i)\right| = \left|\sum_{i=1}^{N}\widehat{T}(i)\right| = \left|\sum_{i=-N}^{-1}\widehat{R}(i)\right| = \left|\sum_{i=1}^{N}\widehat{R}(i)\right| = 1
    \end{aligned}
    \label{eq:DTGD_sum_to_1}
\end{equation}
where $| \cdot |$ denotes the absolute value. Combining Formula~\eqref{eq:DTGD_sum_to_0} and Formula~\eqref{eq:DTGD_sum_to_1}, we get $\widehat{T}(0) = 0$ and $\widehat{R}(0) = -2$.
Furthermore, as the kernel function $w$ for constructing the TGD operators $T$ and $R$ satisfies the Monotonic Constraint, we have:
\begin{equation}
    \begin{aligned}
        |\widehat{T}(i)| > |\widehat{T}(j)| > 0, \quad & 0 < |i| < |j| \leq N \\
        \widehat{R}(i) > \widehat{R}(j) > 0, \quad & 0 < |i| < |j| \leq N \\
    \end{aligned}
\end{equation}

%

Towards practicality and simplicity, the discrete operator weights can be obtained through a direct discrete sampling of the kernel function without the requirement for intricate interval integration (Formula~\eqref{eq:DTGD_operator_1st_1D_2} and Formula~\eqref{eq:DTGD_operator_2nd_1D_2}). 

The interval integral of a kernel function is equal to the discrete sampling of a kernel function having the same form, if the kernel size is much larger than the sampling step. This is a new constraint: 

\textit{$N$ is significantly larger than $1$. }

Thus, the discrete TGD operator can be obtained:
\begin{equation}
    \begin{aligned}
    \widehat{T}(i) = \left\{\begin{array}{cc}
        \frac{T\left(i\right)}{|\sum_{j=-N}^{-1} T\left(j\right)|} = \frac{w\left(-i\right)}{\sum_{j=1}^{N} w\left(j\right)} &  i = -N, -N+1,\dots, -1 \\
        0 \quad & i = 0 \\
        \frac{T\left(i\right)}{|\sum_{j=1}^{N} T\left(j\right)|} = \frac{-w\left(i\right)}{\sum_{j=1}^{N} w\left(j\right)} & \quad i = 1, 2, \dots, N-1, N
    \end{array}\right.
    \end{aligned}
    \label{eq:DTGD_operators_weights_1st_1D}
\end{equation}
\begin{equation}
    \begin{aligned}
    \widehat{R}(i) = \left\{\begin{array}{cc}
        \frac{R\left(i\right)}{|\sum_{j=-N}^{-1} R\left(j\right)|} = \frac{w\left(-i\right)}{\sum_{j=1}^{N} w\left(j\right)} &  i = -N, -N+1,\dots, -1 \\
        -2 \quad & i = 0 \\
        \frac{R\left(i\right)}{|\sum_{j=1}^{N} R\left(j\right)|} = \frac{w\left(i\right)}{\sum_{j=1}^{N} w\left(j\right)} & \quad i = 1, 2, \dots, N-1, N
    \end{array}\right.
    \end{aligned}
    \label{eq:DTGD_operators_weights_2nd_1D}
\end{equation}
where $w$ is the kernel function of TGD. Finally, the discrete TGD operators $\widehat{T}$ and $\widehat{R}$ with $2N+1$ elements have the following form.
\begin{equation}
    \setcounter{MaxMatrixCols}{30}
    \begin{aligned}
        \widehat{T} &= \begin{bmatrix} \widehat{w}_{N} & \widehat{w}_{N-1} & \dots & \widehat{w}_{1} & 0 & -\widehat{w}_{1} & \dots & -\widehat{w}_{N-1} & -\widehat{w}_{N} \end{bmatrix} \\
        \widehat{R} &= \begin{bmatrix} \widehat{w}_{N} & \widehat{w}_{N-1} & \dots & \widehat{w}_{1} & -2 & \widehat{w}_{1} & \dots & \widehat{w}_{N-1} & \widehat{w}_{N} \end{bmatrix}
    \end{aligned}
    \label{eq:DTGD_operators_1D}
\end{equation}
where:
\begin{equation}
    \begin{aligned}
        \widehat{w}_{i} = \frac{w\left(i\right)}{\sum_{j=1}^{N} w\left(j\right)}, \quad i = 1, 2, \dots, N-1, N
    \end{aligned}
\end{equation}
and:
\begin{equation}
    \begin{aligned}
        &\sum_{i=1}^{N} \widehat{w}_{i} = 1 \\
        &\widehat{w}_{i} > \widehat{w}_{j}, \quad 0 < i < j \leq N
    \end{aligned}
\end{equation}
In practical applications, it is generally acceptable to ignore the normalization factors ($C_1$ and $C_2$) and compute the discrete TGD through convolution after obtaining the operators ($\widehat{T}$ and $\widehat{R}$). 

Here, we provide several examples of discrete TGD operators. First of all, let us consider a special case where the kernel size of TGD is exactly equal to the discrete step, i.e., $N = 1$. According to Formula~\eqref{eq:DTGD_operators_weights_1st_1D},~\eqref{eq:DTGD_operators_weights_2nd_1D} and~\eqref{eq:DTGD_operators_1D}, regardless of the kernel function used, only a unique operator is obtained:
\begin{equation}
    \begin{aligned}
        \widehat{T} = \begin{bmatrix} 1 & 0 & -1 \end{bmatrix} \\
        \widehat{R} = \begin{bmatrix} 1 & -2 & 1 \end{bmatrix}
    \end{aligned}
\end{equation}
These are the traditional first-order and second-order central difference. However, the constraint \textit{“$N$ is significantly larger than $1$”} is not met and the \textit{Monotonic Constraint} is broken. This means that the central difference is not discrete TGD and does not have the properties of TGD operators.

Actually, one of the contributions of TGD is to provide larger-size operators constructed using specialized kernel functions. We showcase an example of 1st- and 2nd-order discrete TGD operators constructed with Gaussian kernel function, and kernel size $N=5$ and $7$ respectively, as seen in Formula ~\eqref{eq:DTGD_Example 5} and~\eqref{eq:DTGD_Example 7}. These TGD operators can be conveniently used as convolution kernels for discrete TGD calculation.

\begin{equation}
    \setcounter{MaxMatrixCols}{20}
    \begin{aligned}
        \widehat{T}_{\text{Gaussian}} &= \frac{1}{62}\begin{bmatrix} 1 & 4 & 12 & 25 & 40 & 0 & -40 & -25 & -12 & -4 & -1 \end{bmatrix} \\
        \widehat{R}_{\text{Gaussian}} &= \frac{1}{62}\begin{bmatrix} 1 & 4 & 12 & 20 & 25 & -124 & 25 & 20 & 12 & 4 & 1 \end{bmatrix} \\
    \end{aligned}
    \label{eq:DTGD_Example 5}
\end{equation}

\begin{equation}
    \setcounter{MaxMatrixCols}{35}
    \begin{aligned}
        \widehat{T}_{\text{Gaussian}} &= \frac{1}{131}\!\begin{bmatrix} 1&3&8&18&32&50&64&0&-64&-50&-32&-18&-8&-3&-1 \end{bmatrix} \\
        \widehat{R}_{\text{Gaussian}} &= \frac{1}{178}\begin{bmatrix} 1&3&8&18&32&50&64&-356&64&50&32&18&8&3&1 \end{bmatrix} \\
    \end{aligned}
    \label{eq:DTGD_Example 7}
\end{equation}

We have extended the general difference of a function to the general difference of a sequence, and unified both the computation in a convolution framework, in which the convolution kernel is continuous or discrete relatively. It is feasible to use TGD to characterize the variation of a sequence without knowing its expression of the original function.

\subsection{Directional General Difference of High Dimensional Array}

Similar to the one-dimensional case, we transform the multidimensional array into a continuous function by representing it as a step function, where each constant piece corresponds to a discrete value point. For a $m$-dimensional array $X$, the corresponding step function can be expressed as follows: 
\begin{equation}
    \begin{aligned}
        \hat{f}(x_1, x_2, \dots, x_m) \triangleq X(n_1, n_2, \dots, n_m) \quad\quad x_i \in (n_i-\frac{1}{2}, n_i+\frac{1}{2}] 
    \end{aligned}
    \label{eq:discrete_xn_highD}
\end{equation}
where $n_i \in \mathbb{Z}$. We express the TGD of the array as the TGD (Formula~\eqref{function:realDirectionDerivative} and~\eqref{function:realDirectionDerivative2}) of the corresponding step function (Formula~\eqref{eq:discrete_xn_highD}) at sampling locations: 
\begin{equation}
    \begin{aligned}
        \frac{\partial_{\TGD} X}{\partial \mathbf{v}}\left(n_1, n_2, \dots, n_m\right) &\triangleq \frac{\partial_{\TGD} \hat{f}}{\partial \mathbf{v}}\left(n_1, n_2, \dots, n_m\right) \\
        \frac{\partial_{\TGD}^2 X}{\partial \mathbf{v}^2}\left(n_1, n_2, \dots, n_m\right) &\triangleq \frac{\partial_{\TGD}^2 \hat{f}}{\partial \mathbf{v}^2}\left(n_1, n_2, \dots, n_m\right)
    \end{aligned}
    \label{eq:discrete_TGD_highD}
\end{equation}
It is easily obtained that the TGD of a high-dimensional array remains equal to the convolution of the array with high-dimensional discrete TGD operators, obeying \emph{Difference-by-Convolution}:
\begin{equation}
    \begin{aligned}
        \widehat{T}\left(i_1, i_2, \dots, i_m\right) &= \int_{i_m - \frac{1}{2}}^{i_m + \frac{1}{2}} \ldots \int_{i_1 - \frac{1}{2}}^{i_1 + \frac{1}{2}} \widetilde{T}\left(x_1, x_2, \ldots, x_m\right) \mathrm{d} x_1 \ldots \mathrm{d} x_m \\
        \widehat{R}\left(i_1, i_2, \dots, i_m\right) &= \int_{i_m - \frac{1}{2}}^{i_m + \frac{1}{2}} \ldots \int_{i_1 - \frac{1}{2}}^{i_1 + \frac{1}{2}} \widetilde{R}\left(x_1, x_2, \ldots, x_m\right) \mathrm{d} x_1 \ldots \mathrm{d} x_m \\
        i_j &= -N, -N+1, \dots, N-1, N \quad j = 1,2,\dots,m
    \end{aligned}
    \label{eq:DTGD_operators_highD}
\end{equation}
where $\widetilde{T}$ and $\widetilde{R}$ are multivariate 1st- and 2nd-order TGD operators (See Section~\ref{Sec:TGDofMultivariateFunction}), and $N$ is the \emph{kernel size}. And each discrete TGD operator is composed of $(2N+1)^m$ elements. When the kernel size is substantially larger than the sampling step, the interval integral of $\widetilde{T}$ and $\widetilde{R}$ is approximately equal to the discrete sampling of $\widetilde{T}$ and $\widetilde{R}$. Therefore, in practice, the discrete operator weights can be directly obtained by sampling the multivariate TGD operators without the need for intricate interval integration, as long as the sampled values are meaningful.
\begin{equation}
    \begin{aligned}
        \widehat{T}\left(i_1, i_2, \dots, i_m\right) &= \widetilde{T}\left(i_1, i_2, \dots, i_m\right) \\
        \widehat{R}\left(i_1, i_2, \dots, i_m\right) &= \widetilde{R}\left(i_1, i_2, \dots, i_m\right) \\
        \widehat{\text{LoT}}\left(i_1, i_2, \dots, i_m\right) &= \text{LoT}\left(i_1, i_2, \dots, i_m\right)
    \end{aligned}
\end{equation}
These operators are normalized to the sum of $0$, which is derived directly by Formula~\eqref{eq:TGD_highD_0} and~\eqref{eq:TGD_highD_0_2}. This guarantees that the TGD value remains constant at $0$ for constant arrays.
\begin{equation}
    \begin{aligned}
        &\sum_{i_m = -N}^{N}\ldots\sum_{i_1 = -N}^{N}\widehat{T}\left(i_1, i_2, \dots, i_m\right) = 0 \\
        &\sum_{i_m = -N}^{N}\ldots\sum_{i_1 = -N}^{N}\widehat{R}\left(i_1, i_2, \dots, i_m\right) = 0 \\
        &\sum_{i_m = -N}^{N}\ldots\sum_{i_1 = -N}^{N}\widehat{\text{LoT}}\left(i_1, i_2, \dots, i_m\right) = 0
    \end{aligned}
\end{equation}

As discussed in Section~\ref{subsec:Orthogonal construction method}, the orthogonal constructed operators can decouple dimensions to expedite calculations. Without loss of generality, we demonstrate the 2D discrete TGD operators obtained using orthogonal TGD operators, in which the $x$-axis direction is the differential direction.
\begin{equation}
    \setcounter{MaxMatrixCols}{30}
    \begin{aligned}
        \widehat{T}_x &= \begin{bmatrix} 
            {\widehat{w}_{-N,-N}} & {\dots} & {\widehat{w}_{-N, -1}} & 0 & {-\widehat{w}_{-N, 1}} & {\dots} & {-\widehat{w}_{-N, N}} \\
            {\widehat{w}_{-N+1,-N}} & {\dots} & {\widehat{w}_{-N+1, -1}} & 0 & {-\widehat{w}_{-N+1, 1}} & {\dots} & {-\widehat{w}_{-N+1, N}} \\
            {\ldots} & {\ldots} & {\ldots} & \ldots & {\dots} & {\dots} & {\dots} \\
            {\widehat{w}_{-1,-N}} & {\dots} & {\widehat{w}_{-1, -1}} & 0 & {-\widehat{w}_{-1, 1}} & {\dots} & {-\widehat{w}_{-1, N}} \\
            {\widehat{w}_{0,-N}} & {\dots} & {\widehat{w}_{0, -1}} & 0 & {-\widehat{w}_{0, 1}} & {\dots} & {-\widehat{w}_{0, N}} \\
            {\widehat{w}_{1,-N}} & {\dots} & {\widehat{w}_{1, -1}} & 0 & {-\widehat{w}_{1, 1}} & {\dots} & {-\widehat{w}_{1, N}} \\
            {\ldots} & {\ldots} & {\ldots} & \ldots & {\dots} & {\dots} & {\dots} \\
            {\widehat{w}_{N-1,-N}} & {\dots} & {\widehat{w}_{N-1, -1}} & 0 & {-\widehat{w}_{N-1, 1}} & {\dots} & {-\widehat{w}_{N-1, N}} \\
            {\widehat{w}_{N,-N}} & {\dots} & {\widehat{w}_{N, -1}} & 0 & {-\widehat{w}_{N, 1}} & {\dots} & {-\widehat{w}_{N, N}} 
        \end{bmatrix} \\
    \end{aligned}
    \label{eq:DTGD_Operator2D_1st}
\end{equation}

\begin{equation}
    \setcounter{MaxMatrixCols}{30}
    \begin{aligned}
        \widehat{R}_x &= \begin{bmatrix} 
            {\widehat{w}_{-N,-N}} & {\dots} & {\widehat{w}_{-N, -1}} & {-\widehat{w}_{-N, 0}} & {\widehat{w}_{-N, 1}} & {\dots} & {\widehat{w}_{-N, N}} \\
            {\widehat{w}_{-N+1,-N}} & {\dots} & {\widehat{w}_{-N+1, -1}} & {-\widehat{w}_{-N+1, 0}} & {\widehat{w}_{-N+1, 1}} & {\dots} & {\widehat{w}_{-N+1, N}} \\
            {\ldots} & {\ldots} & {\ldots} & {\dots} & {\dots} & {\dots} & {\dots} \\
            {\widehat{w}_{-1,-N}} & {\dots} & {\widehat{w}_{-1, -1}} & {-\widehat{w}_{-1, 0}} & {\widehat{w}_{-1, 1}} & {\dots} & {\widehat{w}_{-1, N}} \\
            {\widehat{w}_{0,-N}} & {\dots} & {\widehat{w}_{0, -1}} & {-\widehat{w}_{0, 0}} & {\widehat{w}_{0, 1}} & {\dots} & {\widehat{w}_{0, N}} \\
            {\widehat{w}_{1,-N}} & {\dots} & {\widehat{w}_{1, -1}} & {-\widehat{w}_{1, 0}} & {\widehat{w}_{1, 1}} & {\dots} & {\widehat{w}_{1, N}} \\
            {\ldots} & {\ldots} & {\ldots} & {\dots} & {\dots} & {\dots} & {\dots} \\
            {\widehat{w}_{N-1,-N}} & {\dots} & {\widehat{w}_{N-1, -1}} & {-\widehat{w}_{N-1, 0}} & {\widehat{w}_{N-1, 1}} & {\dots} & {\widehat{w}_{N-1, N}} \\
            {\widehat{w}_{N,-N}} & {\dots} & {\widehat{w}_{N, -1}} & {-\widehat{w}_{N, 0}} & {\widehat{w}_{N, 1}} & {\dots} & {\widehat{w}_{N, N}} 
        \end{bmatrix}
    \end{aligned}
    \label{eq:DTGD_Operator2D_2nd}
\end{equation}
where:
\begin{equation}
    \begin{aligned}
        & \widehat{w}_{-i,-j} = \widehat{w}_{-i,j} = \widehat{w}_{i,-j} = \widehat{w}_{i,j} > 0 \\
        & \widehat{w}_{i,|j|} > \widehat{w}_{i,|k|} \quad 0 < |j| < |k| \leq N \\
        & \widehat{w}_{|i|,j} > \widehat{w}_{|k|,j} \quad 0 < |i| < |k| \leq N \\
        & \widehat{w}_{i, 0} = \sum_{j=-N}^{-1} \widehat{R}_x(i, j) + \sum_{j=1}^{N} \widehat{R}_x(i, j) \\
    \end{aligned}
    \label{eq:DTGD_Operator2D_Constraint}
\end{equation}

And the unified form of the 2D discrete LoT operators are:
\begin{equation}
   \setcounter{MaxMatrixCols}{30}
   \begin{aligned}
       \widehat{\text{LoT}} &= \begin{bmatrix} 
           {\widehat{w}_{-N,-N}} & {\dots} & {\widehat{w}_{-N, -1}} & {\widehat{w}_{-N, 0}} & {\widehat{w}_{-N, 1}} & {\dots} & {\widehat{w}_{-N, N}} \\
           {\widehat{w}_{-N+1,-N}} & {\dots} & {\widehat{w}_{-N+1, -1}} & {\widehat{w}_{-N+1, 0}} & {\widehat{w}_{-N+1, 1}} & {\dots} & {\widehat{w}_{-N+1, N}} \\
           {\ldots} & {\ldots} & {\ldots} & {\dots} & {\dots} & {\dots} & {\dots} \\
           {\widehat{w}_{-1,-N}} & {\dots} & {\widehat{w}_{-1, -1}} & {\widehat{w}_{-1, 0}} & {\widehat{w}_{-1, 1}} & {\dots} & {\widehat{w}_{-1, N}} \\
           {\widehat{w}_{0,-N}} & {\dots} & {\widehat{w}_{0, -1}} & {-\widehat{w}_{0, 0}} & {\widehat{w}_{0, 1}} & {\dots} & {\widehat{w}_{0, N}} \\
           {\widehat{w}_{1,-N}} & {\dots} & {\widehat{w}_{1, -1}} & {\widehat{w}_{1, 0}} & {\widehat{w}_{1, 1}} & {\dots} & {\widehat{w}_{1, N}} \\
           {\ldots} & {\ldots} & {\ldots} & {\dots} & {\dots} & {\dots} & {\dots} \\
           {\widehat{w}_{N-1,-N}} & {\dots} & {\widehat{w}_{N-1, -1}} & {\widehat{w}_{N-1, 0}} & {\widehat{w}_{N-1, 1}} & {\dots} & {\widehat{w}_{N-1, N}} \\
           {\widehat{w}_{N,-N}} & {\dots} & {\widehat{w}_{N, -1}} & {\widehat{w}_{N, 0}} & {\widehat{w}_{N, 1}} & {\dots} & {\widehat{w}_{N, N}} 
       \end{bmatrix}
   \end{aligned}
   \label{eq:DTGD_Operator2D_LoT}
\end{equation}
where:
\begin{equation}
   \begin{aligned}
       & \widehat{w}_{-i,-j} = \widehat{w}_{-i,j} = \widehat{w}_{i,-j} = \widehat{w}_{i,j} > 0 \\
       & \widehat{w}_{i,|j|} > \widehat{w}_{i,|k|} \quad 0 < |j| < |k| \leq N \\
       & \widehat{w}_{|i|,j} > \widehat{w}_{|k|,j} \quad 0 < |i| < |k| \leq N \\
       & \sum_{i = -N}^{N}\sum_{j = -N}^{N}\widehat{\text{LoT}}(i, j) = 0 \\
   \end{aligned}
   \label{eq:DTGD_LoT_Constraint}
\end{equation}

In addition, by choosing different multivariate TGD operators, high-dimensional discrete arrays can be calculated for directional TGD, partial TGD, and Laplace TGD. Generally, the classic discrete differential operators, such as the Sobel operator~\cite{DBLP:books/lib/DudaH73}, Sobel–Feldman operator~\cite{DBLP:conf/ijcai/FeldmanFFGPST69}, Scharr operator~\cite{DBLP:journals/ai/Rosenfeld00}, and Laplacian operator~\cite{acharya2005image} in signal processing, can be derived from the Formula~\eqref{eq:DTGD_Operator2D_1st},~\eqref{eq:DTGD_Operator2D_Constraint} for $N=1$. However, these operators are not discrete TGD operators, and lack the properties of TGD operators such as noise suppression.

Here, we present examples of 2D discrete TGD operators orthogonal constructed with the Gaussian kernel function and kernel size $N=6$, including 1st- and 2nd-order discrete TGD operators ($\widehat{T}_x$ and $\widehat{R}_x$) in the $x$-axis direction. We obtain them by sampling and normalizing the orthogonal TGD operators, respectively. Additionally, we provide two examples of discrete TGD operators oriented 45 degrees from the $x$-axis ($\widehat{T}_{45'}$ and $\widehat{R}_{45'}$). An anisotropic discrete LoT operator ($\widehat{\text{LoT}}$) is also provided. The operator values can be found in Formula~\eqref{eq:DTGD_2D_values1} to~\eqref{eq:DTGD_2D_LoT}. By adjusting the kernel function and its parameters, a wide variety of discrete TGD operators can be constructed to satisfy varying task requirements.

More intuitively, Figure~\ref{fig:DTGD_2D_examples} presents a visual representation of the above four operators, in which red and blue respectively indicate positive and negative values. The darker the color, the larger the absolute value. Figure~\ref{fig:DTGD_3D_examples} provides a schematic diagram of 3D discrete TGD operators with a size of $15 \times 15 \times 15$ ($N=7$), which are used for partial TGD calculation when the $z$-variable is involved. These discrete TGD operators can be used as convolution kernels for discrete TGD calculation, as well as serve downstream tasks including video or image sequence gradient calculations in computer vision. And the excellent performance of discrete TGD operators will be verified in subsequent experiments.

\begin{equation}
    \setcounter{MaxMatrixCols}{30}
    \begin{aligned}
        \widehat{T}_x \!&=\! \small{\setlength{\arraycolsep}{2.5pt}
        \begin{bmatrix} 
            1 & 3 & 7 & 16 & 27 & 37 & 0 & -37 & -27 & -16 & -7 & -3 & -1 \\
            3 & 11 & 28 & 58 & 98 & 133 & 0 & -133 & -98 & -58 & -28 & -11 & -3 \\
            10 & 33 & 84 & 174 & 292 & 398 & 0 & -398 & -292 & -174 & -84 & -33 & -19 \\
            26 & 83 & 211 & 435 & 729 & 994 & 0 & -994 & -729 & -435 & -211 & -83 & -26 \\
            56 & 174 & 442 & 912 & 1529 & 2085 & 0 & -2085 & -1529 & -912 & -442 & -174 & -56 \\
            99 & 311 & 788 & 1625 & 2725 & 3715 & 0 & -3715 & -2725 & -1625 & -788 & -311 & -99 \\
            153 & 478 & 1211 & 2496 & 4184 & 5705 & 0 & -5705 & -4184 & -2496 & -1211 & -478 & -153 \\
            99 & 311 & 788 & 1625 & 2725 & 3715 & 0 & -3715 & -2725 & -1625 & -788 & -311 & -99 \\
            56 & 174 & 442 & 912 & 1529 & 2085 & 0 & -2085 & -1529 & -912 & -442 & -174 & -56 \\
            26 & 83 & 211 & 435 & 729 & 994 & 0 & -994 & -729 & -435 & -211 & -83 & -26 \\
            10 & 33 & 84 & 174 & 292 & 398 & 0 & -398 & -292 & -174 & -84 & -33 & -19 \\
            3 & 11 & 28 & 58 & 98 & 133 & 0 & -133 & -98 & -58 & -28 & -11 & -3 \\
            1 & 3 & 7 & 16 & 27 & 37 & 0 & -37 & -27 & -16 & -7 & -3 & -1 \\
        \end{bmatrix}}
    \end{aligned}
    \label{eq:DTGD_2D_values1}
\end{equation}
\begin{equation}
    \setcounter{MaxMatrixCols}{30}
    \begin{aligned}
        \widehat{R}_x \!&=\! \small{\setlength{\arraycolsep}{2.5pt}
        \begin{bmatrix} 
            1 & 3 & 7 & 16 & 27 & 37 & -182 & 37 & 27 & 16 & 7 & 3 & 1 \\
            3 & 11 & 28 & 58 & 98 & 133 & -662 & 133 & 98 & 58 & 28 & 11 & 3 \\
            10 & 33 & 84 & 174 & 292 & 398 & -1982 & 398 & 292 & 174 & 84 & 33 & 19 \\
            26 & 83 & 211 & 435 & 729 & 994 & -4956 & 994 & 729 & 435 & 211 & 83 & 26 \\
            56 & 174 & 442 & 912 & 1529 & 2085 & -10396 & 2085 & 1529 & 912 & 442 & 174 & 56 \\
            99 & 311 & 788 & 1625 & 2725 & 3715 & -18526 & 3715 & 2725 & 1625 & 788 & 311 & 99 \\
            153 & 478 & 1211 & 2496 & 4184 & 5705 & -28454 & 5705 & 4184 & 2496 & 1211 & 478 & 153 \\
            99 & 311 & 788 & 1625 & 2725 & 3715 & -18526 & 3715 & 2725 & 1625 & 788 & 311 & 99 \\
            56 & 174 & 442 & 912 & 1529 & 2085 & -10396 & 2085 & 1529 & 912 & 442 & 174 & 56 \\
            26 & 83 & 211 & 435 & 729 & 994 & -4956 & 994 & 729 & 435 & 211 & 83 & 26 \\
            10 & 33 & 84 & 174 & 292 & 398 & -1982 & 398 & 292 & 174 & 84 & 33 & 19 \\
            3 & 11 & 28 & 58 & 98 & 133 & -662 & 133 & 98 & 58 & 28 & 11 & 3 \\
            1 & 3 & 7 & 16 & 27 & 37 & -182 & 37 & 27 & 16 & 7 & 3 & 1 \\
        \end{bmatrix}}
    \end{aligned}
\end{equation}
\begin{equation}
    \setcounter{MaxMatrixCols}{30}
    \begin{aligned}
        \widehat{T}_{45'} \!&=\! \small{\setlength{\arraycolsep}{2.5pt}
        \begin{bmatrix} 
            0 & -1 & -2 & -6 & -11 & -17 & -21 & -20 & -18 & -12 & -7 & -3 & -1 \\
            1 & 0 & -9 & -21 & -39 & -60 & -75 & -78 & -67 & -48 & -28 & -14 & -3 \\
            2 & 9 & 0 & -59 & -111 & -172 & -220 & -233 & -204 & -150 & -92 & -28 & -7 \\
            6 & 21 & 59 & 0 & -261 & -410 & -532 & -575 & -518 & -392 & -150 & -48 & -12 \\
            11 & 39 & 111 & 261 & 0 & -805 & -1069 & -1184 & -1101 & -518 & -204 & -67 & -18 \\
            17 & 60 & 172 & 410 & 805 & 0 & -1791 & -2047 & -1184 & -575 & -233 & -78 & -20 \\
            21 & 75 & 220 & 532 & 1069 & 1791 & 0 & -1791 & -1069 & -532 & -220 & -75 & -21 \\
            20 & 78 & 233 & 575 & 1184 & 2047 & 1791 & 0 & -805 & -410 & -172 & -60 & -17 \\
            18 & 67 & 204 & 518 & 1101 & 1184 & 1069 & 805 & 0 & -261 & -111 & -39 & -11 \\
            12 & 48 & 150 & 392 & 518 & 575 & 532 & 410 & 261 & 0 & -59 & -21 & -6 \\
            7 & 28 & 92 & 150 & 204 & 233 & 220 & 172 & 111 & 59 & 0 & -9 & -2 \\
            3 & 14 & 28 & 48 & 67 & 78 & 75 & 60 & 39 & 21 & 9 & 0 & -1 \\
            1 & 3 & 7 & 12 & 18 & 20 & 21 & 17 & 11 & 6 & 2 & 1 & 0 \\
        \end{bmatrix}}
    \end{aligned}
\end{equation}
\begin{equation}
    \setcounter{MaxMatrixCols}{30}
    \begin{aligned}
        \widehat{R}_{45'} \!&=\! \small{\setlength{\arraycolsep}{1pt}
        \begin{bmatrix} 
            -1 & 1 & 2 & 6 & 11 & 17 & 21 & 20 & 18 & 12 & 7 & 3 & 1 \\
            1 & -22 & 9 & 21 & 39 & 60 & 75 & 78 & 67 & 48 & 28 & 14 & 3 \\
            2 & 9 & -171 & 59 & 111 & 172 & 220 & 233 & 204 & 150 & 92 & 28 & 7 \\
            6 & 21 & 59 & -903 & 261 & 410 & 532 & 575 & 518 & 392 & 150 & 48 & 12 \\
            11 & 39 & 111 & 261 & -3318 & 805 & 1069 & 1184 & 1101 & 518 & 204 & 67 & 18 \\
            17 & 60 & 172 & 410 & 805 & -8656 & 1791 & 2047 & 1184 & 575 & 233 & 78 & 20 \\
            21 & 75 & 220 & 532 & 1069 & 1791 & -16500 & 1791 & 1069 & 532 & 220 & 75 & 21 \\
            20 & 78 & 233 & 575 & 1184 & 2047 & 1791 & -8656 & 805 & 410 & 172 & 60 & 17 \\
            18 & 67 & 204 & 518 & 1101 & 1184 & 1069 & 805 & -3318 & 261 & 111 & 39 & 11 \\
            12 & 48 & 150 & 392 & 518 & 575 & 532 & 410 & 261 & -903 & 59 & 21 & 6 \\
            7 & 28 & 92 & 150 & 204 & 233 & 220 & 172 & 111 & 59 & -171 & 9 & 2 \\
            3 & 14 & 28 & 48 & 67 & 78 & 75 & 60 & 39 & 21 & 9 & -22 & 1 \\
            1 & 3 & 7 & 12 & 18 & 20 & 21 & 17 & 11 & 6 & 2 & 1 & -1 \\
        \end{bmatrix}}
    \end{aligned}
    \label{eq:DTGD_2D_values5}
\end{equation}
\begin{equation}
   \setcounter{MaxMatrixCols}{30}
   \begin{aligned}
       \widehat{\text{LoT}} \!&=\! \small{\setlength{\arraycolsep}{2.5pt}
       \begin{bmatrix} 
           1 & 3 & 7 & 16 & 27 & 37 & 41 & 37 & 27 & 16 & 7 & 3 & 1 \\
           3 & 9 & 24 & 50 & 84 & 115 & 128 & 115 & 84 & 50 & 24 & 9 & 3 \\
           7 & 24 & 62 & 128 & 215 & 293 & 325 & 293 & 215 & 128 & 62 & 24 & 7 \\
           16 & 50 & 128 & 264 & 443 & 604 & 670 & 604 & 443 & 264 & 128 & 50 & 16 \\
           27 & 84 & 215 & 443 & 743 & 1013 & 1124 & 1013 & 743 & 443 & 215 & 84 & 27 \\
           37 & 115 & 293 & 604 & 1013 & 1382 & 1532 & 1382 & 1013 & 604 & 293 & 115 & 37 \\
           41 & 128 & 325 & 670 & 1124 & 1532 & -49596 & 1532 & 1124 & 670 & 325 & 128 & 41 \\
           37 & 115 & 293 & 604 & 1013 & 1382 & 1532 & 1382 & 1013 & 604 & 293 & 115 & 37 \\
           27 & 84 & 215 & 443 & 743 & 1013 & 1124 & 1013 & 743 & 443 & 215 & 84 & 27 \\
           16 & 50 & 128 & 264 & 443 & 604 & 670 & 604 & 443 & 264 & 128 & 50 & 16 \\
           7 & 24 & 62 & 128 & 215 & 293 & 325 & 293 & 215 & 128 & 62 & 24 & 7 \\
           3 & 9 & 24 & 50 & 84 & 115 & 128 & 115 & 84 & 50 & 24 & 9 & 3 \\
           1 & 3 & 7 & 16 & 27 & 37 & 41 & 37 & 27 & 16 & 7 & 3 & 1 \\
       \end{bmatrix}}
   \end{aligned}
   \label{eq:DTGD_2D_LoT}
\end{equation}

\begin{figure}[!htb]    	
    \centering    	
    \subfigure[$\widehat{T}_{x}$]{  			 
        \includegraphics[width=0.36\linewidth]{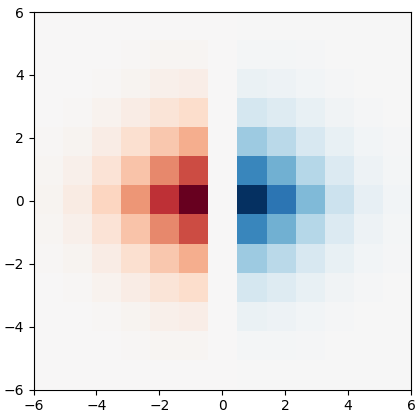}
    }\hspace{10mm}    	 
    \subfigure[$\widehat{R}_{x}$]{   	 		 
        \includegraphics[width=0.36\linewidth]{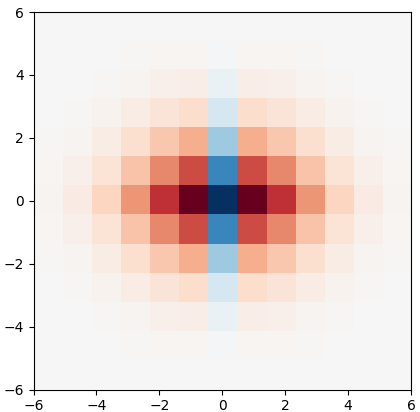}
    }    	 
    \subfigure[$\widehat{T}_{45'}$]{   	 		 
        \includegraphics[width=0.36\linewidth]{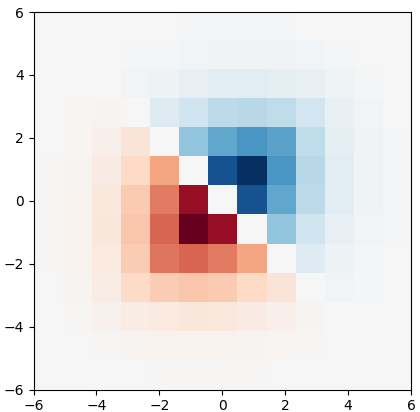}
    }\hspace{10mm}  
    \subfigure[$\widehat{R}_{45'}$]{  			 
        \includegraphics[width=0.36\linewidth]{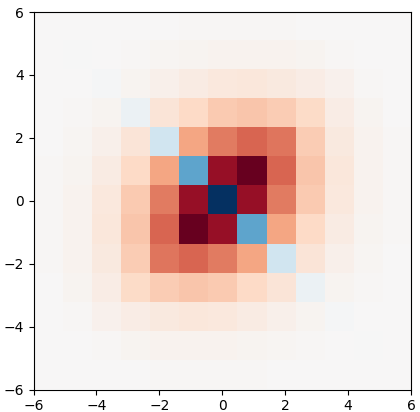}
    }   
    \subfigure[$\widehat{\text{LoT}}$]{  			 
        \includegraphics[width=0.36\linewidth]{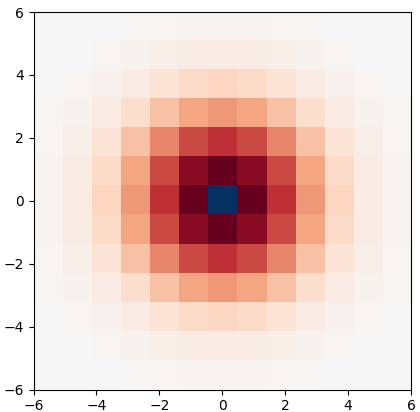}
    } 
    \caption{Visualization of a 2D discrete TGD operators. The operator values can be found in Formula~\eqref{eq:DTGD_2D_values1} to~\eqref{eq:DTGD_2D_LoT}.}  
    \label{fig:DTGD_2D_examples} 
\end{figure}

\begin{figure}[!htb]    	
    \centering    	
    \subfigure[3D 1st-order discrete TGD operator $\widehat{T}_{z}$]{  			 
        \includegraphics[width=0.95\linewidth]{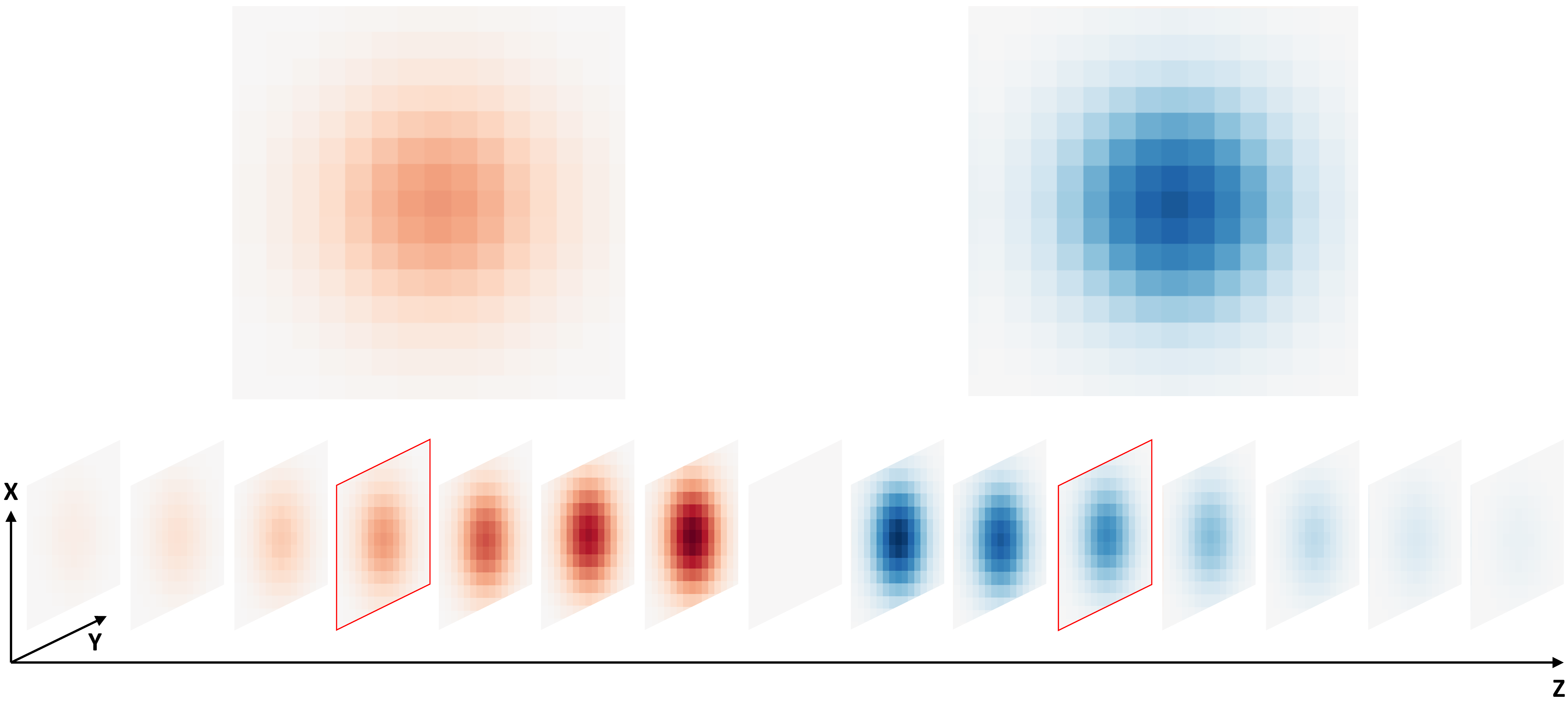}
    }    	 
    \subfigure[3D 2nd-order discrete TGD operator $\widehat{R}_{z}$]{   	 		 
        \includegraphics[width=0.95\linewidth]{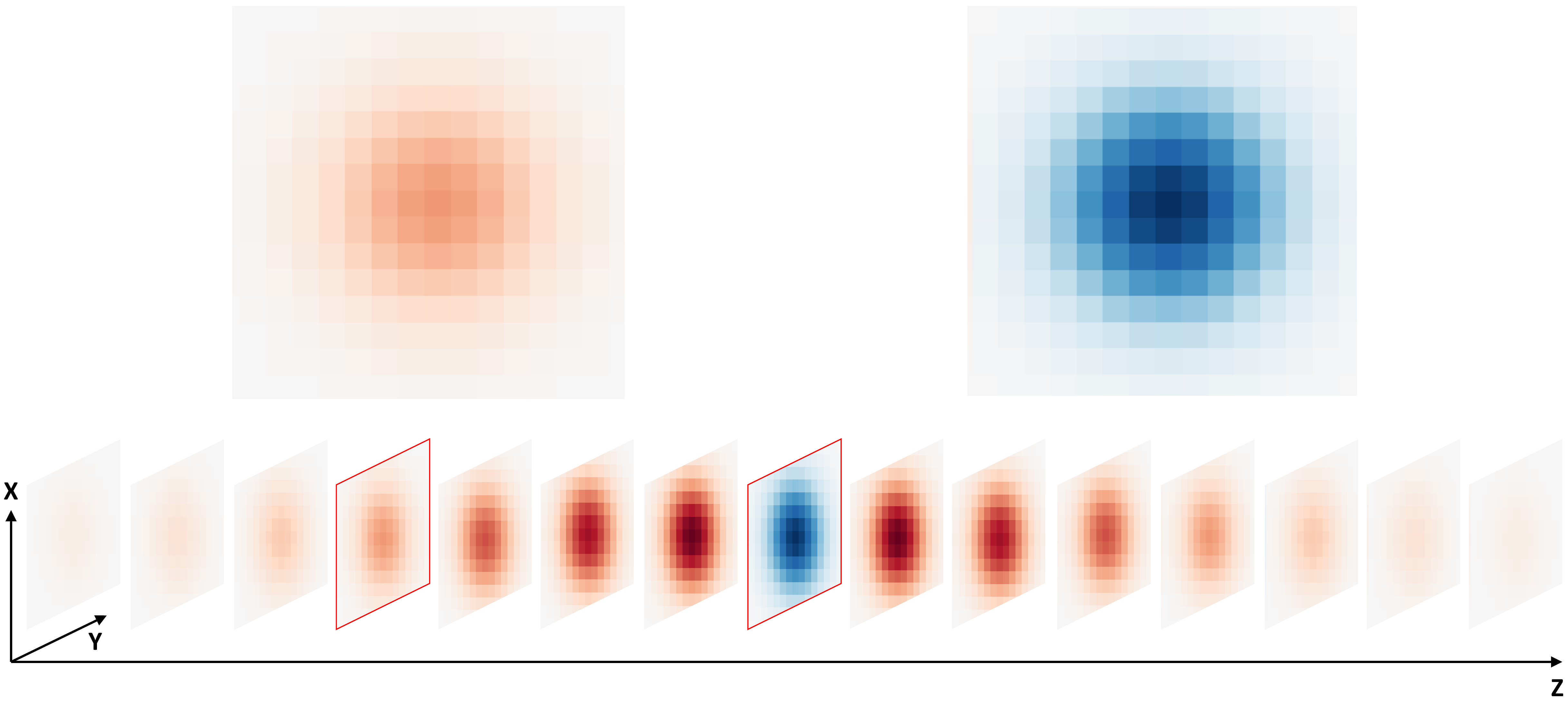}
    }	 
    \caption{Visualization of two 3D discrete TGD operators: 1st- and 2nd-order discrete TGD operators in the $z-axis$ direction. The operators are orthogonal constructed with the Gaussian kernel function. The red box marks the cross-section of the operator for the enlarged visualization.}  
    \label{fig:DTGD_3D_examples} 
\end{figure}

\clearpage
\newpage

\section{Examples and Analysis}
\subsection{Kernel Functions and Operators}

As we have mentioned in Section~\ref{subsec:1Dwindowderivativedefinitions}, various functions meet the three constraints for constructing TGD. We have used the Gaussian function as an example, and introduce four more functions. Meanwhile, we also demonstrate two counterexamples in this section. 

\subsubsection{Linear Kernel Function}
\label{subsubsec:LinearKernel}


Linear function is the simplest kernel function (Formula~\eqref{eq:linear}), which is controlled by the parameters $k$ and $c$. The constraints are satisfied when $k>0$ and $c<0$. And $c \lesssim -kW$, which forces the function value converging to $0$ at the window boundary.
\begin{equation}
  \begin{aligned}
    \text{Line}(x) = -kx - c, \quad x\in(0, W].
  \end{aligned}
  \label{eq:linear}
\end{equation}

By normalization we have:
\begin{equation*}
    \int_{0}^{W}(-kx - c)\d{x} = 1.
\end{equation*}

The second-order TGD operator constructed with the linear kernel function is:
\begin{equation}
  \begin{aligned}
    R_{\text{line}}(x) = \left\{\begin{array}{cc}
        -k|x|-c& \quad\quad x \in [-W,0) \cup (0,W]\\
        -2\delta(0)\int_{0}^{W}(-kx-c)\d{x}& \quad\quad x=0 
    \end{array}\right.
  \end{aligned}
\end{equation}

By integrating over the second-order TGD operator, the relative smooth operator is polynomial (Formula~\eqref{eq:linearfilter}), where $p\gtrsim \frac{k}{6}W^3+\frac{c}{2}W^2$, which forces the filter function converging to 0 at the window boundary.
\begin{equation}
  \begin{aligned}
    S_{\text{line}}(x) = \left\{\begin{array}{cc}
        \frac{k}{6}x^3-\frac{c}{2}x^2+p& \quad\quad x \in [-W,0)\\
        p& \quad\quad x=0 \\
        -\frac{k}{6}x^3-\frac{c}{2}x^2+p& \quad\quad x \in (0,W]
    \end{array}\right.
  \end{aligned}
  \label{eq:linearfilter}
\end{equation}

Figure~\ref{LineConstructor} shows, from left to right, the smooth operator and the second-order TGD operators constructed on the linear kernel function.

\begin{figure}[htb]
  \centering    	
    \subfigure[Smooth operator]{   	 		 
        \includegraphics[width=0.31\linewidth]{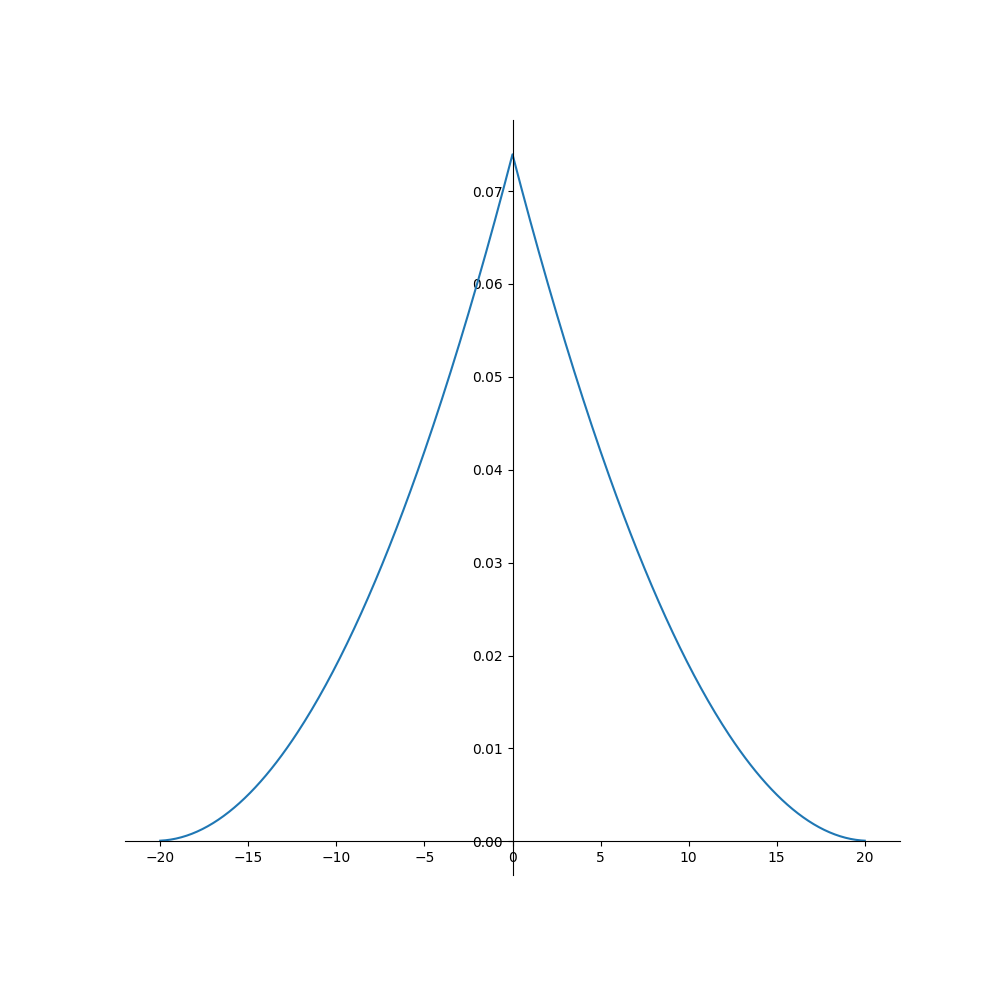}
    }
    \subfigure[Second-order TGD operator]{   	 		 
        \includegraphics[width=0.31\linewidth]{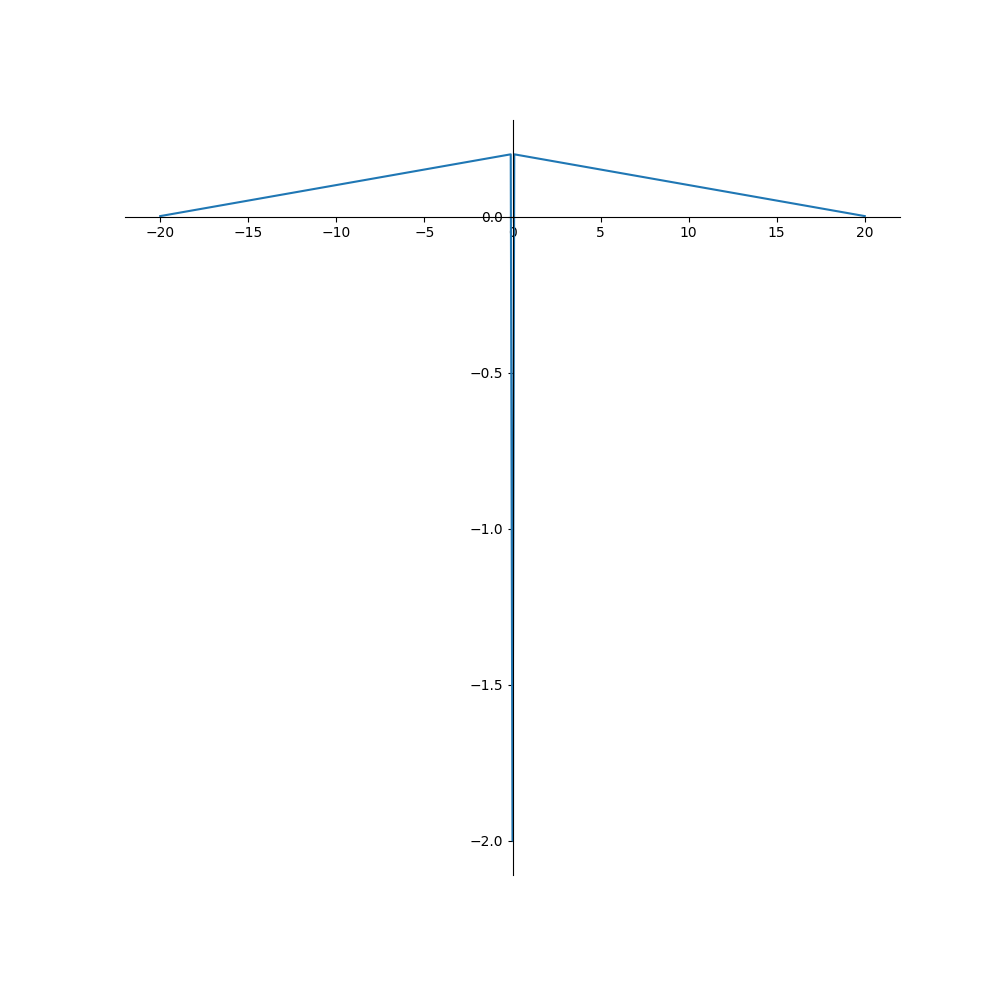}
    }
  \caption{
    The smooth operator and second-order TGD operator constructed based on the linear kernel function.
  }
  \label{LineConstructor}
\end{figure}

\subsubsection{Exponential Kernel Function}
\label{subsubsec:ExponentialKernel}
The exponential function can also be used as a kernel function, which is controlled by the parameters $k$ and $\delta$, where $k>0, \delta>0$ (Formula~\eqref{eq:exponential}). It has a good property that the exponential kernel function is the same family as its integral, which means the relative smooth operator is also exponential.
\begin{equation}
  \begin{aligned}
    \text{Exp}(x) = ke^{-|x|/\delta^2}, \quad x\in(0, W].
  \end{aligned}
  \label{eq:exponential}
\end{equation}

Coefficient $k$ controls normalization, and by normalization we have:
\begin{equation*}
    \int_{0}^{W}\left(ke^{-x/\delta^2}\right)\d{x} = 1.
\end{equation*}

The first- and second-order TGD operators constructed with exponential function are:
\begin{equation}
  \begin{aligned}
    T_{\text{exp}}(x) = \left\{\begin{array}{cc}
        ke^{-|x|/\delta^2}& \quad\quad x \in [-W,0)\\
        0& \quad\quad x=0 \\
        -ke^{-|x|/\delta^2}& \quad\quad x\in(0, W]
    \end{array}\right.
  \end{aligned}
\end{equation}

\begin{equation}
  \begin{aligned}
    R_{\text{exp}}(x) = \left\{\begin{array}{cc}
        ke^{-|x|/\delta^2}& \quad\quad x \in [-W,0) \cup (0,W]\\
        -2\delta(0)\int_{0}^{W}( ke^{-x/\delta^2})\d{x}& \quad\quad x=0 
    \end{array}\right.
  \end{aligned}
\end{equation}

The integral of an exponential function is still an exponential function, which follows that the corresponding smooth operator is exponential, as follows, where $a, b, p$ are the parameters.
\begin{equation}
  \begin{aligned}
    S_{\text{exp}}(x) = \left\{\begin{array}{cc}
        ae^{-b|x|}+p& \quad\quad x \in [-W,0)\\
        a+p& \quad\quad x=0 \\
        ae^{-b|x|}+p& \quad\quad x \in (0,W]
    \end{array}\right.
  \end{aligned}
\end{equation}

Figure~\ref{ExpConstructor} shows, from left to right, the smooth operator, the first- and second-order TGD operators constructed on the exponential kernel function.

\begin{figure}[htb]
  \subfigure[Smooth operator]{   	 		 
        \includegraphics[width=0.31\linewidth]{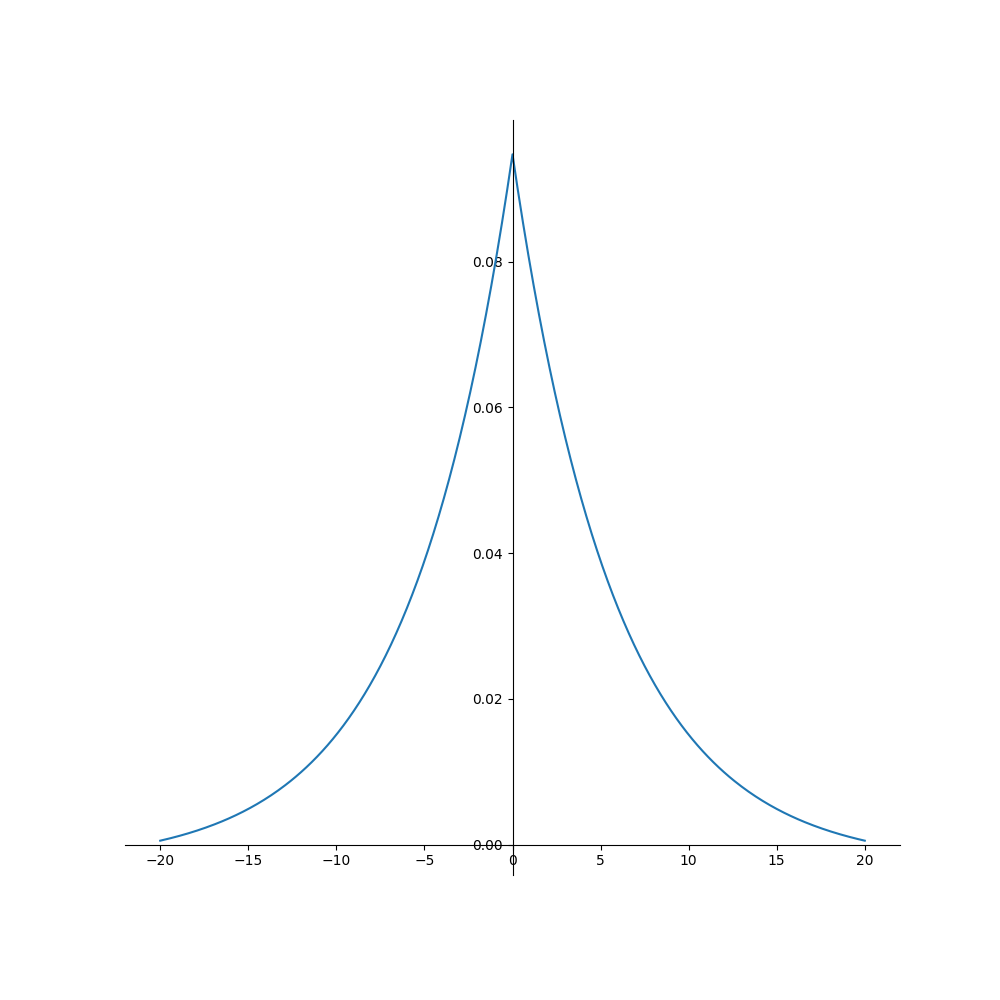}
    }
    \subfigure[First-order TGD operator]{   	 		 
        \includegraphics[width=0.31\linewidth]{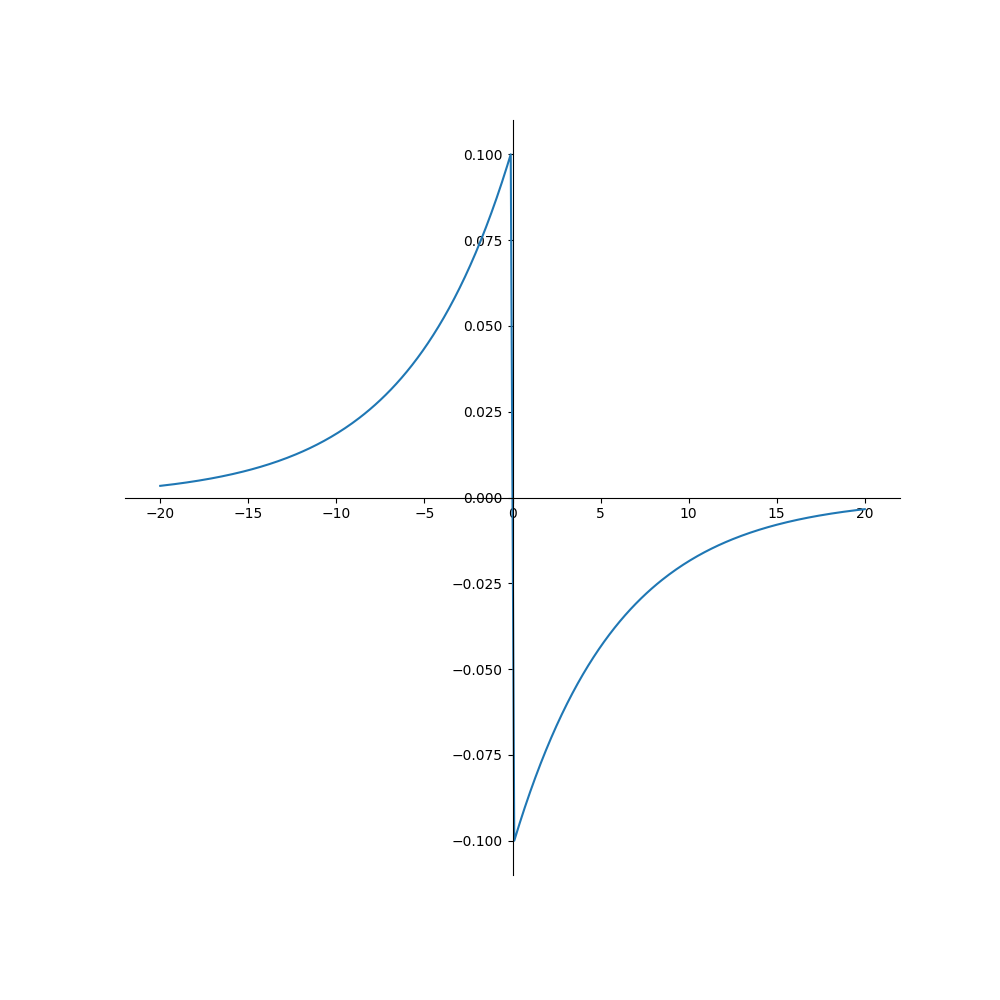}
    }   
    \subfigure[Second-order TGD operator]{   	 		 
        \includegraphics[width=0.31\linewidth]{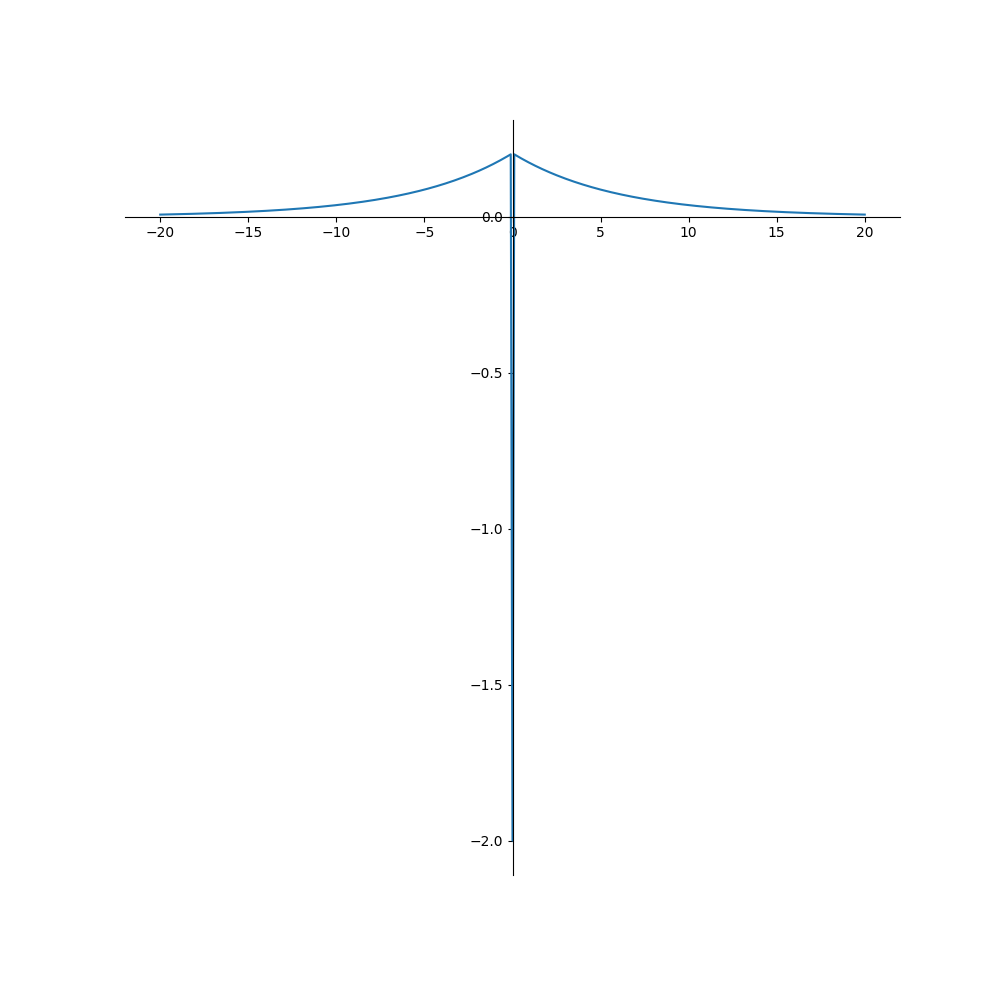}
    }
  \caption{
    The smooth operator, first- and second-order TGD operators constructed based on the exponential kernel function.
  }
  \label{ExpConstructor}
\end{figure}

\subsubsection{Landau Kernel Function}

In fact, many classic kernel functions can be used in TGD. Landau kernel function is a classic function used to prove Weierstrass approximation theorem~\cite{carmignani1977distributional}. In Landau function (Formula~\eqref{eq:Landau}), $c,\tau,n$ are the parameters, in which $\tau \gtrsim W$ is designed to ensure the function follows the constraints of TGD.
\begin{equation}
  \begin{aligned}
    \text{Landau}(x) = c(1-x^2/\tau^2)^n, \quad x\in(0, W].
  \end{aligned}
  \label{eq:Landau}
\end{equation}

Coefficient $c$ controls normalization. And by normalization we have:
\begin{equation*}
    \int_{0}^{W}\left(c\left(1-x^2/\tau^2\right)^n\right)\d{x} = 1.
\end{equation*}

The first- and second-order TGD operators constructed with Landau kernel function are:
\begin{equation}
  \begin{aligned}
    T_{\text{landau}}(x) = \left\{\begin{array}{cc}
        c(1-x^2/\tau^2)^n& \quad\quad x \in [-W,0)\\
        0& \quad\quad x=0 \\
        -c(1-x^2/\tau^2)^n& \quad\quad x\in(0, W]
    \end{array}\right.
  \end{aligned}
\end{equation}

\begin{equation}
  \begin{aligned}
    R_{\text{landau}}(x) = \left\{\begin{array}{cc}
        c(1-x^2/\tau^2)^n& \quad\quad x \in [-W,0) \cup (0,W]\\
        -2\delta(0)\int_{0}^{W}(c(1-x^2/\tau^2)^n)\d{x}& \quad\quad x=0 
    \end{array}\right.
  \end{aligned}
\end{equation}

Although it is difficult to write expressions for the integration of Landau kernel, we can approximate the smooth operator by using numerical integration and symmetric transformation.
\begin{equation}
  \begin{aligned}
    S_{\text{landau}}(x) = \left\{\begin{array}{cc}
        \int_{-W}^{x}(c(1-t^2/\tau^2)^n)\d{t} & \quad\quad x \in [-W,0]\\
        S_{\text{landau}}(-x)& \quad\quad x \in (0,W]
    \end{array}\right.
  \end{aligned}
\end{equation}

Figure~\ref{LandauConstructor} shows, from left to right, the smooth operator, the first- and second-order TGD operators constructed on the Landau kernel function.

\begin{figure}[htb]
  \subfigure[Smooth operator]{   	 		 
        \includegraphics[width=0.31\linewidth]{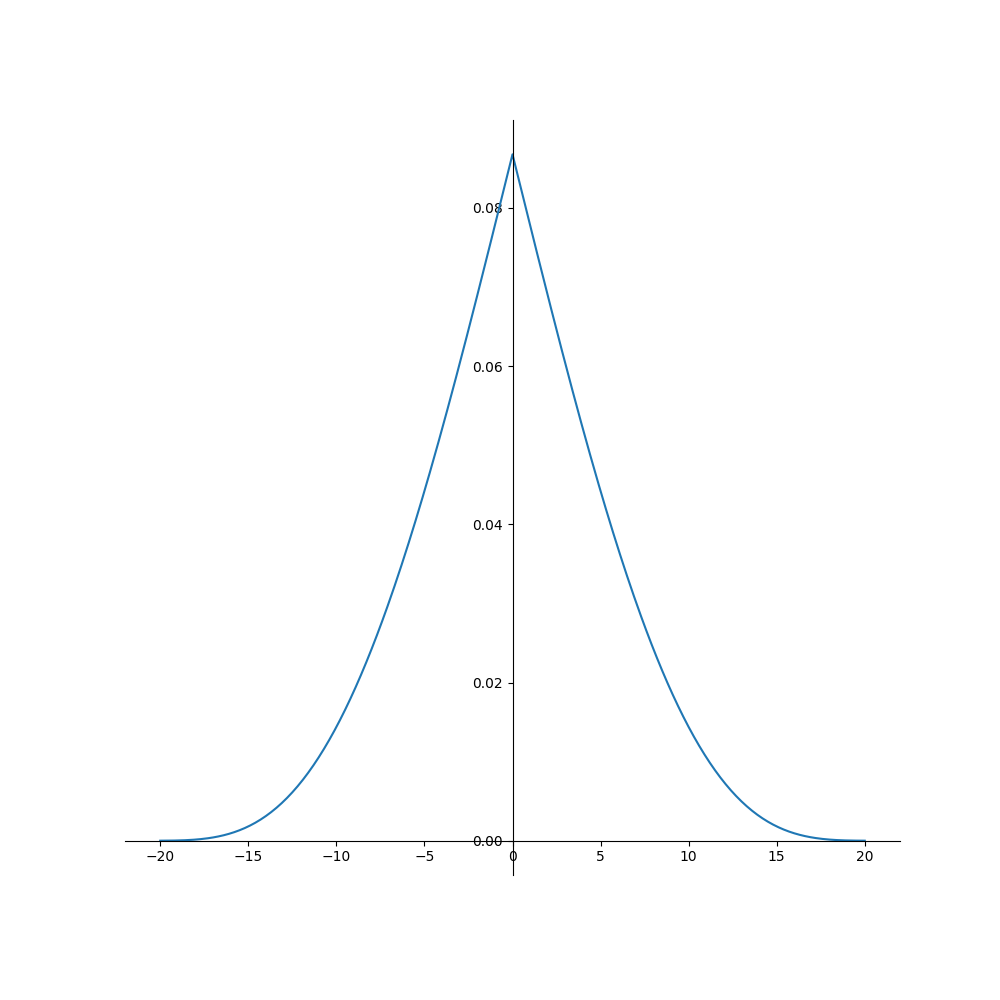}
    }
    \subfigure[First-order TGD operator]{   	 		 
        \includegraphics[width=0.31\linewidth]{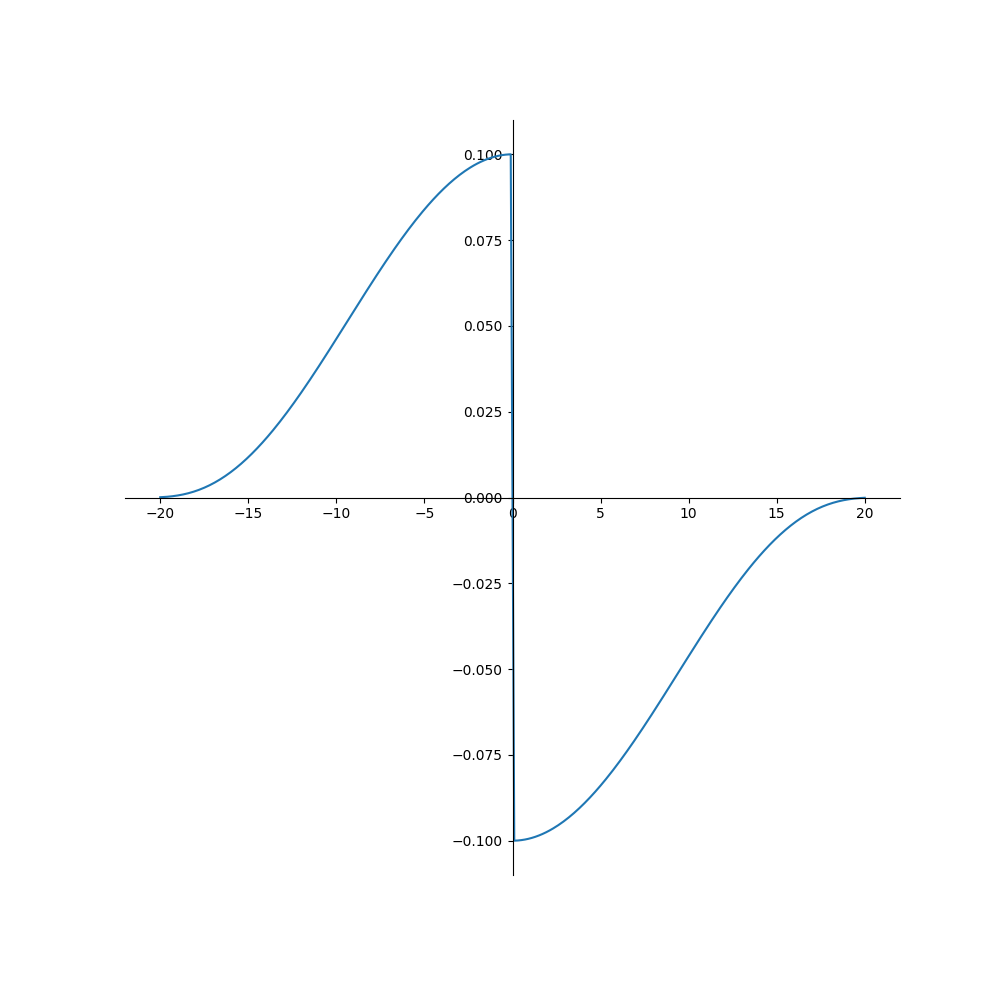}
    }   
    \subfigure[Second-order TGD operator]{   	 		 
        \includegraphics[width=0.31\linewidth]{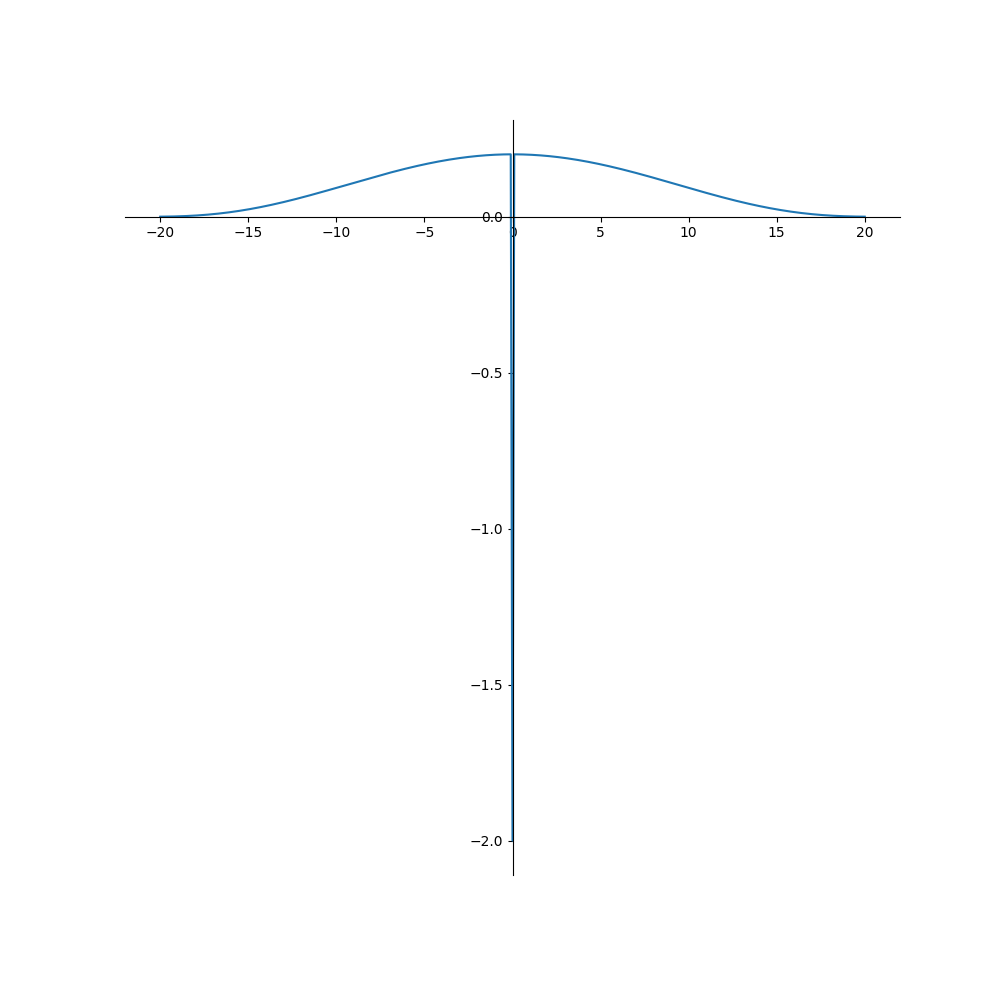}
    }
  \caption{
    The smooth operator, first- and second-order TGD operators constructed based on the Landau kernel function.
  }
  \label{LandauConstructor}
\end{figure}

\subsubsection{Weibull Kernel Function}
In addition to the aforementioned kernel functions, TGD can also use probability density functions as kernel functions, which are usually integrable and finite in energy. Besides Gaussian, we provide an additional example, the Weibull distribution function~\cite{hallinan1993review, rinne2008weibull}, with its parameters $k$ and $\lambda$. When $k\in(0,1]$, Weibull distribution functions satisfy the constraints and are used as the kernel function of TGD.
\begin{equation}
  \begin{aligned}
    \text{Weibull}(x) = \frac{ck}{\lambda}\left(\frac{x}{\lambda}\right)^{k-1} e^{-\left(\frac{x}{\lambda}\right)^k}, \quad x\in(0, W], k \in (0,1].
  \end{aligned}
\end{equation}

Coefficient $c$ controls normalization. And by normalization we have:
\begin{equation*}
    \int_{0}^{W}\left(\frac{ck}{\lambda}\left(\frac{x}{\lambda}\right)^{k-1} e^{-\left(\frac{x}{\lambda}\right)^k}\right)\d{x} = 1.
\end{equation*}

The first- and second-order TGD operators constructed with Weibull distribution function are:
\begin{equation}
  \begin{aligned}
    T_{\text{weibull}}(x) = \left\{\begin{array}{cc}
        \frac{ck}{\lambda}(-\frac{x}{\lambda})^{k-1} e^{-(-\frac{x}{\lambda})^k}& \quad\quad x \in [-W,0)\\
        0& \quad\quad x=0 \\
        -\frac{ck}{\lambda}(\frac{x}{\lambda})^{k-1} e^{-(\frac{x}{\lambda})^k}& \quad\quad x\in(0, W]
    \end{array}\right.
  \end{aligned}
\end{equation}

\begin{equation}
  \begin{aligned}
    R_{\text{weibull}}(x) = \left\{\begin{array}{cc}
        \frac{ck}{\lambda}(\frac{|x|}{\lambda})^{k-1} e^{-(\frac{|x|}{\lambda})^k}& \quad\quad x \in [-W,0) \cup (0,W]\\
        -2\delta(0)\int_{0}^{W}(\frac{ck}{\lambda}(\frac{x}{\lambda})^{k-1} e^{-(\frac{x}{\lambda})^k})\d{x}& \quad\quad x=0 
    \end{array}\right.
  \end{aligned}
\end{equation}

The smooth operator can be obtained by approximation using numerical integration and symmetric transformation.
\begin{equation}
  \begin{aligned}
    S_{\text{weibull}}(x) = \left\{\begin{array}{cc}
        \int_{-W}^{x}(\frac{ck}{\lambda}(\frac{|t|}{\lambda})^{k-1} e^{-(\frac{|t|}{\lambda})^k})\d{t} & \quad\quad x \in [-W,0]\\
        S_{\text{weibull}}(-x)& \quad\quad x \in (0,W]
    \end{array}\right.
  \end{aligned}
\end{equation}

Figure~\ref{WeibullConstructor} shows, from left to right, the smooth operator, the first- and second-order TGD operators constructed on the Weibull distribution.

\begin{figure}[htb]
  \subfigure[Smooth operator]{   	 		 
        \includegraphics[width=0.31\linewidth]{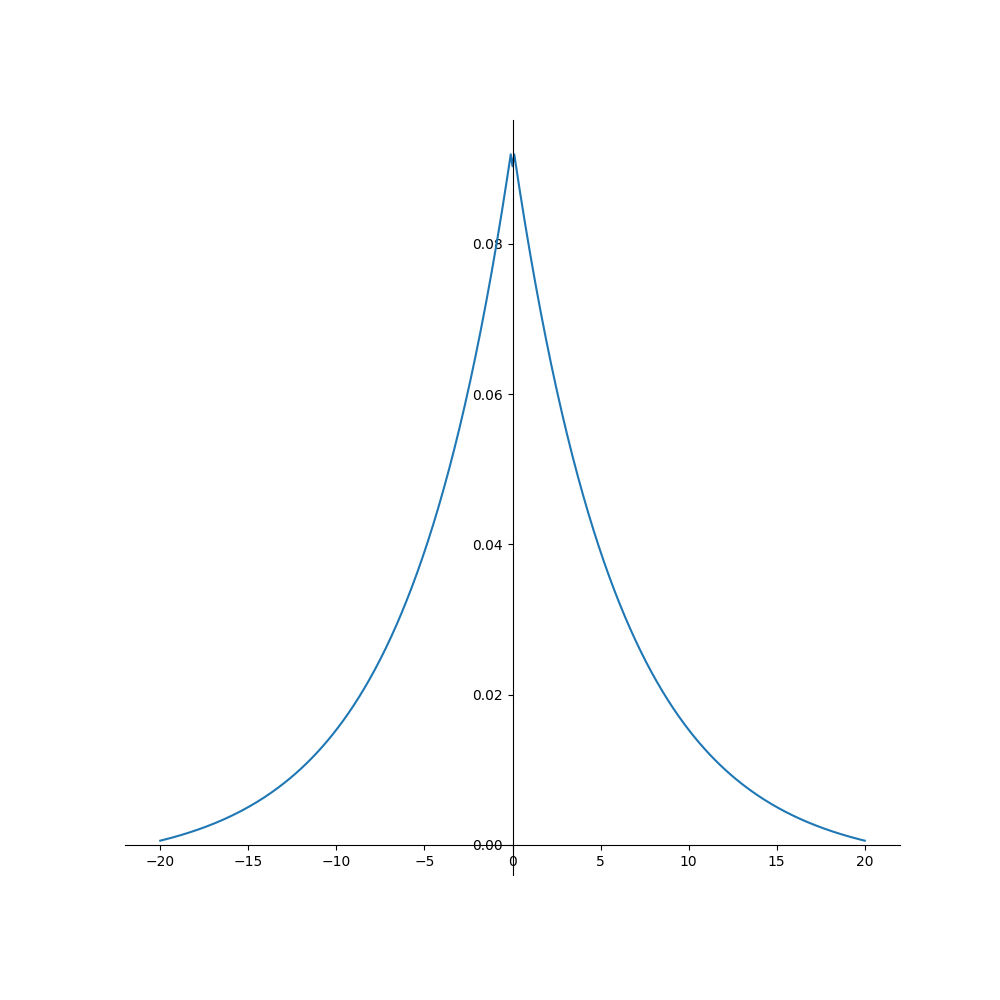}
    }
    \subfigure[First-order TGD operator]{   	 		 
        \includegraphics[width=0.31\linewidth]{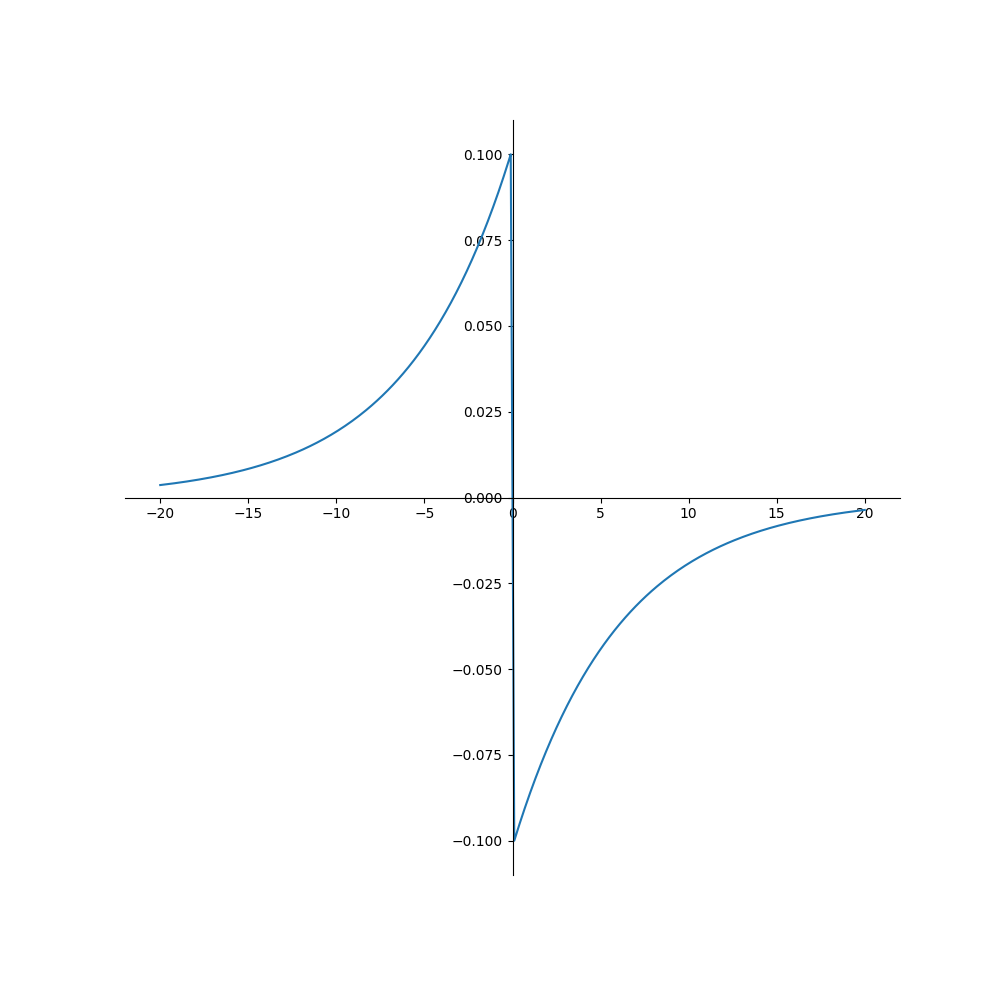}
    }   
    \subfigure[Second-order TGD operator]{   	 		 
        \includegraphics[width=0.31\linewidth]{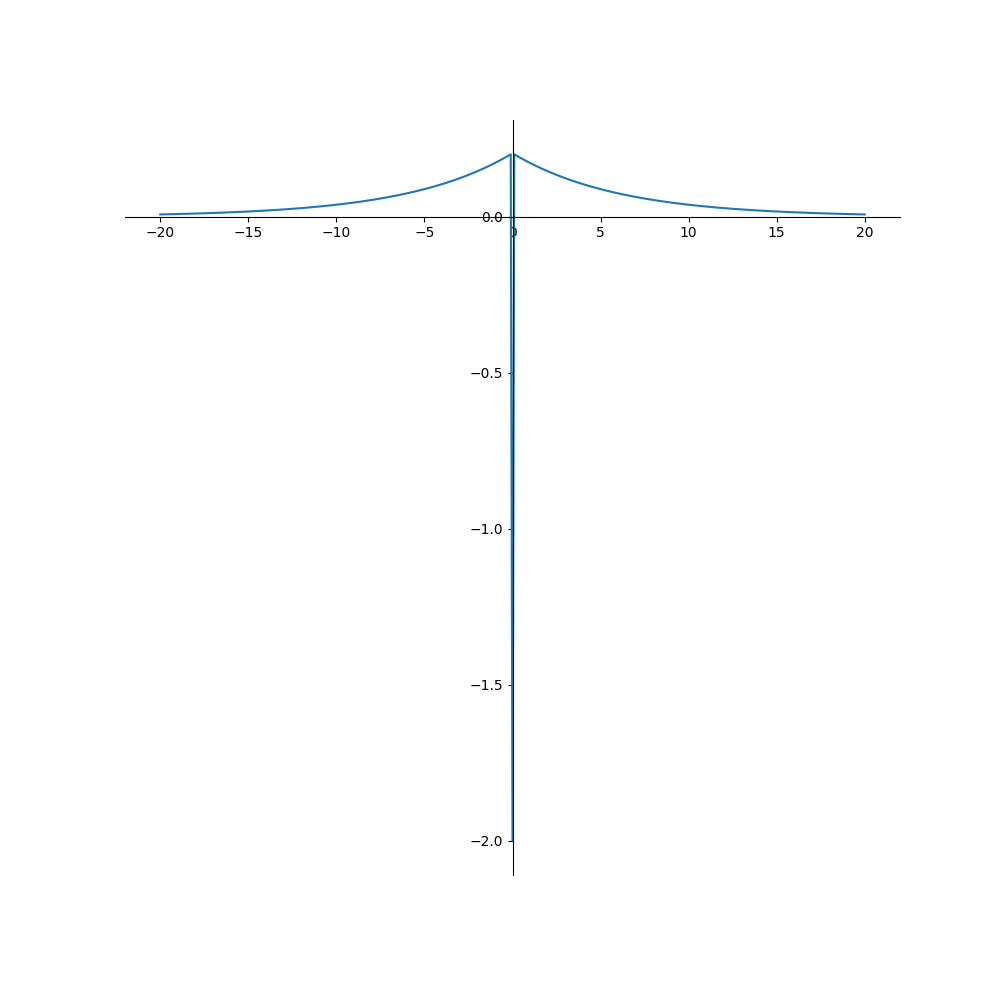}
    }
  \caption{
    The smooth operator, first- and second-order TGD operators constructed based on the Weibull distribution.
  }
  \label{WeibullConstructor}
\end{figure}


\subsubsection{Counterexamples}

Although TGD operators that satisfy the proposed constraints have demonstrated superior properties in differential calculations, we present a few counterexamples for a better understanding of the constraints.

As the first counterexample, we remove the limit symbol in Lanczos Derivative (Formula~\eqref{eq:LGD}). Lanczos employs $w(t) = 2t/W^2$ as the kernel function, which is a single-increasing function on $(0,W]$ and does not meet the Monotonic Constraint. When choosing $w(t) = 2t/W^2$ as the kernel function, Figure~\ref{Counterexamples}.a displays the corresponding smooth operator and 1st-order TGD operator (the first-order derivative of the smooth operator).

As the second counterexample, we consider Gaussian function as the smooth operator (Figure~\ref{Counterexamples}.b), which is a conventional Gaussian filter in signal processing. It is clear that the Gaussian smooth function does not satisfy the Monotonic Convexity Constraint. Meanwhile, we have emphasized that the Gaussian filter is a low-pass filter which fails to satisfy the Full-Frequency Constraint in signal processing area. This indicates that the first- and second-order derivatives of Gaussian are not suitable as a kernel function of TGD. 


\begin{figure}[htb]
\centering
  \subfigure[Remove limit symbol of Lanczos Derivative]{   	 		 
        \includegraphics[width=0.95\linewidth]{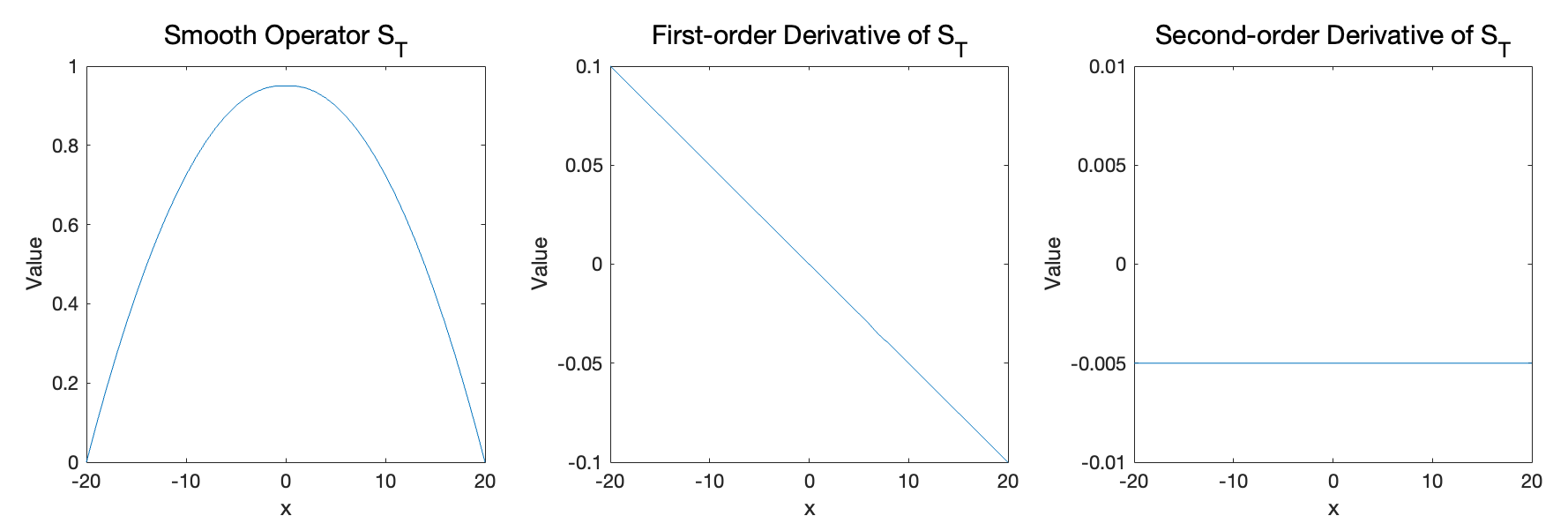}
    }
    \subfigure[Select Gaussian function as the smooth operator]{   	 		 
        \includegraphics[width=0.95\linewidth]{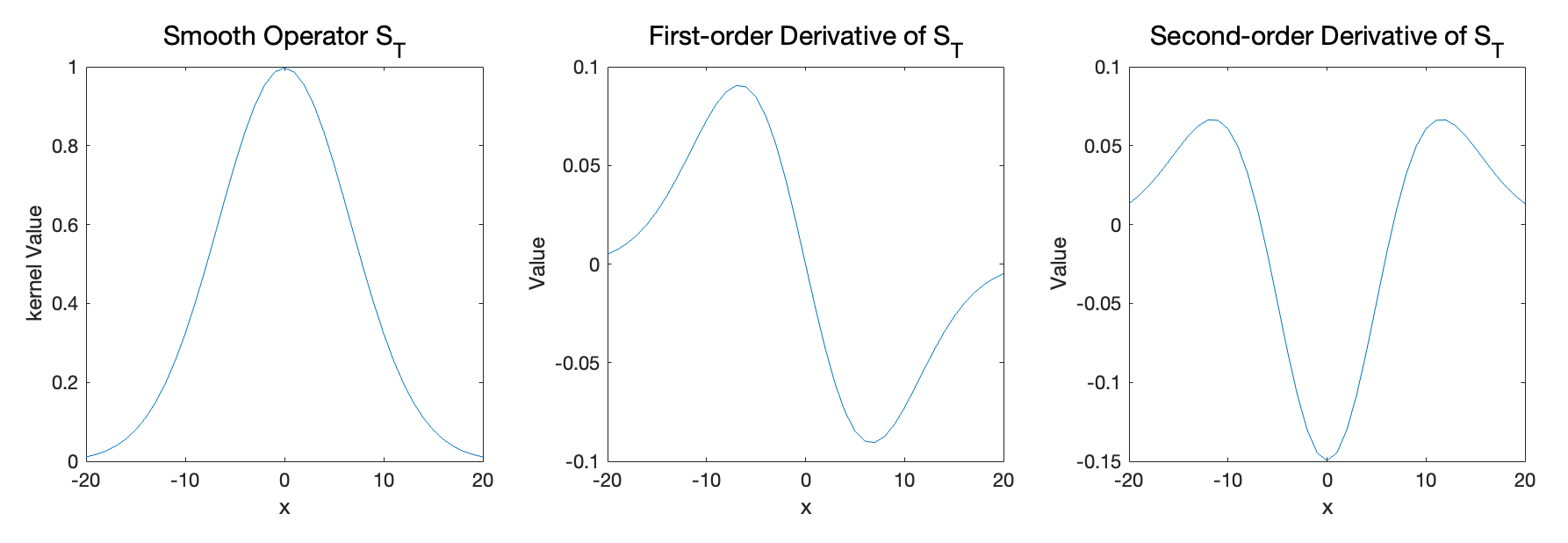}
    }   
  \caption{
    Examples of operators that do not satisfy the constraint $C2$ and/or $C3$.
  }
  \label{Counterexamples}
\end{figure}

\subsubsection{Smooth Operator Analysis}

We have introduced five kernel functions as examples and two counterexamples. We introduce the Monotonic Convexity Constraint to the smooth function in Section~\ref{subsec:1Dwindowderivativedefinitions}. Fourier transform analysis (Figure~\ref{fig:TGDFFTCompare}) reveals that the smooth operators fitting to the constraint have a full-frequency response or higher cut-off frequencies compared to the Gaussian and LD smooth operator of the same kernel width. 

\begin{figure}[htb]
    \centering
    \subfigure[The TGD smooth operators and their FFT transform]{
        \centering
        \includegraphics[width=0.8\linewidth]{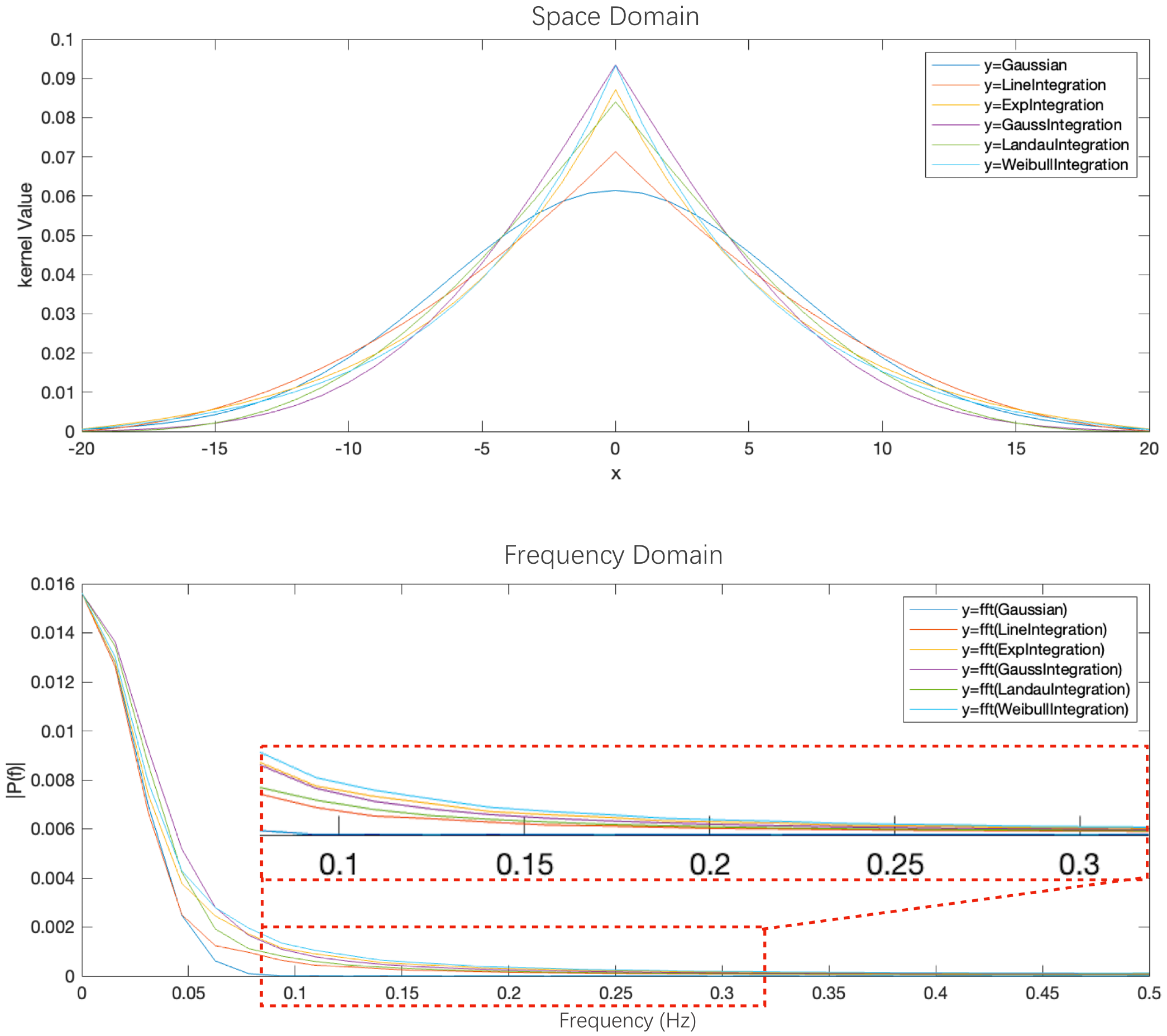}
    }
    \subfigure[The FFT transform of the counterexamples]{
        \centering
        \includegraphics[width=0.8\linewidth]{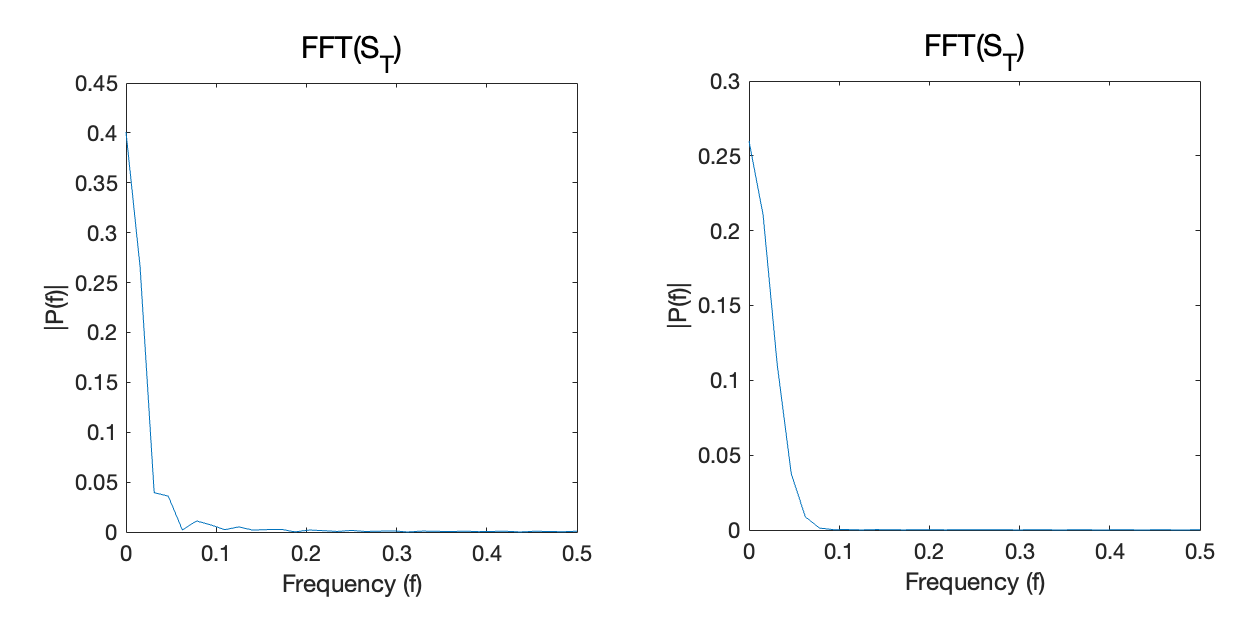}
    }
    \caption{
        Comparison of the Fourier transform of TGD smooth operator with that of the counterexamples (Left: LD smooth operator. Right: Gaussian smooth operator). The TGD smooth operators all have a higher cut-off frequency than that of counterexamples, and the response of the high-frequency part is much smaller than the low-frequency part.
    }
    \label{fig:TGDFFTCompare}
\end{figure}

\subsection{The Calculation with 1D TGD Operators}

We visualize the TGD calculations of smooth and non-smooth functions to provide an intuitive understanding of TGD. In these experiments, we utilize the Gaussian kernel function and its corresponding TGD operators.

Figure~\ref{fig:1.4.signal1}.a and~\ref{fig:1.4.signal1}.b demonstrate the traditional derivative and TGD results for smooth continuous signals, such as Sine and Gaussian functions, respectively. Despite the difference in their respective values, the derivative and TGD results show a positive relationship, which is a good embodiment of \emph{Unbiasedness}. Figure~\ref{fig:1.4.signal1}.c and~\ref{fig:1.4.signal1}.d depict the outcomes where a small amount of Gaussian noise is added to the original function, which is barely perceptible to the human eye. It is observed that the traditional finite difference method is significantly distorted by the added noise, whereas TGD, computed over a certain interval, exhibits strong noise suppression capabilities.

We further extend our analysis from smooth functions to non-smooth functions, such as the square-wave and slope signals, which commonly occur in discrete systems. Figure~\ref{fig:1.4.signal2} compares the derivative and TGD results. For the square-wave signal, the derivative values are absent or infinite at the edges of the square-wave and zero elsewhere, resulting in the visualization of two impulse signals. TGD characterizes this variation as meaningful finite values, and the extreme points of TGD align with the edges of the square-wave, where the signal changes fastest. For the slope signal, the derivatives are constant within a certain range, whereas TGD indicates the center of the ramp as the point of maximum TGD value. To test their noise resistance, We intentionally add a small amount of noise to both signals. While the traditional finite difference method is sensitive to noise, TGD demonstrates robustness against it (Figure~\ref{fig:1.4.signal2}.c and~\ref{fig:1.4.signal2}.d). 

These experiments demonstrate that TGD can effectively characterize the change rate of non-smooth functions with high stability and noise tolerance. And TGD exhibits a remarkable capability to capture signal changes.





Figure~\ref{fig:1.4.signal3} compares the outcomes calculated for a single square-wave signal and a continuous square-wave signal. Notably, when considering a single square-wave signal, the extreme points of both Extend LD and Gaussian-Smooth-Difference fail to determine the edges of the square-wave. Hence, their way of detecting signal changes is incorrect. This error worsens when working with continuous square-wave signals. the LD method severely violates Monotonic Constraint and causes significant shifts in extreme points. On the other hand, Gaussian smoothing amalgamates high-frequency signals, compromising the identification of high-frequency variations and ultimately detecting changes only at two ends of the signal. In contrast, TGD meets both the Monotonic and Full-Frequency Constraints, and accurately identifies and locates high-frequency changes. Therefore, TGD can provide a reliable representation of signal changes that constitutes a stable foundation for various gradient-based downstream tasks in signal processing.

\begin{figure}[!htb]
    \centering
    \subfigure[Sin signal]{   	 		 
        \includegraphics[width=0.95\linewidth]{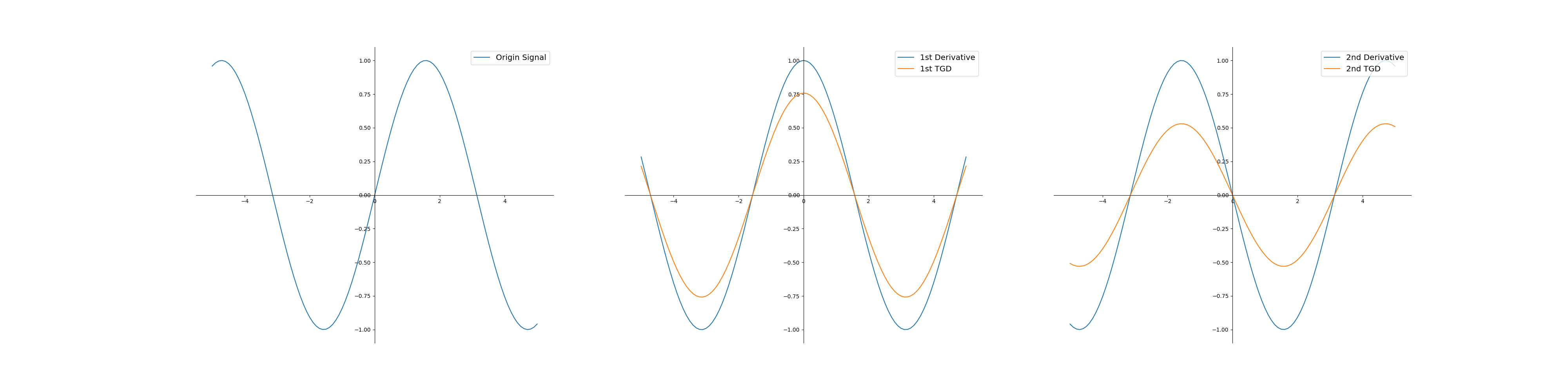}
    }
    \subfigure[Gaussian signal]{   	 		 
        \includegraphics[width=0.95\linewidth]{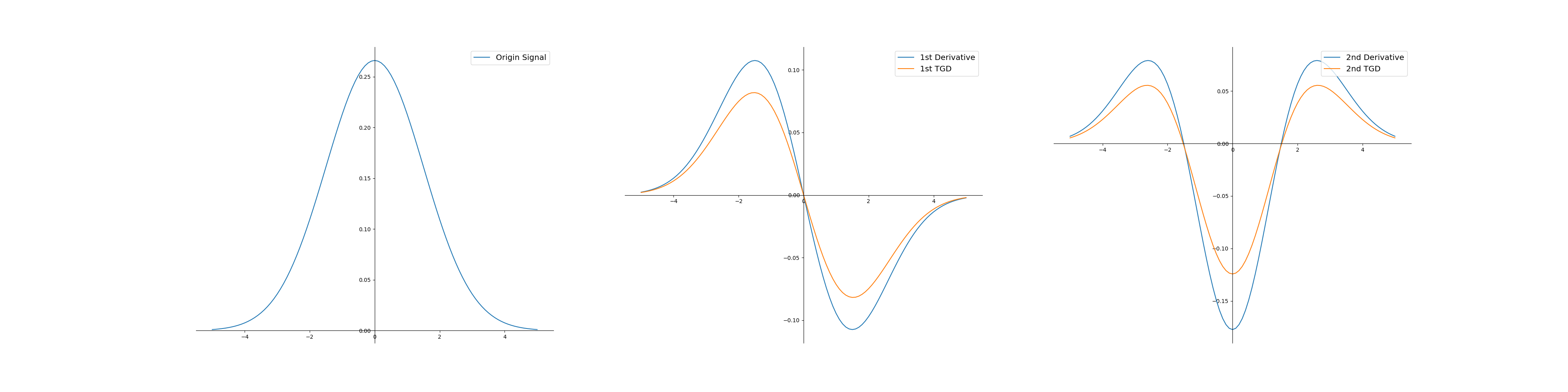}
    }   
    \subfigure[Sin signal with noise]{   	 		 
        \includegraphics[width=0.95\linewidth]{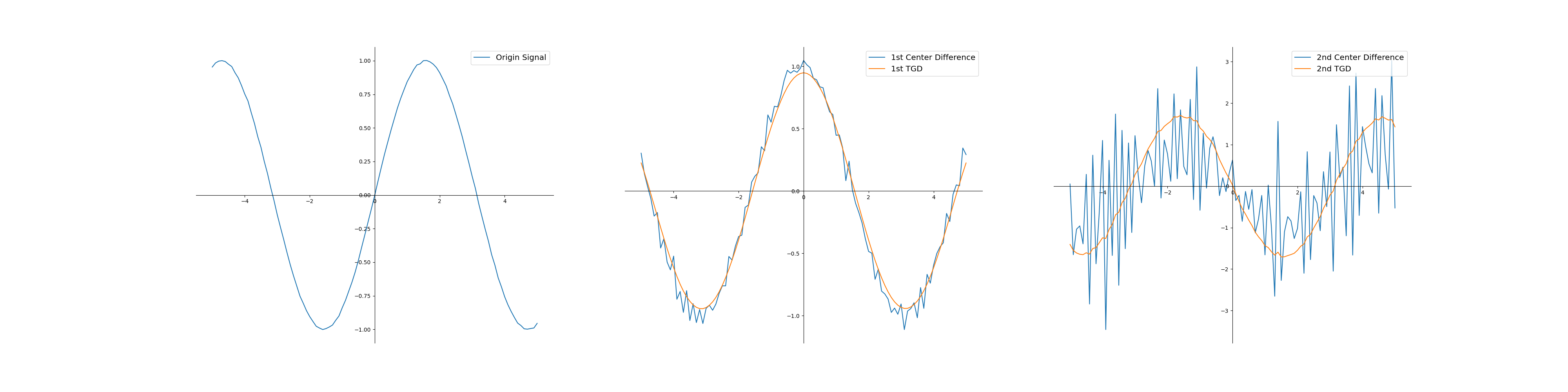}
    }
    \subfigure[Gaussian signal with noise]{   	 		 
        \includegraphics[width=0.95\linewidth]{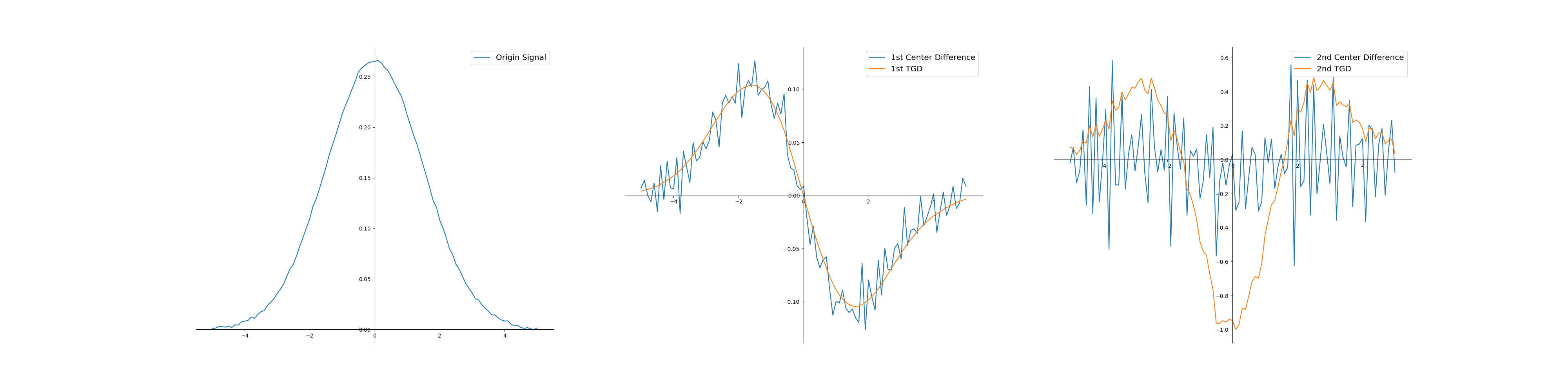}
    } 
  \caption{
    Comparison of derivative and TGD results of smooth continuous signal and noise signal. When noise is present, we approximate the derivatives using the center difference method (Formula~\eqref{eq:difference_definition}). From left to right, the origin signal, the first-order derivative (blue) and TGD (orange), the second-order derivative (blue) and TGD (orange). For better visualization, we have scaled the TGD results.
  }
  \label{fig:1.4.signal1}
\end{figure}

\begin{figure}[!htb]
    \centering
    \subfigure[Square-wave signal]{   	 		 
        \includegraphics[width=0.95\linewidth]{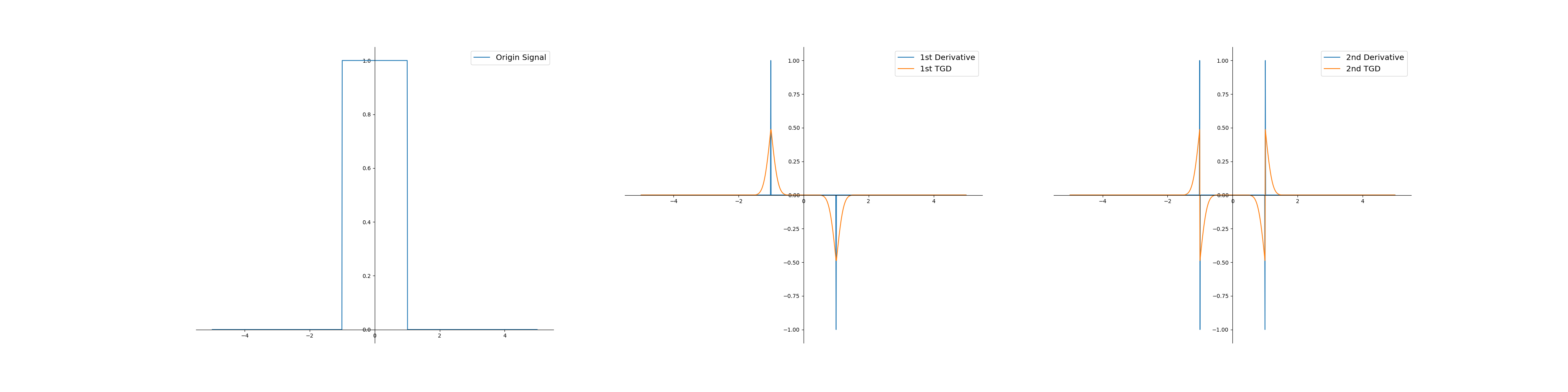}
    }
    \subfigure[Slope signal]{   	 		 
        \includegraphics[width=0.95\linewidth]{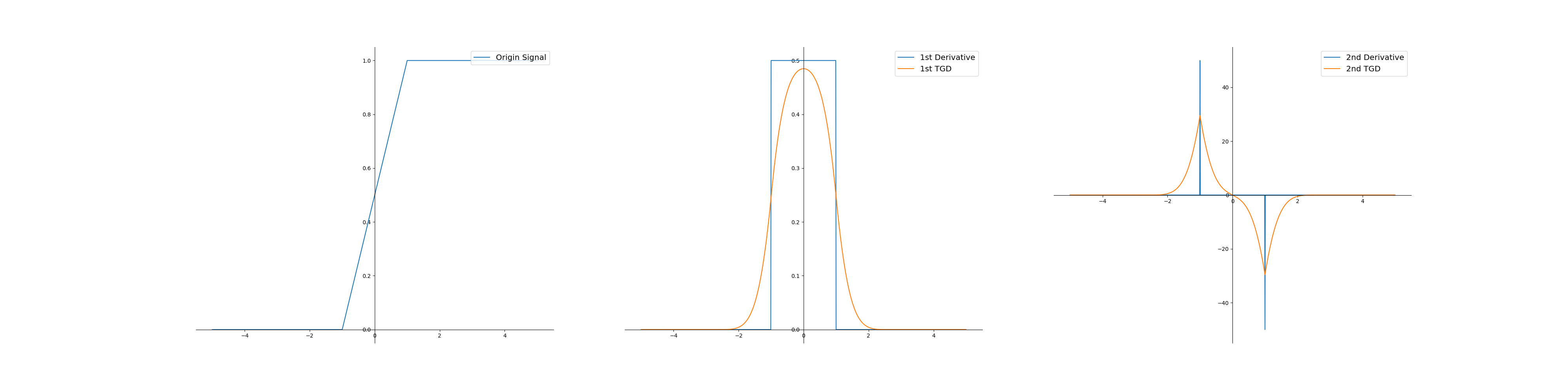}
    }   
    \subfigure[Square-wave signal with noise]{   	 		 
        \includegraphics[width=0.95\linewidth]{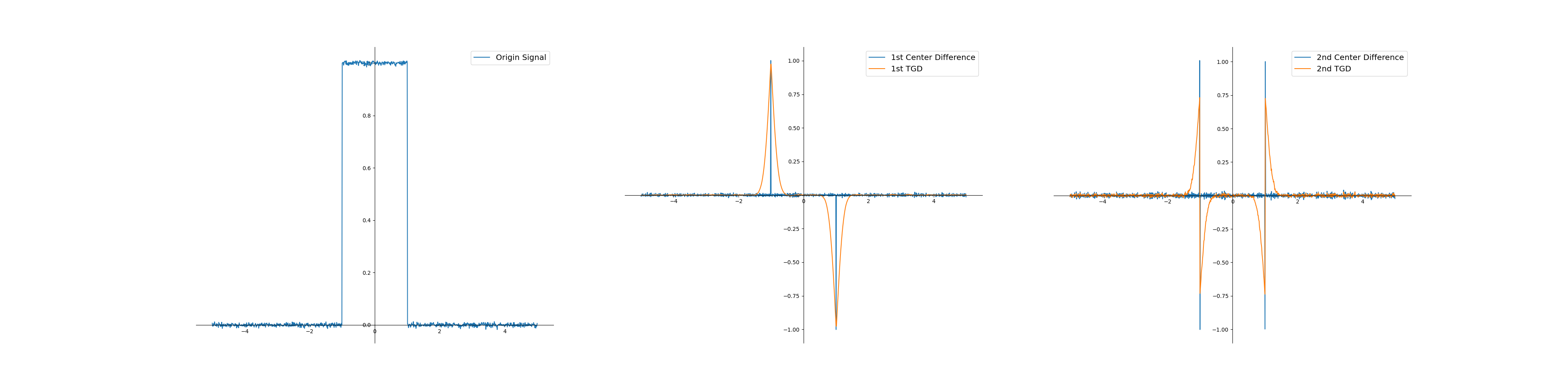}
    }
    \subfigure[Slope signal with noise]{   	 		 
        \includegraphics[width=0.95\linewidth]{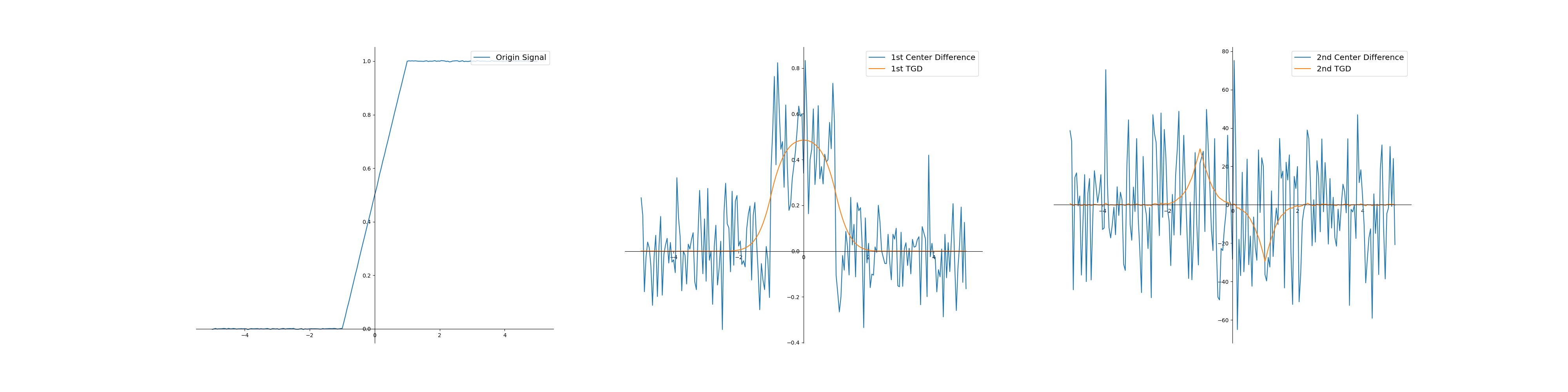}
    } 
  \caption{
    Comparison of general derivative and TGD results of non-smooth continuous signal and noise signal. When noise is present, we approximate the derivatives using the center difference method (Formula~\eqref{eq:difference_definition}). From left to right, the origin signal, the first-order derivative (blue) and TGD (orange), the second-order derivative (blue) and TGD (orange). For better visualization, we have scaled the TGD results.
  }
  \label{fig:1.4.signal2}
\end{figure}

\begin{figure}[htb]
    \centering
    \subfigure[Square-wave signal]{   	 		 
        \includegraphics[width=0.95\linewidth]{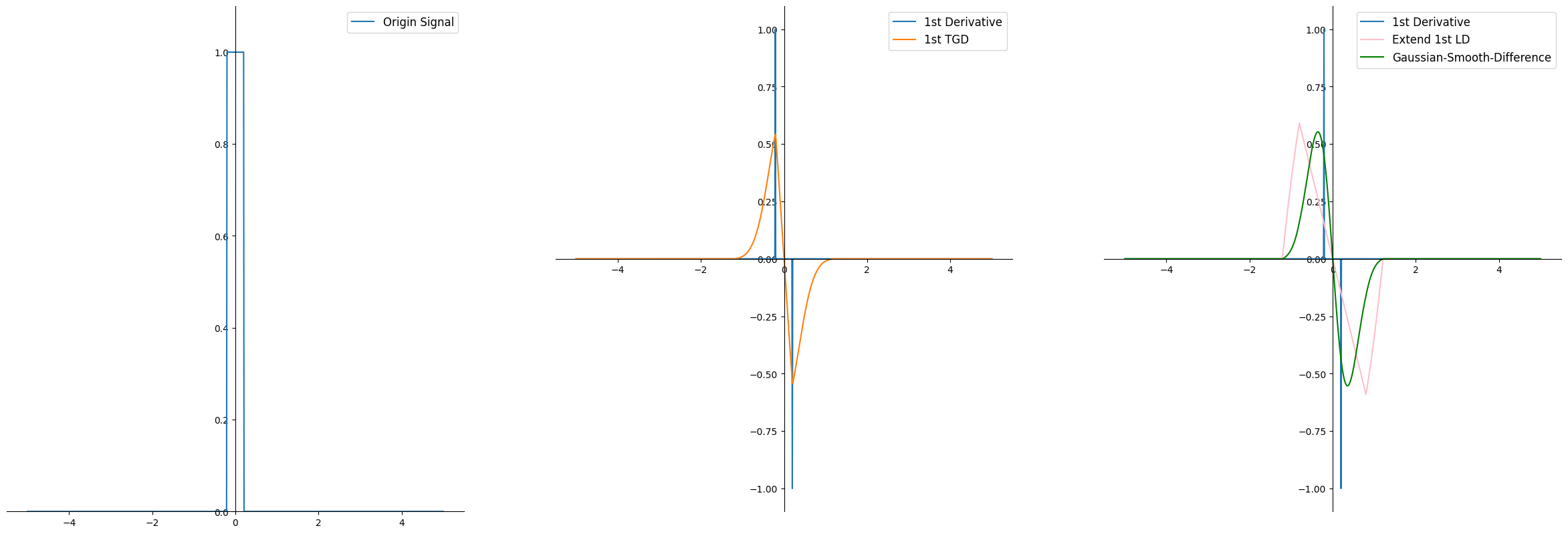}
    }
    \subfigure[Continue square-wave signal]{   	 		 
        \includegraphics[width=0.95\linewidth]{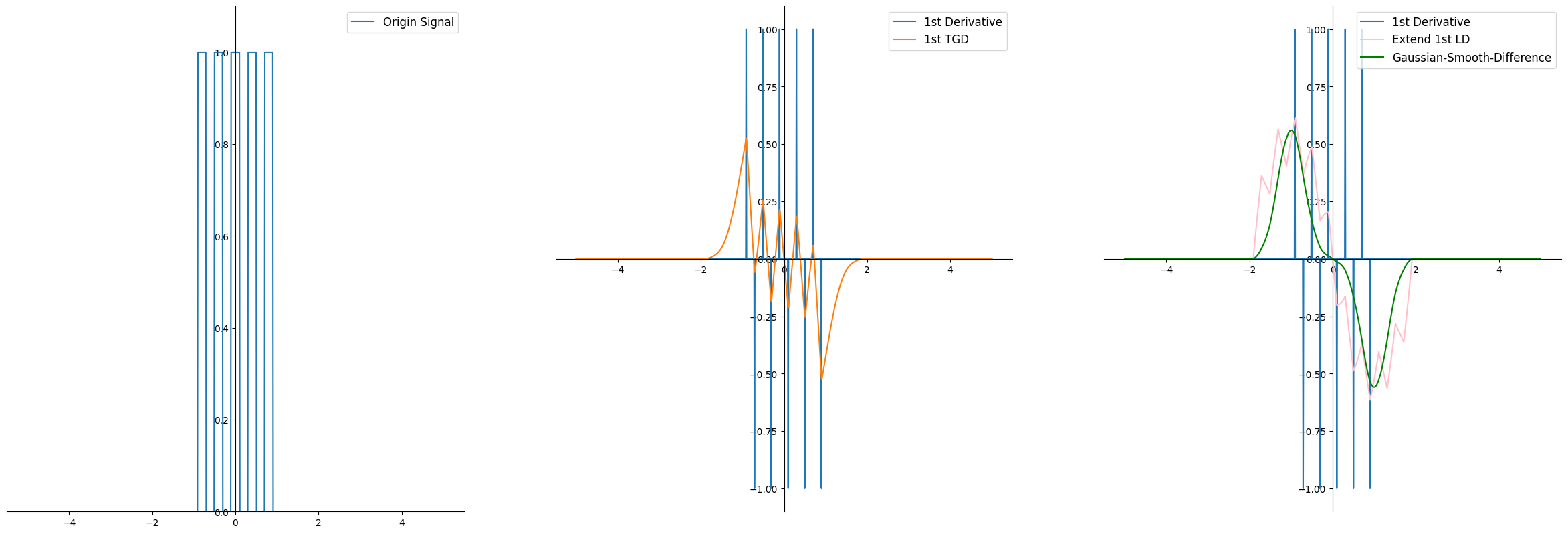}
    }
  \caption{
    Comparison of TGD, extend 1st LD, and Gaussian smoothing before difference results of square-wave signals. For better visualization, we have scaled the results.
  }
  \label{fig:1.4.signal3}
\end{figure}

\subsection{The Calculation with 2D TGD Operators}

We provide intuitive visualizations of the computational outputs of 2D TGD based on smooth and non-smooth functions, aimed at facilitating an enhanced comprehension of high-dimensional TGD. In comparison with the partial derivatives, we utilize the orthogonal TGD operators (Figure~\ref{fig:orthogonalKernel2D}). Furthermore, we present results obtained from the LoT operator. In all these experiments, TGD operators are constructed with the Gaussian kernel function.

Figure~\ref{fig:2DTGD-experiments}.a shows the 2D Gaussian signal and the computational results with respect to traditional partial derivative and partial TGD. For better visualization, we have scaled the results. Despite the difference in values of derivative and TGD, they are positively related, demonstrating \emph{Unbiasedness}. By convolving the original Gaussian signal with the LoT operator, we obtain the Laplace TGD result for the 2D Gaussian, which is known as the “Laplacian of Gaussian” in computer vision. As shown in Figure~\ref{fig:2DTGD-experiments}.b, when we add a small amount of Gaussian noise, which is hardly visible in the original function, it has a drastic effect on the traditional finite difference results. In contrast, the results calculated by TGD within a certain interval depict an exceptional ability to suppress noise.

We further extend the experiments from smooth functions to non-smooth functions. Figure~\ref{fig:2DTGD-experiments}.c displays the comparison of partial derivative and TGD results for a 2D square-wave signal. Since we calculate the partial derivatives in the $x$-direction, the derivative values are absent or infinite at the left and right edges of the square-wave and zero elsewhere. While in TGD results, there is a certain range of non-zero and meaningful values, which is caused by the fact that the operators have a certain size. Notably, the location of extreme points in the first-order TGD and the zero-crossing points in the second-order TGD remain precisely unchanged compared to the location of steps in the signal. Convolution of the square-wave signal with the LoT operator yields its Laplace TGD result, where zero-crossing points indicate exact positions of steps in the signal. These results demonstrate that TGD can precisely characterize the signal changes. Figure~\ref{fig:2DTGD-experiments}.d further highlights the noise resistance within non-smooth functions in TGD calculation.

The experiments conducted illustrate the versatility of the \emph{General Differential} in high-dimensional differential calculations of both smooth and non-smooth functions. The overall performance is stable and insensitive to noise. Additionally, TGD has proven to possess an excellent ability to characterize changes in real signals.

\begin{figure*}[htb]
    \centering
    \includegraphics[width=0.95\linewidth]{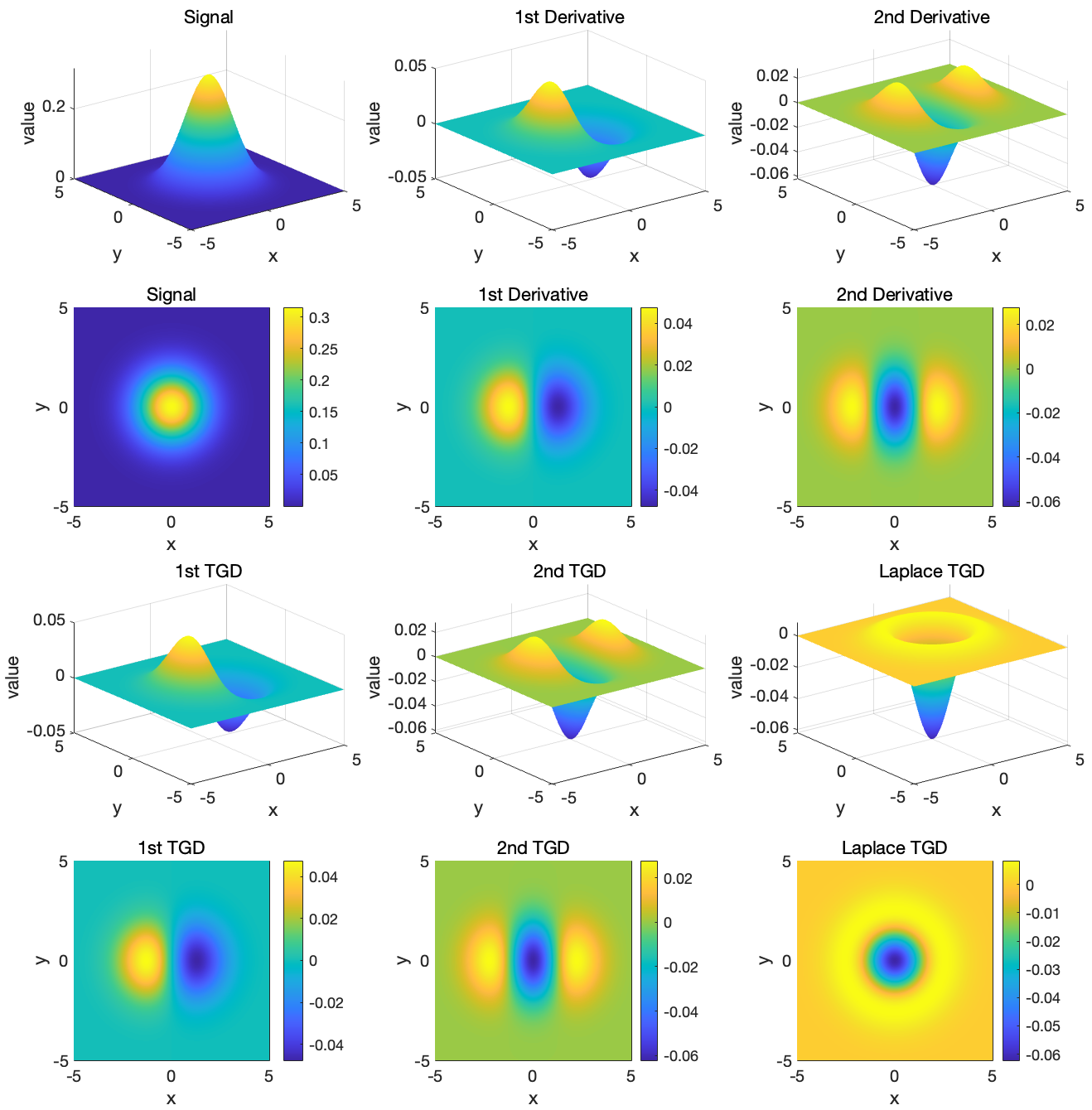}
    \caption*{(a) 2D Gaussian signal}
    \label{fig:2D gaussian signal}
\end{figure*}
\begin{figure*}[htb]
    \centering
    \includegraphics[width=0.95\linewidth]{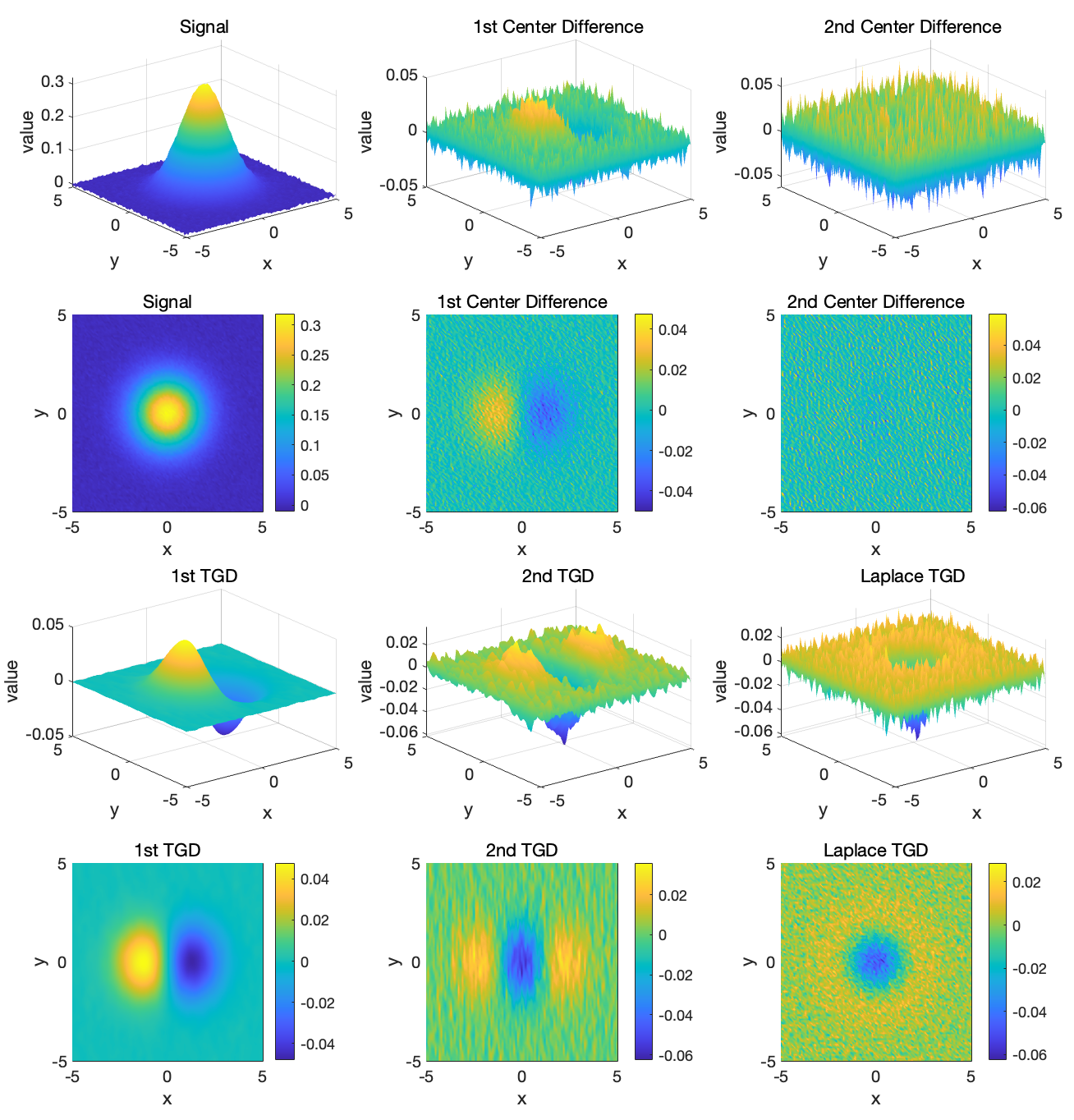}
    \caption*{(b) 2D Gaussian signal with noise}
    \label{fig:2D gaussian signal with noise}
\end{figure*}
\begin{figure*}[htb]
    \centering
    \includegraphics[width=0.95\linewidth]{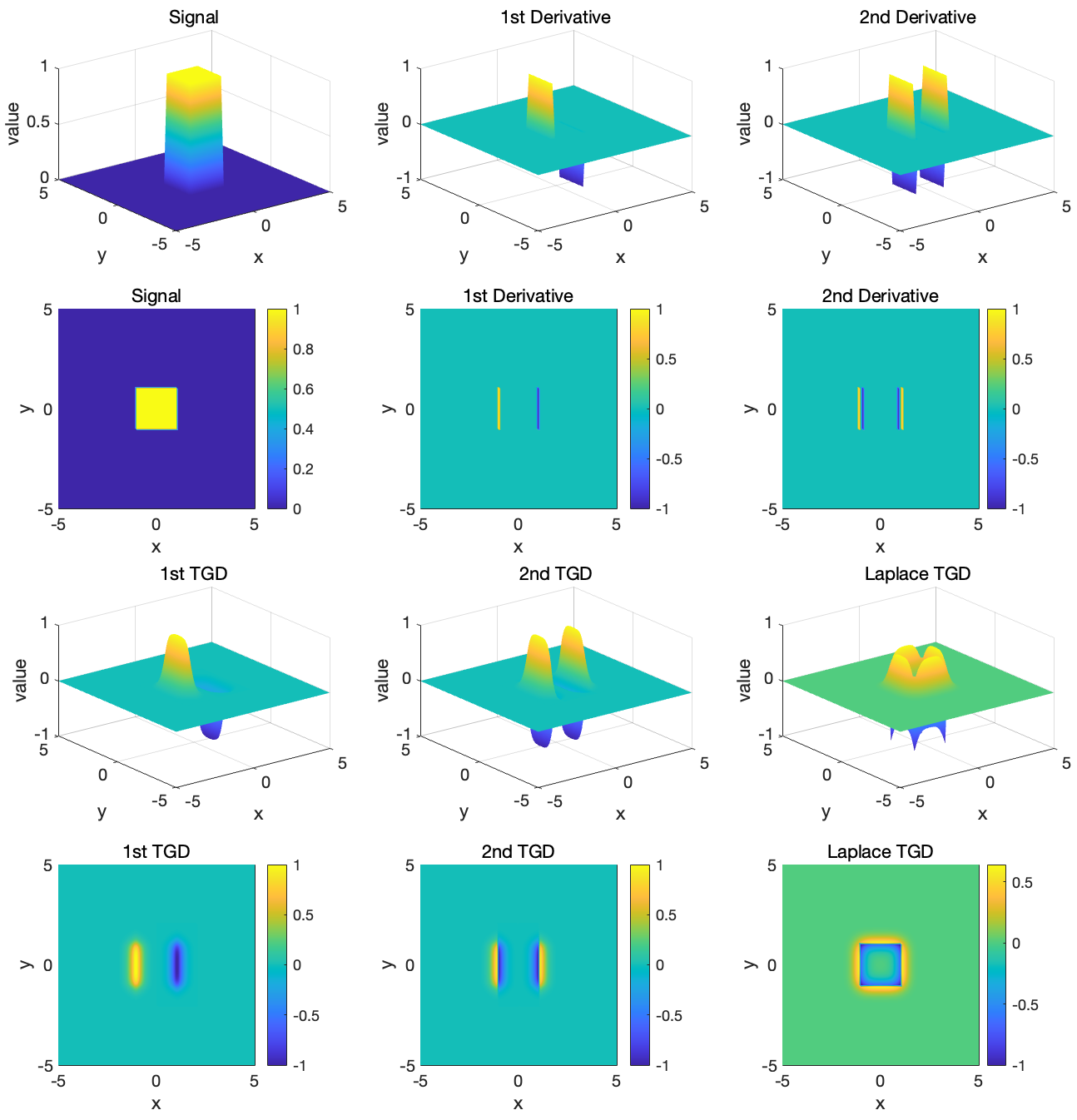}
    \caption*{(c) 2D square signal}
    \label{fig:2D square signal}
\end{figure*}
\begin{figure}[htb]
    \centering
    \includegraphics[width=0.95\linewidth]{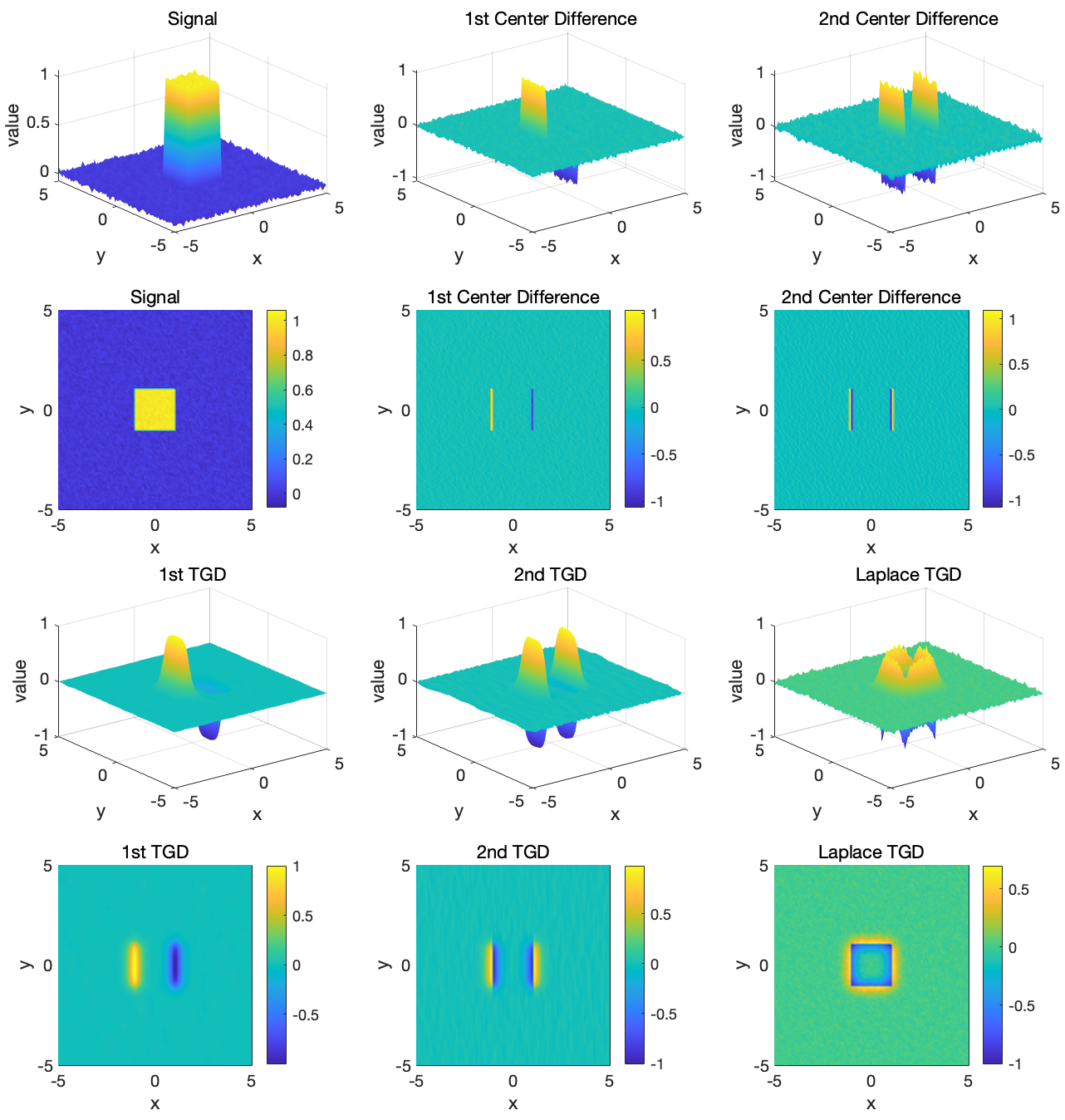}
    \caption*{(d) 2D square signal with noise}
    \caption{Comparison of partial derivative and partial TGD results of 2D smooth and non-smooth functions with/without noise. When noise is present, we approximate the derivatives using the center difference method (Formula~\eqref{eq:difference_definition}). For better visualization, we have scaled the results and show both the 3D visualization (Up) and the projection map (Down) relatively.}  
    \label{fig:2DTGD-experiments}
\end{figure}

\clearpage
\newpage

\section{Conclusion and Outlook}

Calculus is built upon the fundamental concepts of infinitesimal and limit. However, modern numerical analysis fails to conform to these basic concepts, that is, all calculations are discrete and have a finite interval. To address this issue, we begin by proposing the generalization of derivative-by-integration, and subsequently derive General Difference in a finite interval, which we call Tao General Difference. We introduce three constraints to the kernels in TGD, and launch rotational and orthogonal construction methods for constructing multiple dimensional TGD kernels. 

In discrete domain, a novel definition of the smoothness of a sequence can be defined on the first- and second-order TGD. Meanwhile, the center of gradient in an interval can be precisely located via TGD calculation in multidimensional space. In real-world, the TGD operators can be used to denoise based on the novel definition of smoothness, and to accurately capture the edges of change in the various domains. An outstanding property of the operators in real applications is the combination of robustness and precision.

From a mathematical point of view, there may be some lack of rigor in the discussion and deduction, and warrant further examination of the influence of the General Difference on calculus. For instance, does the idea of transcending the bound of infinitesimal help in the integral? We anticipate that a more rigorous discourse and proof from mathematicians would advance the study of calculus in a finite interval. 

\clearpage
\newpage

\bibliographystyle{unsrt}  
\bibliography{references}

\end{document}